%% file: main.tex
\documentclass[conference]{IEEEtran}
\makeatletter
\@ifundefined{ifblind}{\let\ifblind\iffalse}{}
\makeatother
\usepackage{cite}

\ifCLASSINFOpdf
   \usepackage[pdftex]{graphicx}
\else
\fi
\usepackage{amsmath}
\ifCLASSOPTIONcompsoc
  \usepackage[caption=false,font=normalsize,labelfont=sf,textfont=sf]{subfig}
\else
  \usepackage[caption=false,font=footnotesize]{subfig}
\fi
\usepackage{fixltx2e}
\usepackage{url}

\usepackage{stmaryrd}
\usepackage{amsfonts}
\usepackage{amssymb}
\usepackage{amsthm}
\usepackage{amssymb}
\usepackage{mathtools}
\usepackage{cmll}
\usepackage{tikz}
\usepackage{adjustbox}
\usepackage{hyperref}
\usepackage{xspace}
\usepackage[shortlabels]{enumitem}
\usepackage[nameinlink]{cleveref}
\usetikzlibrary{calc,matrix,decorations.markings,decorations.pathreplacing,arrows,cd,positioning,shapes.misc}
\usepackage{csquotes}
\usepackage{soul}
\setul{0.8pt}{0.5pt}
\input{macros}
\ifblind
\usepackage[switch]{lineno}
\else
\fi

\begin{document}
\ifblind
\linenumbers
\else
\fi
%
\title{The Cartesian Closed Bicategory\\ of Thin Spans of Groupoids}

\ifblind
\else
\author{\IEEEauthorblockN{Pierre Clairambault}
\IEEEauthorblockA{Aix Marseille Univ, Université
de Toulon, CNRS, LIS, Marseille\\
Email: Pierre.Clairambault@cnrs.fr}
\and
\IEEEauthorblockN{Simon Forest}
\IEEEauthorblockA{Aix Marseille Univ, CNRS, I2M, Marseille, France\\
Email: Simon.Forest@univ-amu.fr}}
\fi


%


\maketitle

\begin{abstract}
Recently, there has been growing interest in bicategorical models of
programming languages, which are ``proof-relevant'' in the sense that
they keep distinct account of execution traces leading to the same
observable outcomes, while assigning a formal meaning to reduction paths
as isomorphisms.

In this paper we introduce a new model, a bicategory called \emph{thin
spans of groupoids}. Conceptually it is close to Fiore et al.'s
\emph{generalized species of structures} and to Melliès' \emph{homotopy
template games}, but fundamentally differs as to how
replication of resources and the resulting symmetries are treated.
Where those models are \emph{saturated} -- the interpretation is
inflated by the fact that semantic individuals may carry arbitrary
symmetries -- our model is \emph{thin}, drawing inspiration 
from \emph{thin concurrent games}: the interpretation of
terms carries no symmetries, but semantic individuals satisfy a subtle
invariant defined via biorthogonality, which guarantees their
invariance under symmetry. 

We first build the bicategory $\Thin$ of thin spans of groupoids. Its
objects are certain groupoids with additional structure, its morphisms
are spans composed via plain pullback with identities the identity
spans, and its $2$-cells are span morphisms making the induced
triangles commute only up to natural isomorphism. We then equip $\Thin$
with a pseudocomonad $\oc$, and finally show that the Kleisli
bicategory $\Thin_{\oc}$ is cartesian closed. 
\end{abstract}


%
\IEEEpeerreviewmaketitle

\section{Introduction}


The \emph{relational model} \cite{DBLP:journals/tcs/Girard87} is one of
the most basic and elementary denotational models for linear logic or
the $\lambda$-calculus. At its heart, it is simply an interpretation of
formulas / types as \emph{sets} and proofs / programs as \emph{relations},
\emph{i.e.} in the category $\Rel$. Despite its simplicity the
relational model is ubiquitous: it is the basic substrate for the
spectrum of so-called \emph{web-based} models of linear logic, including
coherence or finiteness spaces \cite{DBLP:journals/mscs/Ehrhard05}. It
faithfully predicts reduction time
\cite{DBLP:journals/tcs/CarvalhoPF11}. It supports quantitative
extensions such as in probabilistic coherence spaces
\cite{DBLP:journals/iandc/DanosE11}, the weighted relational model
\cite{DBLP:conf/lics/LairdMMP13}, and even up to quantum computation
\cite{DBLP:conf/popl/PaganiSV14} -- quantitative extensions which enjoy
powerful full abstraction results
\cite{DBLP:journals/jacm/EhrhardPT18,DBLP:journals/pacmpl/ClairambaultV20}.
Presented syntactically, the relational model exactly corresponds to
\emph{non-idempotent intersection types}
\cite{DBLP:journals/corr/abs-0905-4251}, a currently active research
topic in its own right (see \emph{e.g.}
\cite{DBLP:journals/jfp/AccattoliGK20,DBLP:journals/lmcs/BucciarelliKR18})
which enables a syntactic methodology to addressing semantic questions.
Finally, it has a tight connection with \emph{game semantics}
\cite{DBLP:journals/iandc/AbramskyJM00,DBLP:journals/iandc/HylandO00},
of which it appears as a desequentialization (see \emph{e.g.}
\cite{DBLP:conf/csl/BaillotDER97,DBLP:journals/tcs/Mellies06,DBLP:conf/tlca/Boudes09,DBLP:journals/pacmpl/ClairambaultV20}).
In short, it is at the crossroads of multiple topics, past and current,
of the denotational semantics universe. 

Another recent trend in denotational semantics is the adoption of
\emph{bicategorical models} \cite{leinster2004higher} where the
familiar categorical laws hold only up to certain \emph{$2$-cells}
satisfying coherence conditions -- in particular, Fiore and
Saville have recently thoroughly explored \emph{cartesian closed
bicategories} \cite{DBLP:journals/mscs/FioreS21}. In such models, the
denotation is no longer an invariant of reduction: two convertible
terms yield merely \emph{isomorphic} objects, and \emph{reduction
paths} have a genuine interpretation as specific \emph{isomorphisms}
\cite{DBLP:conf/lics/FioreS19} --
thus bringing reduction into the categorical model. There are still
not many concrete bicategorical models, and we are aware of only three
(families of) such models that can deal with non-linear computation, in
chronological order: firstly, Fiore, Gambino, Hyland and Winskel's
cartesian closed bicategory of \emph{generalized species of structure}
\cite{fiore2008cartesian}; secondly, Castellan, Clairambault and
Winskel's \emph{thin concurrent games}
\cite{DBLP:journals/lmcs/CastellanCW19} (as established by Paquet
\cite{paquet2020probabilistic}); thirdly, Melliès'
\emph{homotopy template games} \cite{DBLP:conf/lics/Mellies19}. Of
these three, the first is by far the most studied with various works
including generalizations and application to semantics
\cite{DBLP:conf/lics/TsukadaAO18,DBLP:conf/fscd/Galal20,DBLP:conf/fscd/Galal21},
links with intersection types and Taylor expansion
\cite{DBLP:conf/lics/TsukadaAO17,DBLP:conf/lics/Olimpieri21}, or
applications to the pure $\lambda$-calculus \cite{relevant23}. Beyond
giving a non-degenerated interpretation to reduction paths, those
concrete bicategorical models are ``proof-relevant'', in the sense that
they keep distinct semantic witnesses for the possibly multiple
evaluation traces with the same observable behaviour and thus keep a
clear, branching account of non-determinism.

These models have something else in common: in their construction, the
main subtlety has to do with replication, \emph{i.e.} the
modality $\oc$ of linear logic. In the relational model, $\oc A$ is
the set $\Mf(A)$ of finite multisets of elements of $A$, or
alternatively, the free monoid $A^*$ quotiented by permutations. 
In bicategorical models, this is replaced by a 
\emph{categorification} of $\Mf(A)$: a category (or groupoid) whose
objects keep separate individual resource usages (\emph{e.g.}
$A^*$). Its morphisms are explicit permutations, often called
\emph{symmetries} in this paper. Individuals in the model must refer to
specific resources (\emph{e.g.} $a_i$ in $a_1 \dots a_n \in A^*$), but
the categorical laws expected for models of programming languages
requires that their behaviour should still be invariant under symmetry.
In both generalized species of structure and template games, this is
done by \emph{saturating} the set of witnesses with respect to
symmetries: intuitively, the behaviour of an individual cannot depend
on the specific identity of resources, because those resources are seen
through the ``noise'' of all possible symmetries -- this shall be
reviewed gently in \Cref{sec:longintro}. This saturation
complicates models and their construction, though for good
reasons.  But this contrasts with \emph{thin concurrent
games}, which handles symmetry with a mechanism inspired by
Abramsky-Jagadeesan-Malacaria games
\cite{DBLP:journals/iandc/AbramskyJM00} and Melliès' \emph{orbital
games} \cite{DBLP:journals/tcs/Mellies06}: strategies are not
saturated, but their invariance under ``Opponent's
symmetries'' is ensured by a subtle bisimulation-like structure -- we
call this the \emph{thin} approach.

We believe that the thin approach is helpful at least for applications
to semantics: the absence of symmetries on witnesses allow a more concrete
flavour which may help when ordering individuals allowing continuous
reasoning\footnote{For instance, in \cite{relevant23}, the
generalization from finite to infinite computation is not simply by
continuity as per usual in denotational semantics, because of the
quotient involved in the management of saturation.}, or
simplify quantitative extensions such as
\cite{DBLP:conf/lics/TsukadaAO18}. But more fundamentally, there is a
clear tension between these two worlds that deserves investigation. Are
proof-relevant relational models inherently saturated? Is the thin
approach only possible in games thanks to the presence of time and
causality? These fundamental questions may be of interest beyond denotational
semantics, as the handling of symmetry in such models is
deeply connected to algebraic combinatorics \cite{fiore2008cartesian}
and homotopy theory \cite{DBLP:conf/lics/Mellies19}.

\paragraph{Contributions} We introduce the bicategory $\Thin$ of \emph{thin
spans of groupoids}: its objects are certain groupoids with additional
structure, its morphisms certain spans, and its $2$-cells certain
\emph{weak} span morphisms, \emph{i.e.} making the induced triangles
commute up to chosen natural isos. Identities are
identity spans, and composition of spans is by plain pullback.  

Of course, plain pullbacks are too weak to support the horizontal
composition of weak span morphisms. To address this, we first define
\emph{uniform spans} via a biorthogonality construction, ensuring that
the composition pullbacks also satisfy the \emph{bipullback} universal
property. This allows us to compose $2$-cells horizontally, but that
horizontal composition is still not canonically defined and fails to
give a bicategory. 

For the next step, we import from thin concurrent games and from
Melliès' orbital games a decomposition of symmetries into
\emph{positive} symmetries (due to the program), and \emph{negative}
symmetries (due to the environment). We then define \emph{thin spans}
via a second biorthogonality construction, which ensures that the
horizontal composition of weak span morphisms are canonically defined
as long as we consider \emph{positive} weak span morphisms, where the
chosen iso only involves positive symmetries. We show this results in a
bicategory $\Thin$. Furthermore, we equip $\Thin$ with a pseudocomonad
$\oc$, and show that the Kleisli bicategory $\Thin_{\oc}$ is cartesian
closed.

\paragraph{Outline} In \Cref{sec:longintro} we start with a
gentle introduction to the relational model and its proof-relevant
extensions. In \Cref{sec:bicat_thin} we introduce the bicategory
$\Thin$, deploying first the uniform orthogonality and then the thin
orthogonality. In \Cref{sec:cc_bicat} we introduce the
pseudocomonad $\oc$, and show that the Kleisli bicategory $\Thin_{\oc}$
is cartesian closed.

\section{Relational Models, Spans, Species}
\label{sec:longintro}

\subsection{The Relational Model}

The \emph{relational model} is one of the simplest denotational models
of the $\lambda$-calculus, linear logic, or simple programming
languages such as $\PCF$. It consists in simply interpreting every type
$A$ as a set $\intr{A}$, and a program $\vdash M : A$ as a subset of
$\intr{A}$. This set $\intr{A}$ is often called the \emph{web}
seeing that it is the first component of the so-called web-based models
of linear logic such as coherence spaces and their extensions. One may
think of elements of $\intr{A}$ as completed executions (which is
straightforward enough for ground types such as booleans or natural
numbers but may be more complex for higher-order types), and of
$\intr{M} \subseteq \intr{A}$ as simply the collection of all the
completed executions that $M$ may achieve.

\begin{example}
The ground type for booleans is interpreted as $\intr{\tbool} =
\{\ttrue, \tfalse\}$, and the constant $\vdash \ttrue : \tbool$ as
$\intr{\ttrue} = \{\ttrue\}$. 
\end{example}

The interpretation of a program $M$ is computed compositionally,
following the methodology of denotational semantics, organized
by the categorical structure of sets and relations.

\subsubsection{Basic categorical structure} 
There is a category $\Rel$ with sets as objects, and as morphisms 
the \emph{relations} from $A$ to $B$, \emph{i.e.} subsets $R
\subseteq A \times B$. The identity on $A$ is the diagonal relation
$\{(a, a) \mid a \in A\} \subseteq A \times A$, and the composition of
$R \subseteq A \times B$ and $S \subseteq B \times C$ consists in all
pairs $(a, c) \in A\times B$ such that $(a, b) \in R$ and $(b, c) \in
S$ for some $b \in B$.

Besides, $\Rel$ has a monoidal structure given by the cartesian product
on objects, and for $R_i \in \Rel(A_i, B_i)$, $R_1 \times R_2 \in
\Rel(A_1 \times A_2, B_1 \times B_2)$ set as comprising all $((a_1, a_2),
(b_1, b_2))$ when $(a_i, b_i) \in R_i$ -- the unit $I$ is a fixed
singleton set, say $\{*\}$. Additionally, $\Rel$ is \emph{compact
closed}: each set $A$ has a dual $A^*$ defined simply as $A$ itself,
and there are relations $\eta_A \in \Rel(I, A\times A)$ and $\epsilon_A
\in \Rel(A\times A, I)$, both diagonal relations, satisfying coherence
conditions \cite{kelly1980coherence}. 
In particular, $\Rel$ is \emph{$\star$-autonomous} and as such a model
of multiplicative linear logic, and the linear $\lambda$-calculus: the
linear arrow type is interpreted as $\intr{A \lin B} =
\intr{A} \times \intr{B}$. Finally, $\Rel$ has finite products, with the
binary product of sets $A$ and $B$ given by the disjoint union
$A + B = \{1\} \times A \uplus \{2\} \times B$.

\subsubsection{The exponential modality} The exponential modality of
$\Rel$ is based on \emph{finite multisets}. If $A$ is a set, we write
$\Mf(A)$ for the set of finite multisets on $A$. To denote specific
multisets we use a list-like notation, as in \emph{e.g.} $[0, 1, 1] \in
\Mf(\mathbb{N})$ -- we write $[] \in \Mf(A)$ for the empty multiset. 

For $A$ a set, its \textbf{bang} $\oc A$ is simply the set $\Mf(A)$.
This extends to a comonad on $\Rel$, satisfying the required
conditions to form a so-called \textbf{Seely category} -- in particular,
there is
\[
\Mf(A + B) \iso \Mf(A) \times \Mf(B)
\]
a bijection providing the \emph{Seely isomorphism}. Altogether, this
makes $\Rel$ a model of intuitionistic linear logic; and this makes the
Kleisli category $\Rel_{\oc}$ cartesian closed so that we may interpret
(among others) the simply-typed $\lambda$-calculus.

\begin{example}\label{ex:prelim_term}
Considering the term $\vdash M : \tbool \to \tbool$ of $\PCF$
\[
\begin{array}{l}
\vdash \lambda
x^\tbool.\,\mathtt{if}\,x\,\mathtt{then}\,x\\
\hspace{50pt}\mathtt{else}\,\mathtt{if}\,x\,\mathtt{then}\,\tfalse\,\mathtt{else}\,\ttrue : \tbool
\to \tbool\,,
\end{array}
\]
we have $\intr{M} = \{([\ttrue, \ttrue], \ttrue), ([\ttrue,
\tfalse], \tfalse), ([\tfalse, \tfalse], \ttrue)\}$.

Here we can observe that the model is \emph{quantitative}, in that it
records how many resources each execution consumes: one may observe
output $\ttrue$ either with two evaluations of $x$ to $\ttrue$, or with
two evaluations of $x$ to $\tfalse$. One may observe output $\tfalse$
with two evaluations of $x$, one to $\ttrue$ and one to $\tfalse$.
Recall that in $[\ttrue, \tfalse] = [\tfalse, \ttrue]$, the order is irrelevant.
\end{example}

The relational model also supports the interpretation of 
non-determinism: if $\vdash \choice : \tbool$ is a new
primitive evaluating non-deterministically to $\ttrue$ or $\tfalse$,
then we may simply set
\[
\intr{\choice} = \{\ttrue, \tfalse\}\,.
\]

\subsubsection{Extensions of the relational model}  The relational
model is extremely flexible, and can be extended in multiple different
ways. In one direction one may add to the objects a \emph{coherence
relation} and restrict to compatible morphisms -- we obtain in this way
(multiset-based) \emph{coherence semantics}.

Another extension is the \emph{weighted relational model}
\cite{DBLP:journals/tcs/Lamarche92,DBLP:conf/lics/LairdMMP13} where
a term $\vdash M : A$, instead of denoting a subset of $\intr{A}$ --
\emph{i.e.} a function $\intr{M} : \intr{A} \to \{0, 1\}$ -- denotes a
function 
\[
\intr{M} : \intr{A} \to \R
\]
assigning to each point of the web $a \in \intr{A}$ a \emph{weight}
$\intr{M}_a \in \R$. The weight may be used to record additional
information about executions. One may record the number of distinct
non-deterministic branches leading to a certain result: for instance,
if $\R = \mathbb{N} \cup \{+\infty\}$, then
$\intr{\mathtt{if}\,\choice\,\mathtt{then}\,\ttrue\,\mathtt{else}\,\ttrue}_{\ttrue}
= 2$.
With $\R = \rp = \mathbb{R}_+ \cup \{+\infty\}$, we may track the
\emph{probability} with which a certain result occurs, obtaining a
model fully abstract for probabilistic $\PCF$
\cite{DBLP:journals/jacm/EhrhardPT18}. The paper
\cite{DBLP:conf/lics/LairdMMP13} contains other examples:
resource consumption, must convergence, \emph{etc}.

It is natural to go one step further and make the relational model
``proof-relevant''. This means not merely recording a weight or
counting non-deterministic branches, but keeping track of a set
$\intr{M}_a \in \Set$
of \emph{witnesses}
of the execution of $M$ to $a$, for each $\vdash M : A$ and $a \in
\intr{A}$. There are well-documented ways to do that which we shall
review later on, but for now let us attempt this naively.

\subsection{The Bicategory of Spans} 

A first idea is to simply replace \emph{relations} with \emph{spans}.

\subsubsection{Spans}
Recall that if $\C$ is a category with pullbacks, then we form
$\Span(\C)$ has having as objects those of $\C$, and as morphisms from
$A$ to $B$ triples $(S, \display^S_A, \display^S_B)$ forming a diagram
\[
A 
\quad \stackrel{\display^S_A}{\leftarrow} \quad
S
\quad \stackrel{\display^S_B}{\rightarrow} \quad
B\,,
\]
where intuitively $S$ is a set of \emph{internal witnesses}, projected
to $A$ and $B$ via the maps $\display^S_A$ and $\display^S_B$. For $\C
= \Set$ one obtains a relation by collecting the pairs
$(\display^S_A(s), \display^S_B(s))$ for $s \in S$, but we have more:
for each pair $(a, b) \in A \times B$ we have
\[
\wit^S(a, b) = \{s \in S \mid \display^S_A(s) = a~\&~\display^S_B(s)
= b\}\,,
\]
a set of \textbf{witnesses} that $a$ and $b$ are related -- hence this
indeed provides a notion of a \emph{proof-relevant} relational model.

\begin{example}
Writing $\mathbb{B} = \{\ttrue, \tfalse\}$ and $1 = \{*\}$, we may
represent the program $\vdash
\mathtt{if}\,\choice\,\mathtt{then}\,\ttrue\,\mathtt{else}\,\ttrue$ as
\[
1
\quad \stackrel{\partial_l}{\leftarrow} \quad
\{a, b\}
\quad \stackrel{\partial_r}{\rightarrow} \quad
\tbool
\]
a \emph{span}, where $\partial_l(a) = \partial_l(b) = *$, $\partial_r(a) =
\partial_r(b) = \ttrue$. 

Thus, the evaluation of the program to $\ttrue$ \emph{has two witnesses}.
\end{example}

\subsubsection{A bicategory} \label{subsubsec:bicategory_prelim}
The exact identity of $S$ does not matter
-- the same span above with $S' = \{a', b'\}$ should not be treated
distinctly. A \textbf{morphism} between spans is $f\co S \to S'$ making
\[
  \begin{tikzcd}[cramped,rsbo=1.3em]
    & S	\ar[dl,"\display^S_A"']
    \ar[dr,"\display^S_B"]
    \ar[dd,"f"]
    &
    \\
    A&&B\\
    &S'	\ar[ul,"\display^{S'}_A"]
    \ar[ur,"\display^{S'}_B"']
  \end{tikzcd}
\]
commute; an \textbf{isomorphism} of span is an invertible morphism.

The \textbf{identity span} on $A$ is simply $A \leftarrow A \rightarrow
A$ with two identity maps. 
The \textbf{composition} of $A \leftarrow S \rightarrow B$ and $B \leftarrow T
\rightarrow C$ is obtained by first forming the pullback
\begingroup
\makeatletter
\renewcommand{\maketag@@@}[1]{\hbox to 0.000008pt{\hss\m@th\normalsize\normalfont#1}}%
\makeatother
\begin{equation}
  \begin{tikzcd}[cramped,csbo=large,rsbo=normal]
    &&
    T \odot S
    \ar[dl,"\pl"']
    \ar[dr,"\pr"]
    \phar[dd,"\dcorner",very near start]
    &&
    \\
    &
    S	\ar[dl,"\display^S_A"']
    \ar[dr,"\display^S_B"]
    &&
    T
    \ar[dl,"\display^T_B"']
    \ar[dr,"\display^T_C"]
    \\
    A&&B&&C
  \end{tikzcd}
  \label{eq:pbcomp}
\end{equation}
\endgroup
and setting $\display^{T\odot S}_A = \display^S_A \circ \pl$ and
$\display^{T\odot S}_C = \display^T_C \circ \pr$ -- for
$\Span(\Set)$, this means that $T\odot S$ has elements all pairs $(s,
t)$ such that $\display^S_B(s) = \display^T_B(t)$, projected to $A$ and
$C$ via $\display^{T\odot S}_A((s, t)) = \display^S_A(s)$ and
$\display^{T\odot S}_C((s, t)) = \display^T_C(t)$.

This composition need not be associative on the nose,
but the universal property of pullbacks entails that it is associative
up to canonical isomorphism -- forming a \emph{bicategory}:

\begin{theorem}
If $\C$ has pullbacks, then $\Span(\C)$ defined with
\[
\begin{array}{rl}
\text{\emph{objects:}} & \text{objects of $\C$,}\\
\text{\emph{morphisms:}} & \text{spans $A \leftarrow S \rightarrow B$,}\\
\text{\emph{$2$-cells:}} & \text{morphisms of spans,}
\end{array}
\]
forms a bicategory, denoted $\Span(\C)$.
\end{theorem}

In fact, $\Span(\C)$ is a compact closed bicategory
\cite{stay2013compact}, and thus a model of the linear
$\lambda$-calculus. In particular, $\Span(\Set)$ shares much structure
with $\Rel$: it has the same objects and the operation sending a span
$A \leftarrow S \rightarrow B$ to the pairs $(\display^S_A(s),
\display^S_B(s))$ for $s \in S$ is a functor, establishing
$\Span(\Set)$ as a natural candidate for a proof-relevant relational
model.

\subsubsection{The exponential}\label{subsubsec:naive_exp_span}
However, the exponential of $\Rel$ does
not directly transport to $\Span$. The operation $\Mf(-)$ does yield a
functor on $\Set$ obtained by setting, for $f : A \to B$,
\[
\Mf(f)([a_1, \dots, a_n]) = [f(a_1), \dots, f(a_n)] 
\]
defining $\Mf(f) : \Mf(A) \to \Mf(B)$. But $\Mf(f)$ does not lift to
$\Span(\Set)$ as it does not preserve pullbacks. Indeed, the diagram
obtained by image of the composition pullback
\[
  \begin{tikzcd}[cramped,rsbo=normal,csbo=5.5em]
    &
    \Mf(T\odot S)
    \ar[dl,"\Mf(\pl)"']
    \ar[dr,"\Mf(\pr)"]
    &
    \\
    \Mf(S)
    \ar[dr,"\Mf(\display^S_B)"']
    &&
    \Mf(T)
    \ar[dl,"\Mf(\display^T_B)"]
    \\
    &\Mf(B)
  \end{tikzcd}
\]
is no pullback: this would need a bijection of $\Mf(T\odot S)$
with
\[
  \{(\mu, \nu) \in \Mf(S) \times \Mf(T) \mid \Mf(\display^S_B)(\mu) =
  \Mf(\display^T_B)(\nu) \}\,,
\]
which fails in general. If $S = T = \tbool$ and $B = 1$, the pair of
multisets $([\ttrue, \tfalse], [\ttrue, \tfalse])$ does not uniquely
specify who is synchronized with whom: it may correspond to
both multisets $[(\ttrue, \ttrue), (\tfalse, \tfalse)]$ and $[(\ttrue,
\tfalse), (\tfalse, \ttrue)]$ in $\Mf(T\odot S)$. 

This might be expected: a finite multiset only remembers the
multiplicity of elements, but does not track distinct individual
occurrences. This is in tension with the goal of a proof-relevant
relational semantics, for which specific witnesses are naturally
associated with individual resource occurrences.

\subsubsection{Categorifying objects} If the exponential is to track
individual resource occurrences, that means avoiding the quotient of
finite multisets: an element of $\oc A$ may for instance be a
\emph{list}, or a \emph{word} $a_1 \dots a_n \in A^*$ of elements of
$A$. We must of course still account for reorderings, which turn $A^*$
into a \emph{groupoid} -- in fact, it is an instance of the
construction of the \emph{free symmetric monoidal category} $\Sym(A)$
over a category $A$: its objects are finite words $a_1 \dots a_n$ of
objects of $A$, and a morphism from $a_1 \dots a_n$ to $a'_1 \dots
a'_n$ consists of a permutation $\pi \in \Perm{n}$, and a family $(f_i
\in A(a_i, a_{\pi(i)}))_{1\leq i \leq n}$.

Thus, objects are not mere sets but categories, which means that we move
from $\Span(\Set)$ to $\Span(\Cat)$. Indeed, $\Cat$ also has
pullbacks, and so the exact same construction as above yields a
bicategory $\Span(\Cat)$ -- except that now the functor
$\Sym : \Cat \to \Cat$ preserves pullbacks and thus lifts to
\[
  \Sym : \Span(\Cat) \to \Span(\Cat)\,.
\] 

However, in this categorification, the Seely isomorphism $\Mf(A +
B) \iso \Mf(A) \times \Mf(B)$ is lost. Instead, we only get 
\[
  \Sym(A + B) \simeq \Sym(A) \times \Sym(B)
\]
an \emph{equivalence} of categories. In order to lift it to spans, we
observe that given a functor $F : A \to B$ we get a span
\[
  \hat{F} 
  \quad = \quad
  \begin{tikzcd}[cramped,rsbo=normal,csbo=large]
    &
    A
    \ar[dr,"F"]
    \ar[dl,"\Id_A"']
    &
    \\
    A&&B
  \end{tikzcd}
  \quad \in \quad
  \Span(\Cat)(A, B)
\]
so that lifting an equivalence $F : A \simeq B : G$ to spans requires
us to provide a family of $2$-cells, \emph{i.e.} for each category $A$: 
\[
  \begin{tikzcd}[cramped,rsbo=normal]
    &
    A
    \ar[dl,"\Id_A"']
    \ar[dr,"GF"]
    \ar[dd,dotted,"?"{description}]
    &
    \\
    A&&A
    \\
    &
    A
    \ar[ul,"\Id_A"]
    \ar[ur,"\Id_A"']
  \end{tikzcd}
\]
however whatever our choice for the mediating map is, one of the
triangles fails to commute on the nose but only up to isomorphism, 
which the $2$-cells of $\Span(\Cat)$ are too
strict to accommodate. This invites weakening the $2$-cells to:
\begin{definition}\label{def:weakmor}
  A \textbf{weak morphism} from $A \leftarrow S \rightarrow B$ to $A
  \leftarrow S' \rightarrow B$ is a triple $(F, F^A, F^B)$ where
  \[
    \begin{tikzcd}[cramped,rsbo=normal]
      &&
      S	\ar[dll,"\display^S_A"']
      \ar[drr,"\display^S_B"]
      \ar[dd,"F"{description}]
      &&
      \\
      A
      &
      F^A\!\Downarrow\hspace{-15pt} && \hspace{-15pt}\Downarrow \! F^B&B\\
      &&
      S'
      \ar[ull,"\display^{S'}_A"]
      \ar[urr,"\display^{S'}_B"']
    \end{tikzcd}
  \]
  with $F^A : \display^S_A \Rightarrow \display^{S'}_A \circ F$ and $F^B :
  \display^S_B \Rightarrow \display^{S'}_B \circ F$ natural isos. We call this a
  \textbf{strong morphism} if $F^A$ and $F^B$ are identities.
\end{definition}

Adopting \emph{weak morphisms} seems to solve the problem above, but
only to run into a much more subtle one: in $\Span(\Cat)$, the
horizontal composition of $2$-cells $F : S \Rightarrow S'$ and $G : T
\Rightarrow T'$ as required by the bicategorical structure follows from
the universal property of the pullback $T' \odot S'$:
\begin{equation}
  \begin{tikzcd}[cramped,rsbo={1.5em},csbo={4.5em}]
    &&
    T \odot S
    \ar[dl]
    \ar[dr]
    \phar[dd,"\dcorner",very near start]
    \ar[dddd,dotted,bend left=64,"G\odot F"{description}]
    &&
    \\
    &
    S
    \ar[dl]
    \ar[dd,dotted,"F"{description}]
    \ar[dr]
    &&
    T
    \ar[dl]
    \ar[dd,dotted,"G"{description}]
    \ar[dr]
    \\
    A&&B&&C
    \\
    &
    S'
    \ar[ul]
    \ar[ur]
    &&
    T'
    \ar[ul]
    \ar[ur]
    \\
    &&
    T' \odot S'
    \ar[ul]
    \ar[ur]
    \phar[uu,"\ucorner",very near start]
  \end{tikzcd}
\label{eq:horcomp}
\end{equation}
but this universal property is powerless to compose horizontally weak
morphisms. We cannot have the cake and eat it too: if our method to
compose spans ignores the $2$-categorical nature of $\Cat$, then we
cannot hope composition to preserve an equivalence between spans that
relies on it, as required for a model of linear logic. So it seems
that this road to a proof-relevant relational model is doomed -- except
that this is exactly what we shall do in this paper!

Before we delve into that, we review existing solutions.

\subsection{Proof-Relevant Relational Models, and Other Related Work}

As plain pullbacks are ``too $1$-dimensional'', it is natural to
compose spans with a $2$-dimensional version.

\subsubsection{Bipullbacks} There are multiple variants for weakened
versions of pullbacks in a $2$-category. In this paper, a central
notion will be that of a \emph{bipullback:}\footnote{According to the
nlab, its proper name is a \emph{bi-iso-comma-object}.}
\begin{figure}
\begin{minipage}{.45\linewidth}
\[
\adjustbox{scale=.8,center}{%
\begin{tikzcd}
& P	\ar[dl,swap,"\pl"]
	\ar[dr,"\pr"]\\
S	\ar[dr,swap,"u"]
	\phar[rr,"\xTo{\mu}"]&&
T	\ar[dl,"v"]\\
&B 
\end{tikzcd}}
\]
\caption{A bipullback}
\label{fig:diag1}
\end{minipage}
\hfill
\begin{minipage}{.45\linewidth}
\[
\adjustbox{scale=.8,center}{%
\begin{tikzcd}
& X     \ar[dl,swap,"l'"]
        \ar[dr,"r'"]\\
S       \ar[dr,swap,"u"]
        \phar[rr,"\xTo{\nu}"]&&
T       \ar[dl,"v"]\\
&B 
\end{tikzcd}}
\]
\caption{Alternative square}
\label{fig:diag2}
\end{minipage}
\end{figure}

\begin{definition}\label{def:bipullback}
  In a $2$-category $\C$, a \textbf{bipullback} of the cospan $S \xto{u} B
  \xot{v} T$ is a square commuting up to an invertible $2$-cell as in
  \Cref{fig:diag1}, such that for any square as in \Cref{fig:diag2}:
  \begin{enumerate}[(a),left=0pt .. 16pt]
  \item \label{def:bipullback:existence} There is a morphism $h : X \to P$ and
    $2$-cells $\alpha$ and $\beta$ s.t.:
    \[
      \adjustbox{scale=.8}{%
        \begin{tikzcd}
          & X	\ar[ddl,bend right=30, "l'"',myname=lp]
          \ar[ddr,bend left=30, "r'",myname=rp]
          \ar[d,"h"{description}]\\
          & P     \ar[dl,"\pl"']
          \ar[dr,"\pr"]
          \phar[to=rp,"\xTo{\beta}"]
          \phar[from=lp,"\xTo{\alpha}"]\\
          S       \ar[dr,"u"']
          \phar[rr,"\xTo{\mu}"]&&
          T       \ar[dl,"v"]\\
          &B 	
        \end{tikzcd}}
      \qquad=\qquad
      \adjustbox{scale=.8}{%
        \begin{tikzcd}
          & X     \ar[ddl,bend right=30, "l'"',myname=lp]
          \ar[ddr,bend left=30, "r'",myname=rp]\\\\
          S       \ar[dr,"u"']
          \phar[rr,"\xTo{\nu}"]&&
          T       \ar[dl,"v"]\\
          &B      
        \end{tikzcd}}
    \]
  \item \label{def:bipullback:uniqueness} $h, \alpha, \beta$ are unique up to
    unique $2$-cell -- see \Cref{sec:bipullbacks}.
  \end{enumerate}
\end{definition}

The important observation is that this alternative universal property is
sufficient to extend the definition of the horizontal composition in
\eqref{eq:horcomp} to \emph{weak morphisms} -- with the proviso that
this defines horizontal composition only up to iso; as \emph{(b)} does
not guarantee uniqueness of $h$ on the nose.

\subsubsection{Hoffnung's monoidal tricategory} Hoffnung
\cite{hoffnung2011spans} constructs a categorification of $\Span(\Cat)$
following this idea. He exploits that $\Cat$ actually has
\emph{pseudo-pullbacks}\footnote{According to the nlab, its proper name
is an \emph{iso-comma-object}.}, which are a special case of the
definition above where $\alpha$ and $\beta$ are required to be
identities and $h$ is unique on the nose -- making horizontal
composition of weak morphisms of spans a well-defined function once a
choice of pseudo-pullbacks is fixed.

Concretely, a pseudo-pullback of a cospan $S \xto{u} B \xot{v} T$ may be
constructed as a category with objects triples $(s, \theta, t)$ where $\theta
\in B(u(s), v(t))$. So for instance, if $S = T = \Sym(\tbool)$ and $B =
\Sym(1)$, the pseudo-pullback would have two objects synchronizing $[\ttrue,
\tfalse] \in S$ and $[\ttrue, \tfalse] \in T$: $([\ttrue, \tfalse], \id,
[\ttrue, \tfalse])$ and $([\ttrue, \tfalse, \swap, [\ttrue, \tfalse])$. The
issue of \Cref{subsubsec:naive_exp_span} is avoided by adding new witnesses
carrying all possible symmetries. This is a fundamental phenomenon in models of
linear logic, which we refer to as \emph{saturation}.

Because saturation \emph{inflates} the number of witnesses at each
composition, spans composed by pseudo-pullbacks no longer form a
bicategory. In particular, the post-composition of a span $A \ot S \to
B$ with the identity span $B \ot B \to B$ yields an inflated $S'$ much
bigger than $S$. So neutrality of identity no longer holds up to
isomorphism, but only up to \emph{equivalence} factoring in \emph{maps
between maps of spans}. Accordingly, Hoffnung actually constructs a
\emph{monoidal tricategory} of categorical spans with weak morphisms,
\emph{i.e.} a one-object \emph{tetracategory}! 

\subsubsection{Melliès' template games} Recently, Melliès introduced
\emph{template games} \cite{DBLP:journals/pacmpl/Mellies19}, in an
attempt to unify various games models. This is essentially a
model of categorical spans where categories are regarded as games and
structured by a projection to a category called the  \emph{template},
capturing the mechanisms of synchronization and scheduling between
players. 
Though \cite{DBLP:journals/pacmpl/Mellies19} was developped in a purely
linear setting with spans related by strong morphisms, Melliès 
proposed a non-linear extension, forming a model of differential linear
logic \cite{DBLP:conf/lics/Mellies19}. 

Melliès' contribution puts into play notions from \emph{homotopy
theory}: he starts not with $\Cat$, but from any $2$-category $\S$
equipped with a Quillen model structure (with additional conditions).
Spans are composed by mere pullbacks, but a span
\[
A \xot{u} S \xto{v} B
\]
must satisfy a fibration property to the effect that
symmetries in $A$ and $B$ can be lifted uniquely in $S$ -- in
our terminology, $S$ is \emph{saturated}. Saturation ensures that
pullbacks between those spans are actually \emph{homotopy pullbacks},
and thus that they may be used for the horizontal composition of weak
morphisms. The higher dimensional structure seen in Hoffnung
\cite{hoffnung2011spans} is then managed by the homotopy-theoretic
operation of \emph{localization}, formally inverting weak equivalences.
This yields an actual \emph{bicategory} of objects of
$\S$ related by \emph{homotopy spans}.

This elegant construction gives a model of differential linear logic,
showing that the symmetries implicit in linear logic may be naturally
managed via the tools of homotopy theory.

\subsubsection{Generalized Species of Structures}
Last but not least, the most well-studied proof-relevant extension
of $\Rel$ is definitely Fiore, Gambino, Hyland and Winskel's cartesian
closed bicategory of generalized species of structure
\cite{fiore2008cartesian}. Relations from $A$ to $B$ are
replaced with
\emph{distributors} or \emph{profunctors}:
\[
F : A^{\op} \times B \to \Set
\]
for $A$ and $B$ categories. This forms a
(compact closed) bicategory $\Dist$ of (small) categories, distributors and
natural transformations between them. The free symmetric monoidal
construction $\Sym(-)$ yields a pseudocomonad on $\Dist$, whose Kleisli
bicategory $\Esp$ is cartesian closed.

As for the span-based approaches above, the way in which $\Dist$ and
$\Esp$ handle symmetries is saturated. This may first be seen in the
identity distributor which is defined to be
\[
A[-, -] : A^{\op} \times A \to \Set
\]
the Yoneda embedding,
which associates as witnesses over a pair $(a, a)$ the homset
$A[a, a]$, including all symmetries on $a$.

Composition of distributors is via the coend formula
\[
G \odot F = \int^{b \in B} F(-, b) \times G(b, -)
\]
which sets witnesses in $(G\odot F)(a, c)$ to be pairs $(s, t) \in F(a,
b) \times G(b, C)$ quotiented by a relation identifying the action of a
morphism in $B$ on $s$ or on $t$.

Accordingly, when computing the interpretation of a program $\vdash M :
A$ in $\Esp$, for every $a \in \intr{A}$ we get $\intr{M}(a)$ a set of
witnesses carrying around explicit symmetries, quotiented by an
equivalence relation letting symmetries flow around -- this is
described syntactically elegantly by Olimpieri
\cite{DBLP:conf/lics/Olimpieri21}. The treatment of symmetry in
$\Esp$ is, again, saturated.

\subsubsection{Game semantics} To our knowledge, this saturation
phenomenon in models of linear logic first appears in Baillot, Danos,
Ehrhard and Regnier's (BDER) version \cite{DBLP:conf/lics/BaillotDE97}
of Abramsky-Jagadeesan-Malacaria (AJM) games
\cite{DBLP:journals/iandc/AbramskyJM00}. 

In AJM games, the \emph{moves} of a game $\oc A$ are defined as pairs
$(i, a)$ of $i \in \mathbb{N}$ a \textbf{copy index}, and $a \in A$ a
move in $A$ -- a fundamental difficulty in setting up the games model,
is that of \emph{uniformity}: ensuring that the behaviour of
strategies does not depend on the specific choice of copy indices
(which is the game semantics analogue of composition preserving weak
morphisms). In BDER, uniformity is guaranteed by requiring strategies
to be \emph{saturated}: they are morally wrapped by copycat processes
exchanging non-deterministically all copy indices. This ``noise''
prevents strategies from seeing specific copy indices, let alone
depending on them -- this is analogous to the saturation phenomenon above.

But in AJM games there is another choice: in the original AJM setting
\cite{DBLP:journals/iandc/AbramskyJM00}, strategies carry a
deterministic choice of copy indices. Instead of saturation, uniformity
is guaranteed by requiring that strategies satisfy a bisimulation-like
property, which ensures that whenever Opponent swaps their copy
indices, Player can swap theirs accordingly, leaving the behaviour
``up to copy indices'' invariant. In contrast to the ``saturated''
approach to uniformity, we refer to this as the ``thin'' approach. 
Similar ideas are at play in other
game models based on copy indices: in Melliès' orbital games
\cite{DBLP:journals/tcs/Mellies06}, and more recently in \emph{thin
concurrent games}\footnote{The first version of concurrent
games with symmetry was saturated \cite{DBLP:conf/csl/CastellanCW14}.}
\cite{DBLP:journals/lmcs/CastellanCW19}. 

Thin concurrent games are a particularly striking related work, because
just as $\Esp$, they \emph{also} form a cartesian closed bicategory as
proved by Paquet \cite{paquet2020probabilistic}, and also generalize the
relational model \cite{DBLP:journals/corr/abs-2107-03155}. In thin
concurrent games, strategies are composed by pullback. But it is a
theorem that this pullback is also a bipullback, which can be used to
compose horizontally weak morphisms even though strategies are not
saturated. But this bipullback property follows from a subtle
interactive reindexing mechanism between strategies, relying crucially
on the fact that we have access to time -- it seems hard to replicate
it purely statically as in a relational model.

\section{The Bicategory of Thin Spans}
\label{sec:bicat_thin}

This long discussion lets us state the main question in this
paper: can we construct a \emph{thin} version of 
categorical spans? 

\subsection{Pullbacks and Bipullbacks in Groupoids}

For simplicity, we focus on spans of \emph{groupoids} 
rather than categories, which are sufficient for the 
interpretation of types -- we write $\Gpd$ for the small $2$-category
of groupoids.
So we aim to construct a bicategory whose objects are small groupoids,
whose morphisms are spans $A \ot S \to B$ with identity the identity
span $A \ot A \to A$, whose composition is plain
pullback and yet, whose $2$-cells are \emph{weak} morphisms.

\subsubsection{Notations and terminology} A span $A \ot S \to B$ may be
presented as a functor $S \to A \times B$, so it is convenient not to focus on
spans, but on functors $S \to A$ over a groupoid $A$. We refer to those with
terminology inspired from game semantics. A \textbf{prestrategy} on groupoid $A$
is a pair $(S, \display^F)$ where $\display^F : S \to A$ is called the
\textbf{display map}. We often refer to the prestrategy only with $S$, and write
$\PreStrat(A)$ for the set of prestrategies on $A$. A \textbf{prestrategy from
  $A$ to $B$} is a prestrategy on $A \times B$ -- then, we write $\display^F_A :
S \to A$ and $\display^F_B : S \to B$ for the two display maps. If $S$ is a
prestrategy from $A$ to $B$ and $T$ a prestrategy from $B$ to $C$, we write $T
\odot S$ for the prestrategy from $A$ to $C$ obtained as in
\Cref{subsubsec:bicategory_prelim}. We often refer to morphisms in groupoids as
\textbf{symmetries} and write \emph{e.g.} $\varphi : s \sym_S s'$ instead of
$\varphi \in S(s, s')$.

We write $1$ for the groupoid with one object $*$, and only the
identity morphism; and $o$ for the groupoid with one object $\bullet$
and only the identity morphism. If $A, B$ are groupoids, then we
use $A \vdash B$ and $A \lin B$ as synonyms for $A \times B$, 
with objects respectively denoted by $a \vdash b$ and
$a \lin b$ -- likewise, their morphisms have form $\theta_A \vdash
\theta_B \in (A \lin B)(a \lin b, a' \lin b')$ for $\theta_A \in A(a,
a')$ and $\theta_B \in B(b, b')$ and likewise for $\theta_A \lin
\theta_B$.  We find these purely notational distinctions useful to read
examples, since they coincide with familiar type constructors.

\subsubsection{Indexed families} As explained earlier, types of
$\lambda$-calculi may be interpreted as groupoids -- but in a linear
language, these groupoids remain \emph{discrete}: only the exponential
introduces non-trivial morphisms. As those \emph{symmetries} play a
crucial role, we introduce early our version of the
exponential construction.
If $X$ is a set, then we write $\Fam(X)$ the set of families indexed by
\emph{finite sets of natural numbers}, \emph{i.e.} of $(x_i)_{i\in I}$
where $I \subseteq_f \N$ and for all $i \in I$, $x_i \in X$.

\begin{definition}\label{def:fam}
Consider $A$ a (small) groupoid.
The (small) groupoid $\BFam(A)$ has:
\emph{objects}, the set $\Fam(A)$; \emph{morphisms} from $(a_i)_{i\in
I}$ to $(b_j)_{j\in J}$, pairs $(\pi, (f_i)_{i\in I})$ of a bijection
$\pi : I \bij J$ and for each $i \in I$, $f_i \in A(a_i, b_{\pi(i)})$.
\end{definition}

This yields a functor $\BFam : \Gpd \to \Gpd$ in the obvious way.
For $(A_i)_{i\in I} \in \Fam(A)$, we call elements of $I$
\textbf{copy indices}. A family $(a_i)_{i\in I} \in
\Fam(A)$ is more ``intensional'' than $A^*$ (which is more
intensional than $\Mf(A)$): it gives a presentation of a multiset in
$\Mf(A)$ not only providing a sequence, but assigning to each element a
distinct ``address''.

Just as multisets are connected to non-idempotent intersection types,
families are connected with Vial's \emph{sequence types}
\cite{DBLP:conf/lics/Vial17} -- thus we often write families using Vial's
notation, \emph{e.g.}
\[
[2 \cdot a_2, 4 \cdot a_4, 12\cdot a_{12}] \in \Fam(A)
\]
for $(a_i)_{i\in \{2, 4, 12\}}$ -- in the particular case where $A =
o$, we only write $[i_1, \dots, i_n]$ for $[i_1 \cdot \bullet, \dots,
i_n \cdot \bullet]$.

For any groupoid $A$, $\BFam(A)$ and $\Sym(A)$ are equivalent. However,
using $\BFam(A)$ is crucial in our model construction: it allows the
interpretation of programs to use copy indices as \emph{identifiers}
for resource accesses, that are independent of other concurrent
resource accesses. 
We give a few examples:

\begin{example}\label{ex:dereliction}
For a groupoid $A$, the \textbf{dereliction} span $\bder_A$ is
\[
\BFam(A) \xot{\der_A} A \xto{\id_A} A
\]
where $\der_A : A \to \BFam(A)$ sends $a$ to $[0\cdot a]$.
\end{example}

In models of linear logic, the role of dereliction is to extract a
\emph{single} instance of a replicable resource. In our model -- as in
AJM games \cite{DBLP:journals/iandc/AbramskyJM00} and thin concurrent
games \cite{DBLP:journals/lmcs/CastellanCW19} -- dereliction does so by
picking a copy index (here $0$), chosen in advance once and for all.
The specific choice is irrelevant; in fact for any $n$ the span using
$n$ instead of $0$ will be turn out to be isomorphic to $\bder_A$. But,
the span must comprise a choice.

%

\begin{example}
  The interpretation of the term $M$ of \Cref{ex:prelim_term} in thin spans
  shall have head groupoid that with four objects
\[
\begin{array}{ll}
\,[0 \cdot \ttrue, 1 \cdot \ttrue] \lin \ttrue\,,
\qquad
&[0 \cdot \tfalse, 1\cdot \tfalse] \lin \ttrue\,,\\
\,[0 \cdot \ttrue, 1 \cdot \tfalse] \lin \tfalse\,,
\qquad
&[0 \cdot \tfalse, 1 \cdot \ttrue] \lin \ttrue\,,
\end{array}
\]
morphisms reduced to identities, and display map the identity.
\end{example}

The use of
specific copy indices allows one to observe which occurrence of $x$
evaluates to $\ttrue$ or $\tfalse$, hence associating distinct points
to the two evaluations leading to $\tfalse$.

\subsubsection{Bipullbacks of groupoids} If composition-by-pullback is
to allow us to compose horizontally weak morphisms, we must ensure that every
composition pullback is \emph{also} a bipullback. 

It is useful to understand a bit better the shape of bipullbacks in $\Gpd$. A
first useful fact is that condition \Cref{def:bipullback:uniqueness} of
\Cref{def:bipullback} (uniqueness up to iso) automatically holds in
the case of $\Gpd$; furthermore, we can characterise those pullbacks that are
also bipullbacks (see \Cref{sec:pbs-in-gpd}):

\begin{lemma}\label{lem:carac_pb_bipb}
A pullback square in $\Gpd$, of the form
\[
  \begin{tikzcd}[cramped,rsbo=2em,csbo=2.5em]
    &
    P
    \ar[ld,"\pl"']
    \ar[rd,"\pr"]
    \phar[dd,"\dcorner",very near start]
    &
    \\
    S
    \ar[rd,"f"']
    &&
    T
    \ar[ld,"g"]
    \\
    &
    B
  \end{tikzcd}
\]
is a bipullback if and only if it satisfies the following property: for
all $s \in S$, $t \in T$ and $\theta \in B(fs, gt)$, there is $\varphi
\in S(s, s')$ and $\psi \in T(t', t)$ such that $f s' = g t'$ and
$\theta = f\psi \circ g\varphi$. 
\end{lemma}

Let us comment on this. We regard triples of the form
\[
s \in S\,,
\qquad
\theta \in B(f s, g t)\,,
\qquad
t \in T
\]
as pairs of states $(s, t)$ that match \emph{up to symmetry} -- we call
this a \textbf{reindexing problem}. The lemma above says that given a
reindexing problem, we can always find $s'$ symmetric to $s$ and $t'$
symmetric to $t$ matching \emph{on the nose}, in a way compatible with
$\theta$
-- called a \textbf{solution} to the reindexing problem. Thus, the
lemma above may be reformulated to say that a pullback is a bipullback
iff all its reindexing problems have a solution. 

We show a concrete example of this reindexing process:

\begin{example}\label{ex:reindexing_pb}
Take $B = \BFam(o) \lin \BFam(o)$, with objects
\[
[i_1,  \dots, i_n] \lin [j_1, \dots, j_k]\,.
\]

Take $S$ the sub-groupoid of $B$ with objects $[i_1, \dots,
i_n] \lin [i_1, \dots, i_n]$ and morphisms all $\theta \lin \theta$ for
$\theta$ in $\BFam(o)$; and $T$ the full sub-groupoid of $B$
with objects $[j_1, \dots, j_n] \lin [0]$. 

The pullback of $S \to B \ot T$ is a bipullback.
For instance, 
\[
\theta 
\in
B([2] \lin [2],
[1] \lin [0])
\]
is a reindexing problem that may be solved by first applying
\[
\varphi \in S([2] \lin [2], [0] \lin [0])
\]
in $S$. We are reduced to finding morphisms in $S$ and $T$
\emph{w.r.t.}
\[
\theta' \in 
B([0] \lin [0],
[1] \lin [0])
\]

Now, applying $\psi \in T([0] \lin [0], [1] \lin [0])$ in $T$, we have
\[
\varphi \in S([2] \lin [2], [0] \lin [0])\,,
\quad
\psi \in T([0] \lin [0], [1] \lin [0])
\]
a solution to the reindexing problem, as in \Cref{lem:carac_pb_bipb}. 
\end{example}

That the pullback of two prestrategies forms a bipullback is not a
property of either: in this example neither strategy is a fibration as
in \cite{DBLP:conf/lics/Mellies19}, and solving the reindexing problem
requires reindexing in \emph{both} groupoids. So it is a
property emerging from the harmonious interaction between two
prestrategies. In an appropriate game semantics setting
\cite{DBLP:journals/lmcs/CastellanCW19}, one can \emph{prove} that
under reasonable assumptions, such interactive reindexing always
succeeds. However, this is a gradual process progressing over
time -- which we do not have access to here. 

\subsection{Orthogonality and Uniform Groupoids} 

\subsubsection{Definition}
In the literature on models of linear
logic, there is a technique for choreographing models where one only
composes pairs of morphisms satisfying a given interactive property:
\emph{biorthogonality}. 
The first step is to specify the desired interactive property via an
orthogonality relation: 

\begin{definition}
Take $(S, \display^S)$ and $(T, \display^T)$ prestrategies on
$B$.

We say they are \textbf{uniformly orthogonal}, written $S \perp T$, iff
the pullback of the cospan $S \to B \ot T$ is also a bipullback.
\end{definition}

If $\bS\subseteq \PreStrat(B)$, then its \textbf{uniform orthogonal} is
set to:
\[
  \bS^\perp = \{T \in \PreStrat(B) \mid \forall S \in \bS,~ S \perp T\}\zbox.
\]

As usual with orthogonality, this automatically entails a
number of properties: for all $\bS \subseteq \PreStrat(B)$, we have
$\bS \subseteq \bS^{\perp \perp}$, and $\bS^\perp = \bS^{\perp \perp
\perp}$. We are particularly interested in sets of the form
$\bS^\perp$, which are \emph{invariant under biorthogonal}:

\begin{definition}
A \textbf{uniform groupoid} is a pair $(A, \U_A)$ where $A$ is a
groupoid and $\U_A \subseteq \PreStrat(A)$ is s.t.
$\U_A^{\perp\perp} = \U_A$.
\end{definition}

We often refer to a uniform groupoid $(A, \U_A)$ just with $A$ when it
is clear from the context that it is a uniform groupoid.

\subsubsection{Constructions} The uniform groupoid $1$ is the terminal
groupoid equipped with $\U_1 = \PreStrat(1)$. If $A$ and $B$ are
uniform groupoids, their \textbf{tensor} $A \tensor B$ is the groupoid
$A \times B$ equipped with
the set $\U_{A \tensor B} = (\U_A \tensor \U_B)^{\perp \perp}$, writing
\[
\U_A \tensor \U_B = \{(S \times T, \display^S \times \display^T) \mid S \in \U_A,~T \in \U_B\} 
\]
with $\display^S \times \display^T : S \times T \to A \times B$. The
\textbf{dual} $A^\perp$ of $A$ is $(A, \U_{A^\perp})$ with
$\U_{A^\perp} = \U_A^\perp$. The \textbf{par} of $A$ and $B$
has
\[
\U_{A \parr B} = (\U_A^\perp \tensor \U_B^\perp)^\perp
\]
yielding the De Morgan duality $(A \tensor B)^\perp = A^\perp \parr
B^\perp$. 
From this we derive the \textbf{linear arrow} $A \lin B = A^\perp \parr
B$. 

A \textbf{uniform prestrategy} on uniform groupoid $A$ is simply any $S
\in \U_A$. If $A, B$ are uniform groupoids, then a \textbf{uniform
prestrategy from $A$ to $B$} is a uniform prestrategy on $A \lin
B$.

\subsubsection{Uniform composition} We claim that whenever composing $S
\in \U_{A\lin B}$ with $T \in \U_{B \lin C}$, we have the orthogonality
\[
(S, \display^S_B) \perp (T, \display^T_B)
\]
so that the composition pullback is a bipullback. 


If $S$ is a prestrategy on $A$ and $T$ is a prestrategy from $A$ to
$B$, we write $T@S$ from the prestrategy on $B$ obtained by
\[
  \begin{tikzcd}[cramped,rsbo=2em,csbo=3em]
    &
    T\odot S
    \phar[dd,"\dcorner",very near start]
    \ar[dl]
    \ar[dr]
    &
    &
    \\
    S
    \ar[dr]
    &&
    T
    \ar[dl]
    \ar[dr]
    \\
    &A&&B
  \end{tikzcd}
\]
called the \textbf{application} of $T$ to $S$. This lets us state:

\begin{proposition}
  \label{prop:carac_lin}
  Consider $(A, \U_A)$ and $(B, \U_B)$ uniform groupoids, and $T$ a
  prestrategy from $A$ to $B$; consider furthermore a class $\SU \subseteq
  \U_A$ s.t. $(A, \id_A) \in \SU$ and $\U_A = \SU^{\perp \perp}$.

  Then $T \in \U_{A \lin B}$ iff
  the following two conditions hold:
  \begin{enumerate}[(1)]
  \item \label{prop:carac_lin:pres} for all $S \in \SU$, $T@S \in \U_B$,
  \item \label{prop:carac_lin:anti-unif} $(T, \display^T_A) \in \U_A^\perp$.
  \end{enumerate}
\end{proposition}
\begin{proof}
  Unfolding the definitions, one encounters a few diagram chasing lemmas on
  pullbacks that are also bipullbacks -- themselves proved via
  \Cref{lem:carac_pb_bipb}. See \Cref{sec:unif-and-thinness}.
\end{proof}

The apparent asymmetry is intriguing: by definition $A^\perp \parr B = A^\perp
\parr B^{\perp \perp}$, so that $T \in \U_{A\lin B}$ iff the span $B \ot T \to
A$ denoted by $T^\star$ obtained by reversing the two legs, is in $\U_{B^\perp
  \lin A^\perp}$. A similar phenomenon appears in the orthogonality used by
Ehrhard for his extensional collapse \cite{DBLP:journals/tcs/Ehrhard12}.


Now, observe that $(A, \id_A) \in \U_A$ always -- not the identity span, but the
identity functor regarded as a prestrategy on $A$. Indeed, if $S \in
\U_A^\perp$, then the pullback of $A \to A \ot S$ is clearly a bipullback, so
$(A, \id_A) \in \U_A^{\perp\perp} = \U_A$. But now this lets us instantiate
\Cref{prop:carac_lin} with $\SU = \U_A$. Then given $S \in \U_{A\lin B}$, the
application $S@(A, \id_A)$ is (up to iso) the right leg $(S, \display^S_B)$,
which must by \emph{(1)} be in $\U_B$. Likewise, if $T \in \U_{B \lin C}$, the
left leg $(T, \display^T_B)$ is in $\U_B^\perp$. Hence,
\[
(S, \display^S_B) \quad \perp \quad (T, \display^T_B)
\]
and thus the composition pullback of $S$ and $T$ is a bipullback.

\Cref{prop:carac_lin} has more consequences, all obtained in the particular case
where $\SU = \U_A$: we saw above that $(A, \id_A) \in \U_A$, but the same
argument goes to show $(A, \id_A) \in \U_A^\perp$ as well -- so the identity
span satisfies condition \emph{(2)}. Since it also satisfies \emph{(1)}, we have
$(A \ot A \to A) \in \U_{A\lin A}$ as expected. Likewise, if $A \ot S \to B$ and
$B \ot T \to C$ are uniform prestrategies, then it follows fairly easily that
the composition $A \ot T\odot S \to C$ is indeed in $\U_{A\lin C}$ (see
\Cref{sec:unif-and-thinness}).

\subsubsection{Horizontal composition of $2$-cells}
\label{subsubsec:horcomp}
We have an identity uniform prestrategy in $\U_{A\lin A}$, and a
well-defined composition of $S \in \U_{A\lin B}$ and $T \in \U_{B \lin
C}$ such that the composition pullback is always a bipullback. So given
weak morphisms 
\[
  \begin{tikzcd}[cramped,rsbo=1.6em,csbo=2.4em]
    &&
    S	\ar[dll,"\display^S_A"']
    \ar[drr,"\display^S_B"]
    \ar[dd,"F"{description}]
    &&
    \\
    A
    &
    F^A\!\!\Downarrow\hspace{-15pt} && \hspace{-15pt}\Downarrow \!\! F^B&B\\
    &&
    S'
    \ar[ull,"\display^{S'}_A"]
    \ar[urr,"\display^{S'}_B"']
  \end{tikzcd}
  \qquad
  \begin{tikzcd}[cramped,rsbo=1.6em,csbo=2.4em]
    &&
    T	\ar[dll,"\display^T_B"']
    \ar[drr,"\display^T_C"]
    \ar[dd,"G"{description}]
    &&
    \\
    B
    &
    G^B\!\!\Downarrow\hspace{-15pt} && \hspace{-15pt}\Downarrow \!\! G^C&C\\
    &&
    T'
    \ar[ull,"\display^{T'}_B"]
    \ar[urr,"\display^{T'}_C"']
  \end{tikzcd}
\]
by the bipullback property of $T' \odot S'$ there are a functor $H$ and
natural isos $\alpha$ and $\beta$ such that we have the equality
\[
  \adjustbox{scale=.8}{%
    \begin{tikzcd}[cramped,row sep=10pt, column sep=20pt]
      & T \odot S
      \ar[dl]
      \ar[dr]
      \cphar[ddd,"=",pos=0.33]
      \\
      S	\ar[d]
      \ar[ddr,bend left=30,myname=lp]&&
      T	\ar[d]
      \ar[ddl,bend right=30,myname=rp]\\
      S'	\ar[dr]
      \phar[to=lp,"\xTo{(F^{\!B})^{-1}\!\!}"]&&
      T'	\ar[dl]
      \phar[from=rp,"\xTo{G^{\!B}\!\!}"]\\
      &B
    \end{tikzcd}}
  \quad=\quad
  \adjustbox{scale=.8}{%
    \begin{tikzcd}[cramped,row sep=10pt, column sep=20pt]
      & T \odot S
      \ar[dl]
      \ar[dr]
      \ar[d,"H"]\\
      S	\ar[d]
      \phar[r,"\xTo{\alpha}"]&
      T'\odot S'
      \ar[dl]
      \ar[dr]
      \phar[r,"\xTo{\beta}"]
      \cphar[dd,"="]
      &
      T	\ar[d]\\
      S'	\ar[dr]&&
      T'	\ar[dl]\\
      &B
    \end{tikzcd}}
\]
altogether yielding a weak morphism as in the diagram:
\[
\adjustbox{scale=.8}{%
  \begin{tikzcd}[cramped,row sep=10pt, column sep=25pt]
    &S	\ar[dl,"\partial^S_A"']
    \ar[dd,myname=ds,"F"]&
    T\odot S\ar[l]
    \ar[dd,myname=dh,"H"]
    \ar[r]&
    T	\ar[dd,myname=dg,"G"']
    \ar[dr,"\partial^T_C"]\\
    A	\phar[to=ds,"\Downarrow\!\scalebox{.8}{$F^A$}"]&
    ~	\phar[to=dh,"\Downarrow\!\scalebox{.8}{$\alpha$}"]&
    ~	\phar[to=dg,"\Downarrow\!\scalebox{.8}{$\beta^{-1}$}"]&&C
    \phar[to=dg,"\Downarrow\!\scalebox{.8}{$G^C$}"]\\
    &S'	\ar[ul,"\partial^{S'}_A"]&
    T'\odot S'
    \ar[l]
    \ar[r]&
    T'	\ar[ur,"\partial^{T'}_C"']
  \end{tikzcd}}\,.
\]

However, $H, \alpha, \beta$ are not unique: though \Cref{lem:carac_pb_bipb}
guarantees the existence of solutions to all reindexing problems, those may not
be unique.
We only know by condition \emph{(b)} of \Cref{def:bipullback} that different
choices of $H, \alpha, \beta$ yield \emph{isomorphic} weak morphisms of uniform
prestrategies, by which we mean isomorphic morphisms of the $2$-category
$\Unif(A)$:

\begin{definition}
Consider $A$ a uniform groupoid. 

The $2$-category $\Unif(A)$ has: \emph{objects} $\U_A$,
\emph{i.e.} uniform prestrategies on $A$; \emph{morphisms} from $S$ to
$T$ the weak morphisms, \emph{i.e.} pairs $(F : S \to T, \phi :
\display^S \Rightarrow \display^T F)$;
$2$-cells from $(F, \phi)$ to $(G, \psi)$ the natural
transformations $\mu : F \Rightarrow G$ such that:
\[
\adjustbox{scale=.8}{
\begin{tikzcd}
S	\ar[rr,bend left=20,"G",myname=G]
	\ar[dr]&&
T	\ar[dl]\\
&A	\phar[from=G,"\xTo{\psi}"]	
\end{tikzcd}}
\quad = \quad
\adjustbox{scale=.8}{
\begin{tikzcd}
S	\ar[rr,bend left=20,"G",myname=G]
	\ar[rr,bend right=20,"F"{description},myname=F]
	\phar[from=F,to=G,"\Uparrow\!\mu"]
	\ar[dr]&&
T	\ar[dl]\\
&A	\phar[from=F,"\xTo{\phi}"]	
\end{tikzcd}}\,.
\]
\end{definition}

Thus, although bipullbacks guarantee the existence of a fitting weak
morphism for horizontal composition, there is a priori no canonical
choice. One could pick a choice of horizontal composition, but there is
no reason why an arbitrary choice would satisfy the coherence
conditions for a bicategory.

\subsection{Thin Spans of Groupoids}

In fact, if formulated in the adequate way, the reindexing problems
that arise from the interpretation of programming languages \emph{do
have} a unique solution -- as in \Cref{ex:reindexing_pb}.
But to prove that, we shall need to add more structure to uniform groupoids,
starting with \emph{polarized sub-groupoids}:

\subsubsection{Polarized sub-groupoids} \label{subsubsec:pol_subgr}
Consider the groupoid 
\[
\BFam(o) \lin \BFam(o) 
\]
of \Cref{ex:reindexing_pb}, interpreting the formula $\oc o \lin \oc o$ of
intuitionistic linear logic. Here, the two occurrences of $\oc$ are intuitively
very different: on the left-hand side, as in \Cref{ex:dereliction} the
\emph{program} performs the copying -- in game semantics the copy index would be
carried by a Player move. In contrast, for the right hand side exponential, the
\emph{environment} does the copying -- in game semantics, the copy index would
be carried by an Opponent move. This assigns a polarity to certain symmetries,
very clear in game semantics: those reindexing copy indices only for
exponentials in covariant position (resp. contravariant position) are
\emph{negative} (resp. \emph{positive}). We enrich the groupoids interpreting
types to keep track of these special symmetries:

\begin{definition}
A \textbf{polarized groupoid} is a groupoid $A$ with two
sub-groupoids $A_-$ and $A_+$, with the same objects as $A$.
\end{definition}

It would be natural to require additional conditions for this structure (in
particular, see conditions \emph{(a)} and \emph{(b)} in
\Cref{lem:properties_thin}). We omit them here, as they shall hold automatically
once we introduce the more complete notion of a \emph{thin groupoid}.

If $\theta \in A_-(a_1, a_2)$, we write $\theta : a_1 \sym_A^- a_2$ and
likewise for positive symmetries.  Usual constructions on groupoids
extend to polarized groupoids componentwise. The \textbf{dual} of $(A,
A_-, A_+)$ is defined as $(A, A_+, A_-)$, exchanging the two
sub-groupoids.  Finally, we set $(\oc A)_- = \BFam(A_-)$ and $(\oc A)_+
= \BFam^{\id}(A_+)$, which has morphisms from $(a_i)_{i\in I}$ to
$(b_j)_{j\in J}$ those $(\pi, (\theta_i)_{i\in I})$ such that $I = J$
and $\pi = \id_I$ -- thus we see indeed that Player cannot reindex copy
indices from the outer $\oc$ in $\oc A$, as it appears in covariant
position.

\subsubsection{Thinness} Solutions to reindexing problems may be computed
interactively as in \Cref{ex:reindexing_pb}. Intuitively, the uniqueness of the
solution relies on the fact that at each stage, there is a unique choice of
reindexing. This is captured by the definition of \emph{thin} below, imported
from thin concurrent games:

\begin{definition}\label{def:thin}
Consider $A$ a polarized groupoid, and $S$ a prestrategy on $A$. 
We say that $S$ is \textbf{thin} iff for all $\varphi : s \sym_S s'$,
if $\display^S \varphi$ is positive then $s = s'$ and $\varphi =
\id_s$.
\end{definition}

Intuitively, this captures that positive copy indices are
selected deterministically from negative ones -- so a
non-trivial symmetry $\varphi : s \sym_S s'$ cannot display to a purely
positive symmetry on $A$. This is in contrast with the saturated case,
where spans must be able to reach \emph{all} positive symmetries.

We show how thinness addresses uniqueness for the resolution of reindexing
problems. Call a solution to a reindexing problem $\varphi, \psi$ as in
\Cref{lem:carac_pb_bipb} \textbf{positive} if writing $\display^S \varphi =
\varphi_A \vdash \varphi_B$ and $\display^T \psi = \psi_B \vdash \psi_C$, we
have $\varphi_A \vdash \psi_C$ positive.

\begin{lemma}\label{lem:unique_reindexing}
Consider $A, B, C$ polarized uniform groupoids, $S \in \U_{A\lin B}$
and $T \in \U_{B \lin C}$ s.t. $T \odot S \in \U_{A\lin C}$ is
thin.

Then, any reindexing problem in the composition pullback of $S$ and $T$
has at most one \emph{positive} solution.
\end{lemma}
\begin{proof}
Consider a reindexing problem $s \in S, t \in T, \theta :
\display^S_B s \sym_B \display^T_B t$ with solutions $\varphi_1 : s
\sym_S s'_1$ and $\psi_1 : t'_1 \sym_T t$ with $\display^S_B s'_1 =
\display^T_B t'_1$ and $\display^T_B \psi_1 \circ \display^S_B
\varphi_1 = \theta$, and $\varphi_2 : s \sym_S s'_2$ and $\psi_2 : t'_2
\sym_T t$ with $\display^S_B s'_2 = \display^T_B t'_2$ and
$\display^T_B \psi_2 \circ \display^S_B \varphi_1 = \theta$.  

Then, $\display^S(\varphi_2 \circ \varphi_1^{-1}) = \display^T(\psi_2
\circ \psi_1^{-1})$, so that we have
\[
\Omega = (\varphi_2 \circ \varphi_1^{-1}, \psi_2 \circ \psi_1^{-1})
: (s'_1, t'_1) \sym_{T \odot S} (s'_2, t'_2)\,,
\]
whose display to $A \vdash C$ is positive since $\varphi_1, \psi_1$ and
$\varphi_2, \psi_2$ are positive solutions. Hence, by \emph{thin},
$\Omega$ is an identity map which entails $\varphi_1 = \varphi_2$ and
$\psi_1 = \psi_2$ as required.
\end{proof}

Thus, thinness allows us to find canonical solutions to reindexing
problems by insisting on finding \emph{positive} solutions.

However, this relies on thinness not of $S$ and $T$, but of $T \odot
S$. Again in thin concurrent games, this follows by induction on
the causal structure. In the absence of a handle on
causality, we must as for uniformity treat the fact that $T
\odot S$ is thin as an interactive property, again handled by
biorthogonality.

\subsection{Thin Spans}

\subsubsection{The thin orthogonality} We observe that for $A$ a
polarized groupoid, a prestrategy $S$ on $A$ is thin iff 
the pullback
\begin{equation}
  \begin{tikzcd}[cramped,rsbo=1.5em,csbo=2.8em]
    &
    P
    \phar[dd,"\dcorner",very near start]
    \ar[dl]
    \ar[dr]
    &
    \\
    S
    \ar[dr]
    &&
    A_+
    \ar[dl,"\id^+_A"]
    \\
    &
    A
  \end{tikzcd}
  \label{eq:pbplus}
\end{equation}
is \emph{discrete}, \emph{i.e.} all the morphisms in $P$ are
identities. We shall base our orthogonality on this observation, and
set:

\begin{definition}
For $A$ a polarized uniform groupoid, $S \in \U_A$, and $T \in
\U_A^\perp$, we say $S$ and $T$ are \textbf{thinly orthogonal}, written
\[
S \pperp T
\]
iff the pullback $T \odot S$ is discrete.
\end{definition}

Note that this is already assuming that $S$ and $T$ are uniformly
orthogonal. If $\bS \subseteq \U_A$, then its \textbf{thin orthogonal}
is
\[
\bS^{\pperp} = \{T \in \U_A^{\perp} \mid \forall S \in \bS,~S \pperp T\}\,,
\]
and as before we have $\bS \subseteq \bS^{\pperp \pperp}$ (note that
this typechecks only because $\U_A^{\perp \perp} = \U_A$) and
$\bS^{\pperp} = \bS^{\pperp \pperp \pperp}$ for all $\bS \subseteq \U_A$.

\subsubsection{Thin groupoids}
As before, we are interested in sets of uniform prestrategies
\emph{closed under bi-thin-orthogonal}:

\begin{definition}
A \textbf{thin groupoid} is a polarized uniform groupoid with
a set $\T_A \subseteq \U_A$ of \textbf{strategies} s.t. $\T_A^{\pperp \pperp} =
\T_A$, and such that $(A_-, \id_A) \in \T_A$ and $(A_+, \id_A) \in
\T_A^{\pperp}$.
\end{definition}

If $S \in \T_A$ then $S$ is automatically thin in the sense of \Cref{def:thin}:
as $(A_+, \id_A) \in \T_A^{\pperp}$ the pullback \eqref{eq:pbplus} is discrete.

This also entails properties of the polarized symmetries:

\begin{lemma}\label{lem:properties_thin}
Consider $A$ a thin groupoid. Then we have:

\noindent {(a)} 
if $\theta : a \sym^-_A a'$ and $\theta : a \sym_A^+ a'$, then $a = a'$
and $\theta = \id_a$.

\noindent {(b)} 
if $\theta : a \sym_A a'$, then there are unique $a''$ along with $\theta_- :
a \sym_A^- a''$ and $\theta_+ : a'' \sym_A^+ a'$ such that $\theta =
\theta_+ \circ \theta_-$.
\end{lemma}\simon{lemme utilisé quelque part ?}
\begin{proof}
For \emph{(a)}, this follows from $A_- \pperp A_+$, as then the
pullback of the cospan $A_- \hookrightarrow A \hookleftarrow A_+$ is
discrete. 

For \emph{(b)}, $A_- \in \T_A \subseteq \U_A$ and $A_+ \in
\T_A^{\pperp} \subseteq \U_A^{\perp}$, we also have $A_- \perp A_+$.
Hence, the pullback of the cospan $A_- \hookrightarrow A \hookleftarrow
A_+$ is a bipullback. But then any $\theta : a \sym_A a'$ forms a
reindexing problem, whose solution is exactly the seeked reindexing.
Uniqueness follows immediately from \emph{(a)}.
\end{proof}

Thus, we get from the definition of thin groupoids some of the expected
properties of the polarized sub-groupoids: if a symmetry is both
positive and negative then it must be an identity, and any symmetry can
be obtained by first ``reindexing Opponent moves'', then ''reindexing
Player moves''.

\subsubsection{Further structure}
Constructions on uniform groupoids extend to thin groupoids in the
expected way. The thin groupoid $1$ has $\T_1 =
\PreStrat(1)$. If $A$ and $B$ are thin groupoids, their \textbf{tensor}
is the uniform groupoid $A \tensor B$ extended with $\T_{A\tensor B} =
(\T_A \tensor \T_B)^{\pperp \pperp}$. The \textbf{dual} of $A$ has
$\T_{A^\perp} = \T_A^{\pperp}$. The \textbf{par} of $A$ and $B$ has
$\T_{A\parr B} = (\T_A^{\pperp} \tensor \T_B^{\pperp})^{\pperp}$, and
the \textbf{linear arrow} is $A \lin B = A^\perp \parr B$. 

To establish the compositional properties of strategies, we rely on the
following analogue of \Cref{prop:carac_lin}:

\begin{proposition}
  \label{prop:carac_lin_thin}
  Consider $T \in \U_{A\lin B}$ for $A, B$ thin groupoids, along with a
  class $\SU \subseteq \T_{A}$ such that $\SU^{\pperp \pperp} = \T_A$.

  Then, $T \in \T_{A\lin B}$ iff $T@S \in \T_B$ for all $S \in \SU$.
\end{proposition}

This follows from diagram chasing lemmas on situations where the pullbacks are
discrete, see \Cref{sec:unif-and-thinness}. Interestingly, this is also
equivalent to $T^\star @S \in \T_A^{\pperp}$ for all $S \in \T_B^{\pperp}$.

It is a direct consequence of \Cref{prop:carac_lin_thin} that the identity span
on $A$ is in $\T_{A\lin A}$ for any thin groupoid $A$, and that if $S \in
\T_{A\lin B}$ and $T \in \T_{B \lin C}$ then $T \odot S \in \T_{A\lin C}$.
Together with \Cref{lem:unique_reindexing}, we have thus identified a
compositional situation where the composition pullback of spans is a bipullback,
and where all arising reindexing problems have a \emph{unique} solution if one
insists on this solution being \emph{positive}.

\subsubsection{Positive weak morphisms} Insisting on \emph{positive}
solutions amounts to relating strategies not via arbitrary weak
morphisms, but with \emph{positive} weak morphisms:

\begin{definition}
  Consider $A$ a thin groupoid, $S, T \in \T_A$, and $(F, \phi)$ a weak morphism
  from $S$ to $T$, \emph{i.e.} $F : S \to T$ and $\phi : \display^S
\Rightarrow \display^T \circ F$.
  Then, $(F, \phi)$ is \textbf{positive} if $\phi$ is positive, that
is, if $\forall s \in   S$, $\phi_s : \display^S s \sym_A^+ \display^T
F(s)$ is a positive symmetry.
\end{definition}

Intuitively, comparing strategies with positive morphisms amounts to
relating them only via maps that do not reindex Opponent moves. This
has the effect of making everything stricter, and cutting the higher
dimension. More precisely: 

\begin{proposition}\label{prop:positivization}
Let $A$ be a thin groupoid. Consider $\PreThin(A)$ the sub-$2$-category of
$\Unif(A)$ with objects $\T_A$, and $\Thin(A)$
where additionally morphisms are positive.

Then, $\Thin(A)$ is locally discrete, \emph{i.e.} all $2$-cells are
identities. Moreover, $\PreThin(A)$ and $\Thin(A)$ are biequivalent.
\end{proposition}
\begin{proof}
The first is a direct consequence of \emph{thinness}: if $\mu : (F,
\phi) \Rightarrow (G, \psi) : S \to T$ for $\phi$ and $\psi$ positive,
then by definition of $2$-cells of $\Unif(A)$, for all $s \in S$,
$\mu_s \in T(Fs, Gs)$ is such that $\psi_s = \display^T \mu_s \circ
\phi_s$, \emph{i.e.} $\display^T \mu_s = \psi_s \circ \phi_s^{-1}$
positive. Thus, $\mu_s$ is an identity morphism by thinness.

For the biequivalence, the crux is that if $(F, \phi) : S
\to T$ is a weak morphism, then there is a unique $(F_+, \phi_+) : S
\to T$ positive isomorphic to $(F, \phi)$, and a unique
$2$-cell $\mu$ between them. Uniqueness follows from thinness. For
existence, note that if $s \in S$ and $\theta : \display^S s \sym_A a$,
then there exist unique $\varphi : s \sym_S s'$ and $\theta_+ :
\display^S s \sym_A^+ a$ such that $\theta = \theta_+ \circ \display^S
\varphi$ -- this exploits thinness, and the reindexing problem from the
fact that the pullback of the cospan $S \hookrightarrow A
\hookleftarrow A_+$ is a bipullback. We obtain $(F_+, \phi_+)$ by
applying this lemma pointwise.
\end{proof}

This proposition illustrates the situation well: thanks to the thin
biorthogonality, the $2$-category $\PreThin(A)$ is represented up to
biequivalence as a mere category $\Thin(A)$. The higher dimensional
structure simply vanishes.

\subsubsection{The bicategory $\Thin$}
\label{sec:thin-summarized}
With this in place, we may finally define the components of our bicategory
$\Thin$. Its \emph{objects} are thin groupoids. Its \emph{morphisms} from $A$ to
$B$ are strategies from $A$ to $B$, \emph{i.e.} elements of $\T_{A\lin B}$ --
recall that they are $(S, \display^S : S \to A \times B)$, in particular spans
from $A$ to $B$
\[
A \xot{\display^S_A} S \xto{\display^S_B} B\,.
\]

The \emph{identities} are identity spans, and \emph{composition} is via
the pullback \eqref{eq:pbcomp}. If $S$ and $T$ are strategies from $A$
to $B$, the \emph{$2$-cells} from $S$ to $T$ are the positive morphisms
$(F, \phi) : S \to T$. As $\phi : \display^S \Rightarrow
\display^T \circ F$ is a family of positive morphisms on $A^\perp \parr
B$ with underlying plain groupoid $A \times B$, it may be equivalently
presented as pair of $F^A : \display^S_A \Rightarrow \display^T_A
\circ F$ and $F^B : \display^S_B \Rightarrow \display^T_B \circ F$,
as in \Cref{def:weakmor}. 
For horizontal composition of positive morphisms, we first proceed as in
\Cref{subsubsec:horcomp} and obtain a connected groupoid of (non necessarily
positive) horizontal compositions -- which must all have the same image through
the biequivalence of \Cref{prop:positivization}, providing our unique positive
horizontal composition. Altogether, we have:

\begin{theorem}
  \label{thm:thin-bicat}
Those components form $\Thin$, a bicategory.
\end{theorem}
\begin{proof}
  See details in \Cref{sec:app-thin-bicat}.
\end{proof}

Next, we develop the further structure of $\Thin$.

\section{Cartesian Closed Structure} 
\label{sec:cc_bicat}

To construct a cartesian closed bicategory, we intend to follow
\cite{fiore2008cartesian}. We first turn the construction $\BFam$ --
thereafter denoted by $\oc$ -- into a pseudocomonad, and then
equip the Kleisli bicategory $\Thin_{\oc}$ with the cartesian closed
structure. 

\subsection{The Pseudocomonad}

We first develop the action of $\oc$ on objects of $\Thin$.

\subsubsection{The bang of thin groupoids} First, $\oc$ is defined
on uniform groupoids via $\U_{\oc A} = (\oc \U_A)^{\perp \perp}$, where
we have used
\[
\oc \U_A = \{(\oc S, \oc \display^S) \mid S \in \U_A\}
\]
using the functorial action $\oc \display^S : \oc S \to \oc A$. For thin
groupoids, the positive and negative symmetries of $\oc A$ were defined in
\Cref{subsubsec:pol_subgr}. The thin structure is set as $\T_{\oc A} = (\oc
\T_A)^{\pperp \pperp}$ -- it is a direct verification that this is a thin
groupoid.

\subsubsection{The bang of strategies} If $S \in \T_{A\lin B}$, 
we have $\display^S = \tuple{\display^S_A, \display^S_B}$
for $\display^S_A : S \to A$ and $\display^S_B : S \to B$ -- its bang
is
\[
\oc A \xot{\oc \display^S_A} \oc S \xto{\oc \display^S_B} \oc B
\]
packaged as $(\oc S, \tuple{\oc \display^S_A, \oc \display^S_B})$.
That this is in $\T_{\oc A \lin \oc B}$ relies on:

\begin{lemma}
Consider $A, B$ thin groupoids, and $T$ a prestrategy from $\oc A$ to
$B$. Then, the following two properties hold:
\[
\begin{array}{rl}
\text{\emph{(1)}} &
\text{$T \in \U_{\oc A \lin B}$ iff $(T, \display^T_{\oc A}) \in
\U_{\oc A}^{\perp}$ and}\\
&\text{for all $S \in \U_A$, $T@\oc S \in \U_B$,}\\
\text{\emph{(2)}} &
\text{$T \in \T_{\oc A \lin B}$ iff for all $S \in \T_A$, $T@\oc S \in
\T_B$.}
\end{array}
\]
\end{lemma}

This is an immediate application of \Cref{prop:carac_lin,prop:carac_lin_thin}.
Since $\U_{\oc A} = (\oc \U_A)^{\perp \perp}$ and $\T_{\oc A} = (\oc
\T_A)^{\pperp \pperp}$. From this lemma, it is a rather direct verification that
for any $S \in \T_{A\lin B}$, we have $\oc S \in \T_{\oc A \lin \oc B}$ as
required.

\subsubsection{A pseudofunctor}
Since $\oc$ is a functor, it preserves the identity span on the nose.
Since $\oc$ preserves pullbacks, for $S \in \T_{A\lin B}$ and $T \in
\T_{B \lin C}$, the universal property gives us
\[
m_{S, T} : \oc (T \odot S) \iso \oc T \odot \oc S
\]
a strong invertible $2$-cell in $\Thin$. As expected, this $2$\cell is natural
in $S$ and $T$ (with respect to positive morphisms). Altogether, we obtain a
pseudofunctor $\oc : \Thin \to \Thin$. See \Cref{sec:app-oc-functor} for
details.

\subsubsection{A pseudomonad on groupoids} In fact we first turn $\oc$
into a pseudomonad on $\Gpd$, from which its pseudocomonad structure on
$\Thin$ shall be derived. 
We noted earlier that we have a functor $\BFam : \Gpd \to \Gpd$ -- in
fact, it is extended to a $2$-endofunctor on the $2$-category of small
groupoids, noted 
\[
\oc : \Gpd \to \Gpd\,,
\]
defined on a $2$-cell $\alpha : F \Rightarrow G : A \to B$ as the
natural transformation $\oc\alpha : \oc F \Rightarrow \oc G$ with
components all pairs
\[
(\oc \alpha)_{(A_i)_{i\in I}} = (\id_I, (\alpha_{A_i})_{i\in I}) \in
\oc B((F A_i)_{i\in I}, (G A_i)_{i\in I})\,.
\]

To turn this into a pseudomonad, we must adjoin a multiplication and a
unit. The components of the unit are the functors
\[
\begin{array}{rcrcl}
\eta_A &:& A &\to& \oc A\\
&& a &\mapsto& (a)_{\{0\}} = [0 \cdot a]
\end{array}
\]
with the obvious functorial action. The intuition is that the unit
transports a single resource usage from $A$ to $\oc A$, arriving at a
singleton family. In doing so, it must select a copy index. Any natural
number will do -- the rest of the paper does not depend on this choice
-- but for definiteness and compatibility with the traditional
convention from AJM games, we pick $0$.

For the multiplication $\mu_A : \oc \oc A \to \oc A$, we must flatten a
family of families into a family. For this purpose, we fix an injective
function $\pair{-, -} : \mathbb{N}^2 \to \mathbb{N}$ -- again, the
results of this paper do not depend on that choice. Given
$I \subseteq_f \mathbb{N}$ and a family $(J_i)_{i\in I}$ where $J_i
\subseteq_f \mathbb{N}$ for all $i \in I$, let us write  
\[
\Sigma_{i\in I} J_i = \{\pair{i,j} \mid i \in I,~j \in J_i\}\,,
\]
which is by definition still a finite subset of $\mathbb{N}$. Then we
set
\[
\begin{array}{rcrcl}
\mu_A &:& \oc \oc A & \to & \oc A\\
&&((a_{i, j})_{j \in J_i})_{i \in I} &\mapsto& (a_{i,j})_{\pair{i,j}
\in \Sigma_{i\in I} J_i}
\end{array}
\]
for any groupoid $A$, along with the obvious functorial action. 

Altogether this yields $\eta : \Id_{\Gpd} \Rightarrow \oc$ and $\mu : \oc \oc
\Rightarrow \oc$, two (strict 2-) natural transformations. The monad laws, if
they were to hold on the nose, would mean that $\pair{0, i} = \pair{i, 0} = i$
and $\pair{\pair{i, j}, k} = \pair{i, \pair{j, k}}$ for all $i, j, k \in
\mathbb{N}$; and it is clear that no injection satisfying those laws exists.
Nevertheless, for every groupoid $A$ the coherence laws for a monad hold up to
\emph{natural isomorphisms}: we have $\alpha_A$, $\beta_A$ and $\gamma_A$ as
indicated in \Cref{fig:unitlaw,fig:assoclaw}. For instance, for any $(a_j)_{j
  \in J} \in \oc A$:
\[
(\alpha_A)_{(a_j)_{j\in J}} : (a_j)_{j\in J} \sym_{\oc A}
(a_j)_{\pair{0, j} \in \Sigma_{i \in \{0\}} J}
\]
reindexing along the bijection $J \bij \Sigma_{i \in \{0\}} J$. The
other components act similarly -- note that they are all
\emph{negative} symmetries. The associated families
$(\alpha_A)_{A \in \Gpd}$, $(\beta_A)_{A\in \Gpd}$ and
$(\gamma_A)_{A\in \Gpd}$ satisfy the conditions for
\emph{modifications}, and the additional coherence laws for a
pseudomonad:

\begin{proposition}
The $2$-functor $\oc : \Gpd \to \Gpd$ along with the components above
yield a pseudomonad on $\Gpd$.
\end{proposition}

\begin{figure}
\[
\begin{tikzcd}[cramped,row sep=20pt, column sep=50pt]
\oc \oc A
	\ar[dr,"\mu_A"',myname=mul]&
\oc A	\ar[l,"\eta_{\oc A}"']
	\ar[r, "\oc \eta_A"] 
	\ar[d, "\id_{\oc A}"{description},myname=id]
	\phar[from=mul,"\stackrel{\alpha_A}{\Rightarrow}"]&
\oc \oc A
	\ar[dl,"\mu_A",myname=mur]\\
&\oc A
	\phar[from=mur,to=1-2,"\stackrel{\beta_A}{\Leftarrow}"]
\end{tikzcd}
\]
\caption{Unit natural isomorphisms}
\label{fig:unitlaw}
\end{figure}
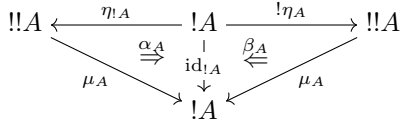
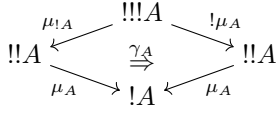
\begin{figure}
\[
\begin{tikzcd}[cramped,row sep=3pt]
& \oc \oc \oc A
	\ar[dl,"\mu_{\oc A}"']
	\ar[dr,"\oc \mu_A"]\\
\oc \oc A
	\ar[dr,"\mu_A"']
	\phar[rr,"\stackrel{\gamma_A}{\Rightarrow}"]&&
\oc \oc A
	\ar[dl,"\mu_A"]\\
&\oc A
\end{tikzcd}
\]
\caption{Associativity natural isomorphism}
\label{fig:assoclaw}
\end{figure}
\subsubsection{Lifting functors to spans} We shall turn $\oc$ into a
pseudocomonad on $\Thin$ by lifting the components above to spans. 
In general, if $F : B \to A$ is a functor, then there is a span $\check
F$
\[
A \xot{F} B \xto{\id_B} B\,,
\]
called the \textbf{lifting} of $F$ -- but we need sufficient conditions on
$F$ for this construction to yield morphisms in $\Thin$. For that
purpose, if $A$ and $B$ are thin groupoids, we say that a functor $F :
A \to B$ is a \textbf{renaming} iff the following conditions hold:
\[
\begin{array}{rl}
\text{\emph{(1)}} &
\text{for all $\theta: a \sym_A a'$, if $\theta$ is positive then so is
$F \theta$,}\\
\text{\emph{(2)}} &
\text{for all $(T,\display^T) \in \U_B^\perp$, $(T, F \circ \display^T)
\in \U_A^\perp$,}\\
\text{\emph{(3)}} &
\text{for all $(T, \display^T) \in \T_B^{\pperp}$, $(T, F \circ
\display^T) \in \T_A^{\pperp}$.}
\end{array}
\]

Clearly, renamings compose -- we consider the $2$-category $\Ren$ whose objects
are thin groupoids, whose morphisms are renamings, and whose $2$-cells are
negative natural transformations. As expected, lifting renamings yields thin
spans (see \Cref{sec:renamings}). Lifting can be extended to $2$-cells: if
$\alpha : F \Rightarrow G : A \to B$ is a negative natural transformation, then
$\check \alpha$ is the positive morphism described by the diagram:
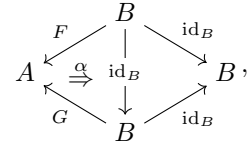
\[
\begin{tikzcd}[cramped,row sep=10pt]
&B	\ar[dl,"F"']
	\ar[dr,"\id_B"]
	\ar[dd,"\id_B"{description},myname=id]\\
A	\phar[to=id,"\stackrel{\alpha}{\Rightarrow}"]&&B\\
&B	\ar[ul,"G"]
	\ar[ur,"\id_B"']
\end{tikzcd}\,,
\]
noting that this is positive as negative $\alpha$ is in contravariant position.
Altogether, we get (see details in \Cref{sec:renamings}):
\begin{proposition}
  \label{prop:check-psfunctor}
  There is a pseudofunctor $\check{-} : \Ren^{\op} \to \Thin$.
\end{proposition}
Here, $\Ren^{\op}$ is $\Ren$ with the morphisms reversed, but the
$2$-cells unchanged. It can be checked that for any thin groupoid $A$,
the functors $\eta_A : A \to \oc A$ and $\mu_A : \oc \oc A \to \oc A$
are renamings, in particular for every thin groupoid $A$ we get
\[
\check{\eta_A} \in \Thin(\oc A, A)
\qquad
\check{\mu_A} \in \Thin(\oc A, \oc \oc A)
\]
the main components to turn $\oc$ into a pseudocomonad. Unlike in
$\Gpd$, the families $\check{\eta}$ and $\check{\mu}$ are not strict
$2$-natural transformations but only pseudonatural transformations,
with $2$-cells
\[
\begin{array}{rcrcl}
\eta_S &:& \check{\eta}_B \odot \oc S &\To& S \odot \check{\eta}_A\\
\mu_S &:& \check{\mu}_B \odot \oc S &\To& \oc \oc S \odot \check{\mu}_A\,,
\end{array}
\]
positive isomorphisms obtained for $S \in \Thin(A, B)$ from the
universal property of pullbacks, via the observation that $\eta :
\Id_{\Gpd} \Rightarrow \oc$ and $\mu : \oc \oc \Rightarrow \oc$ are
\emph{cartesian} natural transormations. It may be checked that
$\eta_S$ and $\mu_S$ are natural in $S$ and satisfy the coherence
conditions of pseudonatural transformations. Finally, the modifications
$\alpha, \beta, \gamma$ involved in the pseudomonad structure of $\oc$
on $\Gpd$ lift to the modifications required for the pseudocomonad
structure of $\oc$ on $\Thin$.

\begin{theorem}
  \label{thm:oc-pseudomonad}
  We have a pseudocomonad $\oc$ on $\Thin$.
\end{theorem}
\begin{proof}
  See details in \Cref{sec:app-oc-pseudocomonad}.
\end{proof}

We move on to studying the Kleisli bicategory $\Thin_{\oc}$ whose
horizontal composition, denoted $\sthcb$, is defined as expected.

\subsection{Cartesian Closed Structure} 

\subsubsection{Finite products} First, we show that $\Thin_{\oc}$ has
finite products, \emph{i.e.} is a \emph{fp-bicategory} in the sense of
Fiore and Saville \cite{DBLP:journals/mscs/FioreS21} -- unlike them, we
work with binary products. 

\paragraph{Terminal object}
Write $\top$ for the empty groupoid, made a thin groupoid with
$\U_\top = \T_\top = \{\id_\emptyset\}$. For any thin groupoid
$A$, $\Thin_{\oc}(A, \top)$ has exactly one element -- the empty groupoid.
Thus, $\Thin_{\oc}$ has a (strict) terminal object.

\paragraph{Binary product} 
If $A$ and $B$ are thin groupoids, then the \textbf{with} $A \with B$
has underlying groupoid $A + B$ the disjoint union, with $(A + B)_- =
A_- + B_-$ and $(A + B)_+ = A_+ + B_+$. We adjoin $\U_{A\with B} =
(\U_A + \U_B)^{\perp \perp}$ and $\T_{A\with B} = (\T_A + \T_B)^{\pperp
\pperp}$, where as usual, $\U_A + \U_B$ comprises the set of all $(S +
T, \display^S + \display^T)$ for $(S, \display^S) \in \U_A$ and $(T,
\display^T) \in \U_B$, using the functorial action of $+$ 
(and likewise for $\T_A + \T_B$).

\paragraph{Pairing and projections} The \textbf{projections} are simply set as
$\stbpl = \widecheck{(\eta_{A+B}\circ \cplprot)} \in \Thin_{\oc}(A\with B, A)$
and $\stbpr = \widecheck{(\eta_{A+B}\circ \cprprot)} \in \Thin_{\oc}(A \with B,
B)$ for $\cpl\co A \to A+B$ and $\cpr\co B \to A+B$ the obvious
coprojections/renamings. The \textbf{pairing} of $S \in \Thin_{\oc}(\Gamma, A)$
and $T \in \Thin_{\oc}(\Gamma, B)$ is
\[
  (S + T, \display_{\oc \Gamma} : S + T \to \oc \Gamma,
  \display_{A \with B} : S + T \to A + B)
\]
with $\display_{\oc \Gamma}$ the co-pairing
and $\display_{A\with B} = \display^S_A + \display^T_B$. We have:
\begin{proposition}
  \label{prop:thinb-fp}
  For any thin groupoids $\Gamma, A$ and $B$, there is
  \[
    \begin{tikzcd}[cramped,row sep=0pt,csbo=7.5em]
      \Thin_{\oc}(\Gamma,A \with B)
      \ar[rr,bend left=15,"{(\stbpl \sthcb -, \stbpr\sthcb - )}"]
      &
      \bot
      &
      \Thin_{\oc}(\Gamma,A)
      \times
      \Thin_{\oc}(\Gamma,B)
      \ar[ll,bend left=15,"\tuple{-, -}"]
    \end{tikzcd}
  \]
  an adjoint equivalence.
\end{proposition}
\begin{proof}
  If $S \in \Thin_{\oc}(\Gamma, A)$ and $T \in \Thin_{\oc}(\Gamma,
  B)$ there are
  \[
    \omega^A_{S, T} : \stbpl \sthcb \tuple{S, T} \iso S
    \qquad
    \omega^B_{S, T} : \stbpr \sthcb \tuple{S, T} \iso T
  \]
  positive isos, and for $U \in \Thin_{\oc}(\Gamma, A \with B)$ there is
  \[
    \bar{\omega}_U : U \iso \tuple{\stbpl \sthcb U, \stbpr \sthcb T}
  \]
  a positive iso, defined in the obvious way.
  Those are all natural in $S, T, U$, and satisfy the required triangle
  identities. 
\end{proof}

See \Cref{sec:cartesian-product} for more details. Altogether,
this establishes that $\Thin_{\oc}$ is a \emph{fp-bicategory} in the
sense of \cite{DBLP:journals/mscs/FioreS21}.

\subsubsection{Cartesian closure} If $A$ and $B$ are thin groupoids,
then we set $A\tto B = \oc A \parr B$. Before we describe the
additional components, we must observe the Seely \emph{equivalence}:
\[
\begin{tikzcd}
\oc A \tensor \oc B
	\ar[rr, bend left=10, "s_{A, B}"]&&
\oc (A \with B)
	\ar[ll, bend left=10, "\bar{s}_{A, B}"]
\end{tikzcd}
\]
where $s_{A, B}$ sends $(a_i)_{i\in I}, (b_j)_{j\in J}$ to $(c_k)_{k\in I\imix
  J}$, with $I\imix J = \ibij(I \sqcup J)$ for some chosen bijection $\ibij =
\coprodfact{\ibij_l}{\ibij_r}$ between $\N\sqcup \N$ and $\N$, and where
$c_{\ibij_l(i)} = a_i$ and $c_{\ibij_r(j)} = b_j$; and $\bar{s}_{A, B}$ sends
$(c_k)_{k\in K}$ to $(a_i)_{i\in I}, (b_j)_{j\in J}$ where $I \subseteq K$ is
the subset of those $i\in K$ such that $c_i = a_i \in A$, and likewise for
$b_j$. Both functors are renamings, and the isomorphisms witnessing the
equivalence are negative.

Via the Seely equivalence, we first define the \textbf{evaluation} as
the span with basic groupoid $\oc A \times B$, with left leg the
functor
\[
\scalebox{.89}{$
\oc A \times B \to (\oc A \times B) \times \oc A \to \oc (\oc A \times
B) \times \oc A \xto{s_{A, B}} \oc ((A \tto B) \with A)
$}
\]
and right leg the projection $\oc A \times B \to B$. This yields a thin
span $\evm_{A, B} \in \Thin_{\oc}((A \tto B)\with A, B)$. Now, we need
\[
\Lambda(-) : \Thin_{\oc}(\Gamma \with A, B) \to \Thin_{\oc}(\Gamma, A
\tto B)
\]
the \textbf{currying} functor: given $S \in \Thin_{\oc}(\Gamma \with A,
B)$, its currying $\Lambda(S)$ is simply $S$, with display map post-composed with
\[
\oc (\Gamma + A) \times B \stackrel{\bar{s}_{\Gamma, A}}{\simeq} (\oc
\Gamma \times \oc A) \times B 
\iso \oc \Gamma \times (\oc A \times B)\,.
\]
With this data in place, we may finally prove:
\begin{proposition}
  \label{prop:adj-closure}
  For any groupoids $\Gamma, A, B$, there is
  \[
    \begin{tikzcd}
      \Thin_{\oc}(\Gamma, A\tto B)
      \ar[rr, bend left=15, "\evm_{A, B} \odot_{\oc} (- \with A)"]&\bot&
      \Thin_{\oc}(\Gamma \with A, B)
      \ar[ll, bend left=15, "\Lambda(-)"]
    \end{tikzcd}
  \]
  an adjoint equivalence.
\end{proposition}
\begin{proof}
  One can first show the existence of adjoint equivalence between the currying
  operation $\stcurry-$, and a symmetric uncurrying operation $\stuncurry-$. The
  unit and counit of this adjunction can be derived from the ones of the Seely
  (adjoint) equivalence. One can then prove that $\stuncurry-$ is in fact
  isomorphic to $\evm_{A, B} \odot_{\oc} (- \with A)$ in order to get the wanted
  equivalence.
\end{proof}

See \Cref{sec:app-ev-adj} for details. Altogether, we have:

\begin{theorem}
We have $\Thin_{\oc}$, a cartesian closed bicategory.
\end{theorem}

This entails that we can interpret types of the simply-typed
$\lambda$-calculus as thin groupoids, morphisms as thin spans and
rewrites between terms as certain positive isomorphisms \cite{DBLP:conf/lics/FioreS19}.

\section{Conclusion}

This paper focuses on the construction of $\Thin_{\oc}$, leaving for
later its application to semantics of $\lambda$-calculi and programming
languages. We believe this opens multiple perspectives for further
research: firstly, we may explore the obtained interpretation of the
$\lambda$-calculus, which syntactically should correspond to the
sequence typing system of Vial \cite{DBLP:conf/lics/Vial17} and to the
non-uniform $\lambda$-calculus of Melliès
\cite{DBLP:journals/tcs/Mellies06}. We should explore links with other
models of the literature, notably with the weighted relational model
recasting ideas from \cite{DBLP:journals/corr/abs-2107-03155}, and with
generalized species of structures and template games. Another related
direction consists in accommodating another feature of template games,
the mechanism to capture scheduling and synchronization
\cite{DBLP:journals/pacmpl/Mellies19}, into thin spans.

In more semantic directions, we believe that with respect to generalized
species of structures, the fact that operations on thin spans involve
no quotient may be helpful in two ways: \emph{(1)} individuals may be
ordered concretely, and the model should support continuous reasoning
allowing one to deal easily with infinite computation; and
\emph{(2)} adding ``typed'' weights coming from an SMCC as in
\cite{DBLP:conf/lics/TsukadaAO18} should be a lot simpler, since those
weigths no longer have to themselves be saturated.


\ifblind
\else
\section*{Acknowledgment}
Work supported by the ANR projects DyVerSe (ANR-19-CE48-0010-01) and PPS
(ANR-19-CE48-0014); and by the Labex MiLyon (ANR-10-LABX-0070) of Universit\'e
de Lyon, within the program ``Investissements d'Avenir'' (ANR-11-IDEX-0007),
operated by the French National Research Agency (ANR).

\fi



\bibliographystyle{IEEEtran}
\bibliography{main}
%

\newpage
\appendix
\input{appendix}

\end{document}

%% file: macros.tex
\newtheorem{example}{Example}
\newtheorem{theorem}{Theorem}
\newtheorem{definition}{Definition}
\newtheorem{lemma}{Lemma}
\newtheorem{proposition}{Proposition}

\newcommand{\PCF}{\mathsf{PCF}}

\newcommand{\finv}[1]{#1^{{-}1}}
\newcommand{\set}[1]{\{#1\}}

\newcommand{\intr}[1]{\llbracket #1 \rrbracket}
\newcommand{\lin}{\multimap}
\let\linto\lin
\newcommand{\Mf}{\mathcal{M}}
\newcommand{\iso}{\cong}
\newcommand{\bij}{\simeq}
\newcommand{\R}{\mathcal{R}}
\newcommand{\rp}{\overline{\mathbb{R}}_+}
\newcommand{\C}{\mathcal{C}}
\newcommand{\display}{\partial}
\newcommand\stdisp[2][]{\display^{#2}\ifempty{#1}{}{_{#1}}}
\newcommand{\wit}{\mathsf{wit}}
\newcommand{\Perm}[1]{\mathcal{S}_{#1}}

\newcommand{\Id}{\mathrm{id}} 
\newcommand{\op}{\mathrm{op}}
\newcommand{\cocat}{{\mathrm{co}}}
\newcommand{\ot}{\leftarrow}
\newcommand{\id}{\mathrm{id}}
\newcommand{\swap}{\mathrm{swap}}
\newcommand{\Fam}{\mathrm{Fam}}
\newcommand{\N}{\mathbb{N}}
\newcommand{\PreStrat}{\mathsf{PreStrat}}
\newcommand{\bS}{\mathbf{S}}

\newcommand{\U}{\mathbf{U}}

\newcommand{\T}{\mathbf{T}}
\newcommand\stsetb[1]{\T_{#1}}
\newcommand{\tensor}{\otimes}
\newcommand{\Unif}{\mathbf{Unif}}
\newcommand{\Thin}{\mathbf{Thin}}
\newcommand{\Thinb}{\mathbf{Thin}_{\oc}}
\newcommand{\PreThin}{\mathbf{PreThin}}
\newcommand{\bder}{\mathbf{der}}
\newcommand{\der}{\mathrm{der}}
\newcommand{\tuple}[1]{\langle #1 \rangle}
\newcommand{\sym}{\cong}
\newcommand{\pair}[1]{\tuple{#1}}
\newcommand\iprod[2]{\lbrack #1,#2 \rbrack}
\newcommand{\tto}{\Rightarrow}
\let\cptto\tto 
\newcommand{\cpl}{{\operatorname{\mathit{\bar l}}}}
\newbox\privcpl
\setbox\privcpl\hbox{$\bar l$}
\newcommand{\cplprot}{{\copy\privcpl}}
\newcommand{\cpr}{{\operatorname{\mathit{\bar r}}}}
\newbox\privcpr
\setbox\privcpr\hbox{$\bar r$}
\newcommand{\cprprot}{{\copy\privcpr}}
\newcommand\prodfact[2]{\tuple{#1,#2}}

\newcommand\coprodfact[2]{\lbrack #1,#2 \rbrack}
\newcommand\catprodfact[2]{( #1,#2 )}
\newcommand\catprodfactp[1]{( #1 )}
\newcommand\gpdprodfact[2]{( #1,#2 )}
\newcommand\gpdprodfactp[1]{( #1 )}
\DeclareFontEncoding{LS1}{}{}
\DeclareFontSubstitution{LS1}{stix}{m}{n}
\DeclareSymbolFont{stixarrows}{LS1}{stixsf}{m}{n}
\DeclareMathSymbol{\nvrightarrow}{\mathrel}{stixarrows}{"F6}
\newcommand\stto{\nvrightarrow}
\DeclareSymbolFont{stiximport}{LS1}{stix}{m}{it}
\DeclareMathAccent{\widecheck}{\mathord}{stiximport}{"9C}
\newcommand\stid[1]{\operatorname{cc}^{#1}}
\newcommand{\evm}{\mathbf{ev}}
\newcommand{\dual}[1]{#1^\star}
\newcommand\pl{\mathop{\mathit l}\nolimits}
\newcommand\pr{\mathop{\mathit r}\nolimits}

\newcommand\stbpl{\mathop{\mathit L_\oc}\nolimits}
\newcommand\stbpr{\mathop{\mathit R_\oc}\nolimits}
\newcommand\sthc{\mathbin{\odot}}
\newcommand\pshccoh[2]{m_{#1,#2}} 
\newcommand\pshccohe{m} 
\let\sthcomp\sthc
\newcommand\sthcb{\mathbin{\odot}_{\oc}}
\newcommand\termcat{\mathrm{1}}
\newcommand\unit[1]{\id_{#1}}
\newcommand\seely{s}
\newcommand\seelyinv{\bar{s}}
\newcommand\seelycoh{\operatorname{See}}
\newcommand\ocd{\oc\oc}
\newcommand\stcurry[1]{\Lambda(#1)}
\newcommand\stuncurry[1]{\bar\Lambda(#1)}
\newcommand\ocass{\gamma}
\newcommand\oclu{\alpha}
\newcommand\ocru{\beta}
\newcommand\psu{\check\eta}%
\newcommand\psm{\check\mu}%
\newcommand{\tbool}{\mathbb{B}}

\newcommand{\ttrue}{\mathbf{t\!t}}
\newcommand{\tfalse}{\mathbf{f\!f}}
\newcommand{\choice}{\mathbf{choice}}

\newcommand{\Rel}{\mathbf{Rel}}
\newcommand{\Set}{\mathbf{Set}}
\newcommand{\Cat}{\mathbf{Cat}}
\newcommand{\Gpd}{\mathbf{Gpd}}
\newcommand{\Span}{\mathbf{Span}}
\newcommand{\Sym}{\mathbf{Sym}}
\newcommand{\BFam}{\mathbf{Fam}}
\renewcommand{\S}{\mathbb{S}}
\newcommand{\SU}{\mathbf{S}}
\newcommand{\Dist}{\mathbf{Dist}}
\newcommand{\Esp}{\mathbf{Esp}}
\newcommand{\Ren}{\mathbf{Ren}}

\newcommand\Bicat{\mathbf{Bicat}}
\mathchardef\mhyphen="2D
\newcommand\pmFunct{\pm\mhyphen\mathbf{Funct}}

\newcommand\stsupp[1]{#1}
\newcommand\stsuppreally[1]{\hbox{\ul{$#1$}}} 

\newcommand\stlu{\operatorname{L}}
\newcommand\stru{\operatorname{R}}
\newcommand\stass{\operatorname{A}}

\newcommand\prtosp[1]{\check{\ifdash{#1}{(-)}{{#1}}}}

\newcommand\prtospw[1]{\widecheck{\ifdash{#1}{(-)}{{#1}}}}
\newcommand\prtospcohe{m}
\newcommand\prtospcoh[2]{m_{#1,#2}}

\newcommand\ibij{\varpi}
\newcommand\isum{{\textstyle\sum}}

\newcommand\imix{\mathbin{\bowtie}}

\tikzcdset{%
    myname/.code args={#1}{%
      \tikzcdset{""{auto=false,name=#1}}%
    }
}
\makeatletter
\pgfarrowsdeclare{tripleimplies}{tripleimplies}
{
  \pgfmathsetlength{\pgfutil@tempdima}{.25\pgflinewidth+.25*\pgfinnerlinewidth}%
  \pgfmathsetlength{\pgfutil@tempdimb}{.5\pgflinewidth-.5*\pgfinnerlinewidth}%
  \pgfarrowsrightextend{2.06\pgfutil@tempdima+.5\pgfutil@tempdimb}
  \pgfarrowsleftextend{-1.36\pgfutil@tempdima-.5\pgfutil@tempdimb}
}
{
  \pgfmathsetlength{\pgfutil@tempdima}{.25\pgflinewidth+.25*\pgfinnerlinewidth}%
  \pgfmathsetlength{\pgfutil@tempdimb}{.5\pgflinewidth-.5*\pgfinnerlinewidth}%
  \pgftransformxshift{.06\pgfutil@tempdima}
  \pgfsetlinewidth{\pgfutil@tempdimb}
  \pgfsetdash{}{+0pt}
  \pgfsetroundcap
  \pgfsetroundjoin
  \pgfpathmoveto{\pgfpoint{-1.4\pgfutil@tempdima}{2.65\pgfutil@tempdima}}
  \pgfpathcurveto
  {\pgfpoint{-0.75\pgfutil@tempdima}{1.25\pgfutil@tempdima}}
  {\pgfpoint{1\pgfutil@tempdima}{0.05\pgfutil@tempdima}}
  {\pgfpoint{2\pgfutil@tempdima}{0pt}}
  \pgfpathcurveto
  {\pgfpoint{1\pgfutil@tempdima}{-0.05\pgfutil@tempdima}}
  {\pgfpoint{-.75\pgfutil@tempdima}{-1.25\pgfutil@tempdima}}
  {\pgfpoint{-1.4\pgfutil@tempdima}{-2.65\pgfutil@tempdima}}
  \pgfusepathqstroke
  \pgfpathmoveto{\pgfpoint{0\pgfutil@tempdima}{0\pgfutil@tempdima}}
  \pgfpathlineto{\pgfpoint{2\pgfutil@tempdima}{0\pgfutil@tempdima}}
  \pgfusepathqstroke
}%
\makeatother
\tikzset{tikztriple/.style={preaction={draw,tikzcd
cap-tripleimplies,shorten
      >=0pt,double distance=3pt},tikzcd cap-round cap,shorten >=3.5pt}}
\tikzcdset{%
  Rrightarrow/.code={\tikzset{tikztriple}}
}
\tikzcdset{%
  cs/.code args={#1}{%
    \tikzcdset{column sep={#1}}%
  }
}
\tikzcdset{%
  rs/.code args={#1}{%
    \tikzcdset{row sep={#1}}%
  }
}
\makeatletter
\def\tikzcd@sep@bo#1#2{%
  \pgfkeysifdefined{/tikz/commutative diagrams/#1 sep/#2}%
  {\pgfkeysgetvalue{/tikz/commutative diagrams/#1 sep/#2}\tikzcd@temp%
    \edef\tikzcd@temp{{\tikzcd@temp,between origins}}%
    \pgfkeysalso{/tikz/#1 sep/.expand once=\tikzcd@temp}}%
  {\pgfkeysalso{/tikz/#1 sep={#2,between origins}}}}
\tikzcdset{%
  rsbo/.code={\tikzcd@sep@bo{row}{#1}}}
\tikzcdset{%
  csbo/.code={\tikzcd@sep@bo{column}{#1}}}
\tikzcdset{%
  sbo/.code={\tikzcd@sep@bo{row}{#1}\tikzcd@sep@bo{column}{#1}}}
\makeatother
\newcommand*{\tarrow}[1][]{%
  \ar[tikzcd cap-tripleimplies,shorten >=0pt,double
  distance=3pt,postaction={draw,tikzcd cap-round cap,shorten
>=0.2pt},#1]
}
\newcommand{\tar}{\tarrow}

\newcommand\phar[1][]{\ar[phantom,#1]}
\newcommand\cphar[1][]{\ar[phantom,start anchor=center,end anchor=center,#1]}
\newcommand\tikzcdin[2][column sep=small]{
  \begin{tikzcd}[cramped,ampersand replacement=\&,#1]
    #2 
  \end{tikzcd}%
}
\newcommand\drcorner{\rotatebox[origin=c]{180}{$\ulcorner$}}
\newcommand\dlcorner{\rotatebox[origin=c]{180}{$\urcorner$}}

\newcommand{\xto}{\xrightarrow}
\newcommand{\xot}{\xleftarrow}
\newcommand{\To}{\Rightarrow}
\newcommand{\xTo}{\xRightarrow}

\newcommand{\TO}{\Rrightarrow}
\newcommand{\xTO}{\xRrightarrow}
\makeatletter
\newcommand{\xRrightarrow}[2][]{\ext@arrow 0359\Rrightarrowfill@{#1}{#2}}
\newcommand{\Rrightarrowfill@}{\arrowfill@\equiv\equiv\Rrightarrow}
\makeatother

\makeatletter
\newsavebox\my@boxdcorner
\savebox\my@boxdcorner{\begin{tikzpicture}[rotate=45,x=3.85pt,y=3.85pt]
  \draw[line cap=round,line join=round] (0,1) -- (0,0) -- (1,0);
\end{tikzpicture}}
\newcommand\dcorner{\usebox\my@boxdcorner}
\newsavebox\my@boxucorner
\savebox\my@boxucorner{\begin{tikzpicture}[rotate=-135,x=3.85pt,y=3.85pt]
    \draw[line cap=round,line join=round] (0,1) -- (0,0) -- (1,0);
  \end{tikzpicture}}
\newcommand\ucorner{\usebox\my@boxucorner}
\newbox\my@privdorthon
\setbox\my@privdorthon\hbox to 9pt{$\perp$\hss$\perp$}
\newbox\my@privdorthos
\setbox\my@privdorthos\hbox to 7pt{$\scriptstyle\perp$\hss$\scriptstyle\perp$}
\newbox\my@privdorthoss
\setbox\my@privdorthoss\hbox to 6pt{$\scriptscriptstyle\perp$\hss$\scriptscriptstyle\perp$}
\newbox\my@privdorthod
\setbox\my@privdorthod\hbox to 9pt{$\displaystyle\perp$\hss$\displaystyle\perp$}
\newcommand\storthob{\ensuremath{\mathbin{\mathchoice{\copy\my@privdorthod}{\copy\my@privdorthon}{\copy\my@privdorthos}{\copy\my@privdorthoss}}}}
\newcommand{\pperp}{\storthob}
\makeatother

\newcommand\co\colon
\newcommand\cA{A}
\newcommand\cB{B}
\newcommand\cC{C}
\newcommand\cD{D}
\newcommand\ups{\upsilon}

\newcommand{\ie}{i.e.,\xspace}
\newcommand{\resp}{resp.\xspace}
\newcommand{\wrt}{w.r.t\@.\xspace}

\crefname{enumi}{}{}
\Crefname{enumi}{}{}
\crefname{lemma}{Lemma}{Lemmas}
\Crefname{lemma}{Lemma}{Lemmas}
\newcommand\PCC{\hyperref[prop:pcc]{PCC}\xspace}

\newcommand\nbd{\protect\nobreakdash}
\newcommand\sethyp[2]{%
  \expandafter\def\csname nbd#1\endcsname{\nbd-#2\xspace}%
  \expandafter\ifx\csname #1\endcsname\relax
  \expandafter\def\csname #1\endcsname{\nbd-#2\xspace}%
  \fi
}
\sethyp{cell}{cell}
\sethyp{cells}{cells}
\sethyp{category}{ca\-te\-gory}
\sethyp{categories}{ca\-te\-gories}
\sethyp{categorical}{ca\-te\-go\-rical}
\sethyp{pullback}{pull\-back}
\sethyp{pullbacks}{pull\-backs}
\sethyp{groupoid}{grou\-poid}
\sethyp{groupoids}{grou\-poids}
\sethyp{composition}{com\-po\-si\-tion}
\sethyp{globular}{glo\-bu\-lar}
\sethyp{functor}{fun\-ctor}
\sethyp{functors}{fun\-ctors}
\sethyp{composable}{com\-po\-sa\-ble}
\sethyp{isomorphism}{iso\-mor\-phism}
\let\natural\undefined
\sethyp{natural}{na\-tu\-ral}
\sethyp{naturality}{na\-tu\-ra\-li\-ty}
\sethyp{dimensional}{di\-men\-sio\-nal}
\sethyp{transformation}{trans\-for\-ma\-tion}
\sethyp{transformations}{trans\-for\-ma\-tions}
\sethyp{modification}{mo\-di\-fi\-ca\-tion}
\sethyp{modifications}{mo\-di\-fi\-ca\-tions}
\DeclareRobustCommand\pmfunctor{$\pm$\functor}
\DeclareRobustCommand\pmfunctors{$\pm$\functors}
\DeclareRobustCommand\pmtransformation{$\pm$\transformation}
\DeclareRobustCommand\pmtransformations{$\pm$\transformations}
\DeclareRobustCommand\pmmodification{$\pm$\modification}
\DeclareRobustCommand\pmmodifications{$\pm$\modifications}

\makeatletter
\@ifundefined{ifdebugmode}{\let\ifdebugmode\iffalse}{}
\makeatother
\def\ifparam#1#2#3{\csname if#1\endcsname #2\else #3\fi}
\newcommand{\activatecomments}{\let\ifdebugmode\iftrue}
\newcommand{\TODO}[1]{\ifparam{debugmode}{\marginpar{\tiny #1}}{}}
\newcommand{\todo}[1]{\ifparam{debugmode}{\marginpar{\tiny #1}}{}}

\newcommand{\simon}[1]{\ifparam{debugmode}{\TODO{\color{blue} SF: #1}}{}}
\newcommand{\pierre}[1]{\ifparam{debugmode}{\TODO{\color{red} PC: #1}}{}}

\newcommand\defeq{\mathrel{\hat=}}

\newcommand{\qeq}{\quad=\quad}
\newcommand{\qqeq}{\qquad=\qquad}
\newcommand{\qand}{\quad\text{and}\quad}
\newcommand{\qqand}{\qquad\text{and}\qquad}
\newcommand\eq\enquote
\newcommand\zbox[1]{\makebox[0pt][l]{#1}}
\makeatletter
\DeclareRobustCommand\ifdash[3]{%
  \def\ifdash@o{-}%
  \def\ifdash@t{#1}%
  \def\ifdash@yes{#2}%
  \def\ifdash@no{#3}%
  \ifx\ifdash@o\ifdash@t
  \ifdash@yes
  \else
  \ifdash@no
  \fi
}
\makeatother
\newcommand\ifempty[3]{%
  \def\ifemptytemp{#1}%
  \ifx\ifemptytemp\empty%
  #2%
  \else%
  #3%
  \fi%
}

%% file: appendix.tex
In the following, we will type spans $A \xot{u} S \xto{v} B$ between groupoids
or thin groupoids $A$ and $B$ as $S \co A \stto B$.

\subsection{Bipullbacks}
\label{sec:bipullbacks}

In order to complete the equational definition of bipullbacks of
\Cref{def:bipullback}, it is useful to consider the intensional definition
first: given a $2$\category $C$ with invertible $2$\cells (\ie a
$(2,1)$\category) and a cospan $S \xto{u} B \xot{v} T$ in $C$, a bipullback is a
pseudocone $(P,l,r,\mu)$ as in \Cref{fig:diag1} such that, for every $X \in C$,
the precomposition of the pseudocone $(P,l,r,\mu)$ by morphisms $X \to P$
induces an equivalence of categories between $C(X,P)$ and the category of
pseudocones over the cospan $S \xto{u} B \xot{v} T$ and of vertex $X$ and
pseudocone morphisms. The essential surjectiveness of this precomposition
corresponds exactly to the condition \Cref{def:bipullback:existence} of
\Cref{def:bipullback}. Its full faithfulness can be expressed as the following
condition:
\begin{enumerate}[label=(b'),ref=(b'),left=0pt .. 17pt]
\item \label{def:bipullback:uniquenessbis} given a $2$\cell equality
  \[
    \begin{tikzcd}
      &
      X
      \ar[d,"h"{description}]
      &
      \\
      &
      P
      \ar[dl,"l"']
      \ar[dr,"r"]
      \\
      S
      \ar[dr,"u"']
      \phar[rr,"\xTo{\mu}"]&&
      T       \ar[dl,"v"]\\
      &B 	
    \end{tikzcd}
    \qeq
    \begin{tikzcd}
      & X	\ar[ddl,bend right=30, "l \circ h"',myname=lp]
      \ar[ddr,bend left=30, "r \circ h",myname=rp]
      \ar[d,"h'"{description}]\\
      & P     \ar[dl,"l"']
      \ar[dr,"r"]
      \phar[to=rp,"\xTo{\beta}"]
      \phar[from=lp,"\xTo{\alpha}"]\\
      S       \ar[dr,"u"']
      \phar[rr,"\xTo{\mu}"]&&
      T       \ar[dl,"v"]\\
      &B 	
    \end{tikzcd}
  \]
  for some $h,h' \co X \to P$ and $2$\cells $\alpha \co l \circ h \To l \circ
  h'$ and $\beta \co r \circ h' \To r \circ h$, there is a unique $\theta \co h
  \To h'$ such that $\alpha = l \theta$ and $\beta = r \finv\theta$.
\end{enumerate}
It is not too difficult to show that the latter property is equivalent to the
one asserting that, given two decompositions of a pseudocone $\nu$
\[
  \adjustbox{scale=.8}{%
    \begin{tikzcd}[csbo=3em]
      & X     \ar[ddl,bend right=30, "l'"',myname=lp]
      \ar[ddr,bend left=30, "r'",myname=rp]
      &
      \\\\
      S       \ar[dr,"u"']
      \phar[rr,"\xTo{\nu}"]&&
      T       \ar[dl,"v"]\\
      &B      
    \end{tikzcd}}
  \qeq
  \adjustbox{scale=.8}{%
    \begin{tikzcd}[csbo=3em]
      & X	\ar[ddl,bend right=30, "l'"',myname=lp]
      \ar[ddr,bend left=30, "r'",myname=rp]
      \ar[d,"h"{description}]
      &
      \\
      & P     \ar[dl,"l"']
      \ar[dr,"r"]
      \phar[to=rp,"\xTo{\beta}"]
      \phar[from=lp,"\xTo{\alpha}"]\\
      S       \ar[dr,"u"']
      \phar[rr,"\xTo{\mu}"]&&
      T       \ar[dl,"v"]\\
      &B 	
    \end{tikzcd}}
  \qeq
  \adjustbox{scale=.8}{%
    \begin{tikzcd}[csbo=3em]
      & X	\ar[ddl,bend right=30, "l'"',myname=lp]
      \ar[ddr,bend left=30, "r'",myname=rp]
      \ar[d,"h'"{description}]
      &
      \\
      & P     \ar[dl,"l"']
      \ar[dr,"r"]
      \phar[to=rp,"\xTo{\beta'}"]
      \phar[from=lp,"\xTo{\alpha'}"]\\
      S       \ar[dr,"u"']
      \phar[rr,"\xTo{\mu}"]&&
      T       \ar[dl,"v"]\\
      &B 	
    \end{tikzcd}}
\]
there exists a unique $\theta \co h \To h'$ such that $l \theta = \alpha' \circ
\finv\alpha$ and $r \theta = \finv{\beta'} \circ \beta$, or equivalently
$\alpha' = (l\theta) \circ \alpha$ and $\beta' = \beta \circ (r \finv\theta)$,
which is the complete form of the condition \Cref{def:bipullback:uniqueness}.

\subsection{Pullbacks in $\Gpd$}
\label{sec:pbs-in-gpd}

It happens that pullbacks in $\Gpd$ are well-behaved w.r.t $2$\cells:
\begin{proposition}
  \label{prop:gpd-1pb-is-2pb}
  A pullback of a cospan $S \xto{u} B \xot{v} T$ in $\Gpd$ is a strict $2$\pullback,
  that is, also admits a universal factorization property w.r.t. morphisms of cones.
\end{proposition}
\begin{proof}
  Let $I$ be the groupoid consisting of a walking isomorphism $u$ between two
  objects $0$ and $1$. Given two functors $F$ and $G$ between two groupoids $C$
  and $D$, a $2$\cell
  \[
    \alpha \co F \To G \co C \to D \in \Gpd
  \]
  is then exactly the data of a functor $H \co I \times C \to D$ such that
  $H(0,-) = F$ and $H(1,-) = G$. Using this correspondence, the property that a
  pullback $(P,l,r)$ over the cospan is a $2$\pullback easily reduces to the one
  that $(P,l,r)$ is a $1$\pullback.
\end{proof}
Note that a pullback is a cone, which is in particular a pseudocone with
identity as inner $2$\cell. One might then ask when a pullback is a bipullback,
in which case we have the following characterization:
\begin{proposition}
  \label{prop:charact-pb-bpb-ess-surj}
  Let $(P,l,r)$ be a pullback over a cospan $S \xto{u} B \xot{v} T$ in $\Gpd$.
  Then the pseudocone induced by $(P,l,r)$ is a bipullback if and only if it
  satisfies the condition \Cref{def:bipullback:existence} of
  \Cref{def:bipullback}.
\end{proposition}
\begin{proof}
  By \Cref{prop:gpd-1pb-is-2pb}, $(P,l,r)$ is a $2$\pullback, so that it
  satisfies a universal property w.r.t. cone morphisms. This condition is in
  fact exactly \Cref{def:bipullback:uniquenessbis} which is equivalent to
  \Cref{def:bipullback:uniqueness}. Thus, only \Cref{def:bipullback:existence}
  is left to check for $(P,l,r)$ to be a bipullback. 
\end{proof}
We can then refine the previous proposition into a \eq{pointwise}
characterization in the form of already stated \Cref{lem:carac_pb_bipb} for
which we now provide a proof.
\begin{proof}[Proof of \Cref{lem:carac_pb_bipb}]
  The implication is immediate, since the data of an isomorphism $\theta \co
  f(s) \to g(t)$ is equivalent to the one of a pseudocone on \tikzcdin{S
    \ar[r,"f"] \& B \& T \ar[l,"g"']} of vertex the terminal groupoid, whose
  factorization in the form of the condition of
  \Cref{prop:charact-pb-bpb-ess-surj} is exactly the wanted property.

  We now show the converse property. So let
  \[
    \begin{tikzcd}
      &
      Z
      \ar[ld,"h"',bend right=20]
      \ar[rd,"k",bend left=20]
      &
      \\[1em]
      S
      \ar[rd,"f"']
      \cphar[rr,"\xTo\theta"{yshift=1em}]
      &&
      T
      \ar[ld,"g"]
      \\
      &
      B
    \end{tikzcd}
  \]
  be a pseudocone. By hypothesis, for every $z\in Z$, there exist $s_z \in S$,
  $\phi_z \co h(z) \to s_z$, $t_z \in T$, $\psi_z \co t_z \to k(z)$ such that
  $f(s_z) = g(t_z)$ and $\theta_z = g(\psi_z) \circ f(\phi_z)$. The
  collection of isomorphisms $(\phi_z)_{z \in Z}$ induces a functor $h'$ defined
  by $h'(z) = s_z$ for $z \in Z$, and $h'(w) = \phi_{z'} \circ h(w) \circ
  \finv \phi_z$ for $w \co z \to z' \in Z$. Similarly, we get a functor $k'$
  defined by $h'(z) = t_z$ for $z \in Z$, and $k'(w) = \finv \psi_{z'} \circ
  k(w) \circ \psi_z$ for $w \co z \to z' \in Z$. Given $w \co z \to z' \in Z$, we
  check that $f \circ h'(w)$ and $g \circ k'(w)$ are equal. Since there
  are only isomorphisms involved, it is enough to check that the equality holds
  when in the context $g(\psi_{z'}) \circ (-) \circ f(\phi_z)$:
  \begin{align*}
    g(\psi_{z'}) \circ f (h'(w))\circ f(\phi_z) 
    &=
      g(\psi_{z'}) \circ f(h'(w) \circ \phi_z) 
    \\
    &=
      g(\psi_{z'}) \circ f(\phi_{z'} \circ h(w)) 
    \\
    &=
      g(\psi_{z'}) \circ f(\phi_{z'}) \circ f(h(w)) 
    \\
    &=
      \theta_{z'} \circ f(h(w)) 
    \\
    &=
      g(k(w)) \circ \theta_{z} 
    \\
    &=
      g(k(w)) \circ  g(\psi_{z}) \circ f(\phi_{z})
    \\
    &=
      g(\psi_{z'}) \circ  g(k'(w)) \circ f(\phi_{z})\zbox.
  \end{align*}
  Thus, $(Z,h',k')$ is a cone on \tikzcdin{S \ar[r,"f"] \& B \& T \ar[l,"g"']},
  so there exists $m \co Z \to P$ which factors $h'$ and $k'$ through $\pl$ and
  $\pr$. The collections $(\phi_z)_{z\in Z}$ and $(\psi_z)_{z\in Z}$ defines
  natural isomorphisms $\phi \co h \To \pl \circ m$ and $\psi \co \pr \circ m
  \To k$ which satisfy $(g \psi) \circ (f \phi) = \theta$. Hence, the condition
  of \Cref{prop:charact-pb-bpb-ess-surj} is satisfied and $(P,\pl,\pr)$ is a
  bipullback.
\end{proof}
We also have the following criterion for rectangles of bipullbacks:
\begin{lemma}
  \label{lem:rect-left-square-bipullback}
  Given a rectangle made of two squares which are pullbacks in~$\Gpd$ as in
  \[
    \begin{tikzcd}
      L
      \ar[rd,phantom,very near start,"\drcorner"]
      \ar[r,"\pi^L_M"]
      \ar[d,"\pi^L_A"']
      &
      M
      \ar[rd,phantom,very near start,"\drcorner"]
      \ar[r,"\pi^M_R"]
      \ar[d,"\pi^M_B"{description}]
      &
      R
      \ar[d,"h"]
      \\
      A
      \ar[r,"f"']
      &
      B
      \ar[r,"g"']
      &
      C
    \end{tikzcd}
    \zbox,
  \]
  the following hold:
  \begin{enumerate}[(i)]
  \item \label{lem:rect-left-square-bipullback:rect-to-square} if the whole
    rectangle is a bipullback, then the left square is too;
    
  \item \label{lem:rect-left-square-bipullback:squares-to-rect} if the left and
    right square are bipullback, then the whole rectangle is a bipullback.
  \end{enumerate}
\end{lemma}

\begin{proof}
  We first prove \ref{lem:rect-left-square-bipullback:rect-to-square}. For this
  purpose, we use \Cref{prop:charact-pb-bpb-ess-surj}. So let $a \in A$, $m \in
  M$ and $\theta \co f(a) \to \pi^M_B(m)$. We have that $g(\theta)$ is a
  morphism from $g\circ f(a)$ to $g\circ \pi^M_B(m) = h\circ \pi^M_R(m)$. Since
  the outer rectangle is assumed to be a bipullback, there exist $a' \in A$, $r'
  \in R$, $u^A \co a \to a' \in A$, $v^R \co r' \to \pi^M_R(m)$ such that
  $g(\theta) = h(v^R) \circ (g\circ f)(u^A)$. Thus, we have $g(\theta \circ
  \finv{(f(u^A))}) = h(v^R)$. Since the right square is a pullback, there exists
  a unique $w^M$ such that $\pi^M_B(w^M) = \theta \circ \finv{(f(u^A))}$ and
  $\pi^M_R(w^M) = v^R$. Moreover, by the right pullback again, the target of
  $w^M$ is $m$; its source is some $m'$ such that $\pi^M_B(m') = f(a')$. Then,
  we have that $\theta = \pi^M_B(w^M) \circ f(u^A)$. We can conclude with
  \Cref{prop:charact-pb-bpb-ess-surj} that the left pullback is a bipullback.

  We now prove \ref{lem:rect-left-square-bipullback:squares-to-rect} using
  \Cref{prop:charact-pb-bpb-ess-surj} again. So let $a \in A$, $r \in R$ and
  $\theta \co (g\circ f)(a) \to h(r)$. Since the right square is a bipullback,
  there exist $b' \in B$, $r' \in R$, $u^B \co f(a) \to b'$ and $v^R \co r' \to
  r$ such that $\theta = h(v^R) \circ g(u^B)$. Since $b'$ and $r'$ have the same
  projection in $C$ through $g$ and $h$ respectively, there exists $m' \in M$
  such that $\pi^M_B(m') = b'$ and $\pi^M_R(m') = r'$. Thus, we have $u^B \co
  f(a) \to \pi^M_B(m')$. Since the left square is a bipullback, there exist $a''
  \in A$, $m'' \in M$, $\tilde u^A \co a \to a''$, $\tilde v^M \co m'' \to m'$
  such that $u^B = \pi^M_B(\tilde v^M) \circ f(\tilde u^A)$. We thus have
  \begin{align*}
    \theta
    &=
    h(v^R) \circ g(u^B)
    =
    h(v^R) \circ g(\pi^M_B(\tilde v^M) \circ f(\tilde u^A))
      \\
    &=
    h(v^R) \circ g(\pi^M_B(\tilde v^M)) \circ g(f(\tilde u^A))
    \\
    &
    =
    h(v^R) \circ h(\pi^M_R(\tilde v^M)) \circ g(f(\tilde u^A))
    \\
    &
    =
    h(v^R \circ \pi^M_R(\tilde v^M)) \circ (g\circ f)(\tilde u^A)
  \end{align*}
  which is precisely the factorization required by
  \Cref{prop:charact-pb-bpb-ess-surj} to conclude that the whole rectangle is a
  bipullback.
\end{proof}

\subsection{Uniformity and thinness}
\label{sec:unif-and-thinness}

Several arguments concerning uniformity requires some sort of diagram chasing
relative to bipullbacks. An important lemma for this is the following:
\begin{lemma}
  \label{lem:pb-bpb-three-tiles}
  Consider the diagram in $\Gpd$
  \[
    \begin{tikzcd}[rsbo=3em,csbo=4em]
      &&
      S
      \ar[ld,"\pl^{S}_{P}"']
      \ar[rd,"\pr^{S}_{Q}"]
      \phar[dd,"2"]
      \cphar[dd,near start,"\dcorner"]
      &&
      \\
      &
      P
      \ar[ld,"\pl^{P}_{L}"']
      \ar[rd,"\pr^{P}_{M}"{description}]
      \phar[dd,"1"]
      \cphar[dd,near start,"\dcorner"]
      &&
      Q
      \ar[ld,"\pl^{Q}_{M}"{description}]
      \ar[rd,"\pr^{Q}_{R}"]
      \phar[dd,"3"]
      \cphar[dd,near start,"\dcorner"]
      \\
      L
      \ar[rd,"f^L_A"']
      &&
      M
      \ar[ld,"{f^M_A}"]
      \ar[rd,"{f^M_B}"']
      &&
      R
      \ar[ld,"f^R_B"]
      \\
      &
      A
      &&
      B
    \end{tikzcd}
  \]
  where the square $1$, $2$ and $3$ are pullbacks, and derive from it the
  following diagram using the product structure:
  \[
    \begin{tikzcd}[column sep={5em,between origins}]
      &
      S
      \ar[dl,"\pr^{P}_M \circ \pl^S_{P}"']
      \ar[dr,"\gpdprodfactp {\pl^{P}_L \circ \pl^S_{P},\pr^{Q}_R \circ \pr^S_{Q}}"]
      \phar[dd,"4"]
      &
      \\
      M
      \ar[dr,"\gpdprodfactp{f^M_A,f^M_B}"']
      &&
      L \times R
      \ar[dl,"f^L_A\times f^R_B"]
      \\
      &
      A \times B
    \end{tikzcd}
  \]
  Then $4$ is a pullback.
  Moreover, the following hold:
  \begin{enumerate}[(i)]
  \item \label{lem:pb-bpb-three-tiles:right} if $1$ and the rectangle made of
    $2$ and $3$ are bipullbacks, then $4$ is a bipullback;
  \item \label{lem:pb-bpb-three-tiles:left} if $3$ and the rectangle made of
    $1$ and $2$ are bipullbacks, then $4$ is a bipullback;
  \item \label{lem:pb-bpb-three-tiles:mixed} if $4$ is a bipullback, then the
    rectangle made of $1$ and $2$ (\resp $2$ and $3$) is a bipullback.
  \end{enumerate}
\end{lemma}
\begin{proof}
  The fact that 4 is a pullback is an easy consequence of the fact that 1,2,3
  are pullbacks.

  One can then use \Cref{lem:carac_pb_bipb} without too much trouble on the
  different bipullback hypotheses in order to deduce that the wanted pullbacks
  are bipullbacks.
\end{proof}

With the above tool, we can now prove \Cref{prop:carac_lin}.
\begin{proof}[Proof of \Cref{prop:carac_lin}] We first prove the first
  implication, and start by showing \Cref{prop:carac_lin:pres}. So let
  $(S,\display^S_A) \in \SU$. Given $(U,\display^U_B) \in \U_B^\perp$, we must
  show that $T@S \perp U$. By hypothesis, we have that $T \perp S \times U$, \ie
  the pullback of $\display^T_{A \times B}$ and $\display^S_A \times
  \display^U_B$ is a bipullback. Thus, we conclude by
  \Cref{lem:pb-bpb-three-tiles}\ref{lem:pb-bpb-three-tiles:mixed} that the
  pullback of $\display^{T@S}_B$ and $\display^U_B$ is a bipullback, \ie
  $T@S \perp U$. Hence, $T@S \in \U_B$.

  We now show \Cref{prop:carac_lin:anti-unif}. Since $\id_B$ is an isofibration,
  we have that the pullback of $\display^T_B$ and $\id_B$ is a bipullback. Thus,
  given $U \in \U_A$, by \Cref{lem:pb-bpb-three-tiles}, we have that
  $\display^T_A \perp \display^U_A$ if and only if $\display^T_{A\times B} \perp
  \display^U_A \times \id_B$. But the latter holds, since $T \in \U_{A \lin B}$.
  Hence, $\display^T_A \in \U_A^\perp$.

  We now show the converse implication. So assume that $T$ satisfies
  \Cref{prop:carac_lin:pres,prop:carac_lin:anti-unif}. First note that, since
  $(A,\id_A) \in \U_A$, we have that $\display^T_B \in \U_B^\perp$ by
  \Cref{prop:carac_lin:pres}. Given $V \in \U_B$, we must show that, for every
  $U \in \U_A$, we have $T \perp U \times V$. Since we have $\display^T_B \perp
  \display^V_B$, by \Cref{lem:pb-bpb-three-tiles}, this is equivalent to have $U
  \perp \dual T@V$ for every $U \in \U_A$. By hypothesis, it is equivalent to
  only check the previous condition for $U \in S$. By
  \Cref{lem:pb-bpb-three-tiles} again, it is equivalent to check that $T \perp U
  \times V$ for every $U \in S$. Since $U \perp \display^T_A$, by
  \Cref{lem:pb-bpb-three-tiles} again, it is equivalent to check that $T@U \perp
  V$, but the latter holds by~\ref{prop:carac_lin:pres}. Thus, $T \in \U_{A \lin
    B}$.
\end{proof}
Using \Cref{prop:carac_lin}, we can prove the compatibility of uniformity with
composition:
\begin{proposition}
  \label{prop:compat-thinness-composition}
  Given uniform groupoids $(A,\U_A)$ and $(B,\U_B)$ and prestrategies $S \in
  \U_{A \lin B}$ and $T \in \U_{B \lin C}$, we have $T \sthc S \in \U_{A \lin C}$.
\end{proposition}
\begin{proof}
  Recall that the composition of the two spans $S$ and $T$ is formed as in
  \Cref{eq:pbcomp}. We show the uniformity of $T \sthc S$ using
  \Cref{prop:carac_lin} with $\U_A$ taken as generating class of $\U_A$. The
  fact that \Cref{prop:carac_lin:pres} is satisfied for $T \sthc S$ is immediate
  from its validity for both $S$ and $T$. The fact that
  \Cref{prop:carac_lin:anti-unif} holds, that is, that $\display^S_A \circ \pl
  \in \U_A^\perp$, is a consequence of the fact that $\display^T_B \in
  \U_B^\perp$ by \Cref{prop:carac_lin:anti-unif} on $T$, and the dual of
  \Cref{prop:carac_lin} for $S$, asserting in particular that $\dual S$ maps
  elements of $\U_B^\perp$ to $\U_A^\perp$.
\end{proof}
We handle thinness similarly, and start by proving \Cref{prop:carac_lin_thin}:

\begin{proof}[Proof of \Cref{prop:carac_lin_thin}]
  Assume that $T \in \T_{A \lin B}$ and let $S \in \SU$. By
  \Cref{prop:carac_lin}, we already have $T@S \in \U_{A \lin B}$. Next, we use
  the fact that $\T_B = \T_B^{\pperp\pperp}$ to show that $T@S \in \T_B$. So let
  $U \in \T_B^{\pperp}$. By \Cref{lem:pb-bpb-three-tiles}, we have $T@S \pperp
  U$ iff $T \pperp S \times U$. But the latter holds since $T\in \T_{A \lin B}$
  and $S \in \SU \subseteq \T_A$. Thus, $T@S \in \T_B$.

  Conversely, assume that $T@S \in \T_B$ for every $S \in \SU$. First observe
  that $\T_{A \lin B} = (\T_A \otimes \T_B^{\pperp})^{\pperp}$. So we must show
  that, for every $S \in \T_A$ and $U \in \T_B^{\pperp}$, $T \pperp S\times U$. By
  \Cref{lem:pb-bpb-three-tiles}, for a given $U$, the latter is equivalent to $S
  \pperp \dual T@U$ for every $S \in \T_A$, \ie $\dual T@U \in \T_A^{\pperp}$. But
  since $\T_A^{\pperp} = \SU^{\pperp}$, for a given $U$, it is equivalent to $S
  \pperp \dual T@U$ for every $S \in \SU$, itself equivalent to $T \pperp
  S\times U$ for every $S \in \SU$, and finally equivalent to $T@S \pperp U$ for
  every $S\in \SU$, which amounts to our initial assumption.
\end{proof}
Using \Cref{prop:carac_lin_thin}, we can then prove the compatibility of thinness
with composition:
\begin{proposition}
  Given thin groupoids $A$ and $B$ and strategies $S \in \U_{A \lin B}$ and
  $T \in \U_{B \lin C}$, we have $T \sthc S \in \T_{A \lin C}$.
\end{proposition}
\begin{proof}
  The proof is similar to (in fact simpler than) the one of
  \Cref{prop:compat-thinness-composition} and follows from the criterion given
  by \Cref{prop:carac_lin_thin}.
\end{proof}

\subsection{Details about the bicategory $\Thin$}
\label{sec:app-thin-bicat}

We have the following convenient characterization of $0$\composition of
$2$\cells of $\Thin$:
\begingroup
\renewcommand\rho{R}%
\renewcommand\sigma{S}%
\renewcommand\tau{T}%
\begin{proposition}[Paved Characterization of Composition (PCC)]
  \label{prop:pcc}
  Given thin groupoids $A,B,C$, strategies $\rho,\rho' \co A \stto B$,
  $\sigma,\sigma' \co B \stto C$ and weak morphisms $F \co \rho \To \rho'$ and
  $G \co \sigma \To \sigma'$ of $\Thin$, if there exist a functor $\stsupp H \co
  \stsupp{\sigma\sthcomp\rho} \to \stsupp{\sigma'\sthcomp\rho'}$ and two natural
  transformations $H^l$ and $H^r$ as in
  \[
    \begin{tikzcd}
      \stsupp{\sigma\sthcomp\rho}
      \ar[d,"\pl"']
      \ar[r,"\stsupp H",myname=F]
      \cphar[rd,"\xTo{H^l}"]
      &
      \stsupp{\sigma'\sthcomp\rho'}
      \ar[d,"\pl"]
      \\
      \rho
      \ar[r,"\stsupp F"']
      &
      \rho'
    \end{tikzcd}
    \qqand
    \begin{tikzcd}
      \stsupp{\sigma\sthcomp\rho}
      \ar[d,"\pr"']
      \ar[r,"\stsupp H",myname=F]
      \cphar[rd,"\xTo{H^r}"]
      &
      \stsupp{\sigma'\sthcomp\rho'}
      \ar[d,"\pr"]
      \\
      \sigma
      \ar[r,"\stsupp G"']
      &
      \sigma'
    \end{tikzcd}
  \]
  such that
  \[
    \begin{tikzcd}
      \stsupp{\sigma\sthcomp\rho}
      \ar[d,"\pl"']
      \ar[r,"\stsupp H",myname=F]
      \cphar[rd,"\xTo{H^l}"]
      &
      \stsupp{\sigma'\sthcomp\rho'}
      \ar[d,"\pl"]
      \\
      \rho
      \ar[r,"\stsupp F"{description}]
      \ar[d,"\stdisp\rho_B"']
      \cphar[rd,"\xTo{F^B}"]
      &
      \rho'
      \ar[d,"\stdisp{\rho'}_B"]
      \\
      B
      \ar[r,equals]
      &
      B
    \end{tikzcd}
    \qeq
    \begin{tikzcd}
      \stsupp{\sigma\sthcomp\rho}
      \ar[d,"\pr"']
      \ar[r,"\stsupp H",myname=F]
      \cphar[rd,"\xTo{H^r}"]
      &
      \stsupp{\sigma'\sthcomp\rho'}
      \ar[d,"\pr"]
      \\
      \sigma
      \ar[r,"\stsupp G"{description}]
      \ar[d,"\stdisp\sigma_B"']
      \cphar[rd,"\xTo{G^B}"]
      &
      \sigma'
      \ar[d,"\stdisp{\sigma'}_B"]
      \\
      B
      \ar[r,equals]
      &
      B
    \end{tikzcd}
  \]
  and such that the natural transformations
  \[
    H^A
    \defeq
    \begin{tikzcd}
      \stsupp{\sigma\sthcomp\rho}
      \ar[d,"\pl"']
      \ar[r,"\stsupp H",myname=F]
      \cphar[rd,"\xTo{H^l}"]
      &
      \stsupp{\sigma'\sthcomp\rho'}
      \ar[d,"\pl"]
      \\
      \rho
      \ar[r,"\stsupp F"{description}]
      \ar[d,"\stdisp\rho_A"']
      \cphar[rd,"\xTo{F^A}"]
      &
      \rho'
      \ar[d,"\stdisp{\rho'}_A"]
      \\
      A
      \ar[r,equals]
      &
      A
    \end{tikzcd}
    H^C
    \defeq
    \begin{tikzcd}
      \stsupp{\sigma\sthcomp\rho}
      \ar[d,"\pr"']
      \ar[r,"\stsupp H",myname=F]
      \cphar[rd,"\xTo{H^r}"]
      &
      \stsupp{\sigma'\sthcomp\rho'}
      \ar[d,"\pr"]
      \\
      \sigma
      \ar[r,"\stsupp G"{description}]
      \ar[d,"\stdisp\sigma_C"']
      \cphar[rd,"\xTo{G^C}"]
      &
      \sigma'
      \ar[d,"\stdisp{\sigma'}_C"]
      \\
      C
      \ar[r,equals]
      &
      C
    \end{tikzcd}
  \]
  are positive over $A^\perp$ and $C$ respectively, we have that $H \defeq
  (\stsupp H,H^A,H^C)$ is a positive morphism $\sigma\sthcomp\rho \To
  \sigma'\sthcomp\rho'$ of $\Thin$ and that $H = G \sthc F$.
\end{proposition}
\begin{proof}
  The fact that $H$ is a $2$\cell $\sigma\sthcomp\rho \To \sigma'\sthcomp\rho'$
  of $\Thin$ is immediate by the polarity assumption. The equality of $2$\cells
  given by the hypothesis can be rewritten as
  \[
    \begin{tikzcd}[rsbo=large,csbo=large]
      &
      \stsupp{\sigma\sthc\rho}
      \ar[dd,"\stsupp H"{description}]
      \ar[ld,myname=leftcone]
      \ar[rd,myname=rightcone]
      &
      \\
      \stsupp\rho
      \ar[dd,"\stsupp F"']
      \ar[rd,phantom,"{\xTo{H^l}}"]
      &
      &
      \stsupp\sigma
      \ar[dd,"\stsupp G"]
      \ar[ld,phantom,"{\xTo{\finv{(H^r)}}}"]
      \\
      &
      |[alias=midcell]|\stsupp{\sigma'\sthc\rho'}
      \ar[dl]
      \ar[dr]
      &
      \\
      \stsupp {\rho'}
      \ar[dr,"\stdisp{\rho'}_B"']
      &=&
      \stsupp{\sigma'}
      \ar[dl,"\stdisp{\sigma'}_B"]
      \\
      &
      B
    \end{tikzcd}
    \qeq
    \begin{tikzcd}[rsbo=huge,csbo=4.4em]
      &
      \stsupp{\sigma\sthc\rho}
      \ar[ld]
      \ar[rd]
      &
      \\
      \stsupp\rho
      \ar[d,"\stsupp F"']
      \ar[rdd,bend left,myname=midleft,"\stdisp\rho_{B}"{inner sep=0pt},pos=0.4]
      &
      =
      &
      \stsupp\sigma
      \ar[d,"\stsupp G"]
      \ar[ldd,bend right,myname=midright,"\stdisp\sigma_{B}"'{inner sep=0pt},pos=0.4]
      \\
      |[alias=stsupprhop]|\stsupp{\rho'}
      \ar[rd,"\stdisp{\rho'}_B"']
      &
      &
      |[alias=stsuppsigmap]|\stsupp{\sigma'}
      \ar[ld,"\stdisp{\sigma'}_B"]
      \\
      &
      B
      &
      \ar[from=stsupprhop,to=midleft,phantom,"{\xTo{\finv{(F^B)}}}"]
      \ar[from=midright,to=stsuppsigmap,phantom,"{\xTo{G^B}}"]
    \end{tikzcd}
  \]
  so that $H^l$ and $H^r$ provide a factorization of the pseudocone on the
  right, and define an object of the groupoid of compositions mentioned in
  \Cref{sec:thin-summarized}. The actual horizontal composition in $\Thin$ is
  then obtained by applying the biequivalence of \Cref{prop:positivization}. But
  since $\pair{H^A,H^C}$ is already a positive natural transformation on $\cA
  \lin \cC$, this biequivalence does nothing on this object and $H = G \sthc F$.
\end{proof}
\endgroup
\begingroup
\renewcommand\rho{R}%
\renewcommand\sigma{S}%
\renewcommand\tau{T}%
\begin{lemma}
  Given thin groupoids $\cA,\cB,\cC$, we have a functor 
  \[
    (-)\sthcomp (-) \co \Thin(\cB,\cC) \times \Thin(\cA,\cB) \to \Thin(\cA,\cC)\zbox.
  \]
\end{lemma}
\begin{proof}
  By the definition we took for the composition of the $2$\cells of $\Thin$, we
  already have that $(-)\sthcomp(-)$ respects the sources and targets of weak
  morphisms, so that we are left to verify functoriality.

  Given $\rho \in \Thin(\cA,\cB)$ and $\sigma \in\Thin(\cB,\cC)$, a solution
  in $\stsupp H$, $H^l$ and $H^r$ for the equation
  \[
    \begin{tikzcd}
      \stsupp{\sigma\sthc\rho}
      \ar[r,"\stsupp H"]
      \ar[d,"\pl"']
      \cphar[rd,"\xTo{H^l}"]
      &
      \stsupp{\sigma\sthc\rho}
      \ar[d,"\pl"]
      \\
      \stsupp{\rho}
      \ar[r,myname=midarr,"\id_{\stsupp\rho}"{description}]
      \ar[d,myname=rhodisp,"\stdisp\rho_B"']
      \cphar[rd,"="]
      &
      \stsupp{\rho}
      \ar[d,myname=rhopdisp,"\stdisp\rho_B"]
      \\
      B
      \ar[r,equals]
      &
      B
    \end{tikzcd}
    \qeq
    \begin{tikzcd}
      \stsupp{\sigma\sthc\rho}
      \ar[r,"\stsupp H"]
      \ar[d,"\pr"']
      \cphar[rd,"\xTo{H^r}"]
      &
      \stsupp{\sigma\sthc\rho}
      \ar[d,"\pr"]
      \\
      \stsupp{\sigma}
      \ar[r,myname=midarr,"\id_{\stsupp\sigma}"{description}]
      \ar[d,myname=rhodisp,"\stdisp\sigma_B"']
      \cphar[rd,"="]
      &
      \stsupp{\sigma}
      \ar[d,myname=rhopdisp,"\stdisp\sigma_B"]
      \\
      B
      \ar[r,equals]
      &
      B
    \end{tikzcd}
  \]
  is given by $\stsupp H = \id_{\rho\sthcomp\sigma}$, $H^l = \id_{\pl}$ and $H^r
  = \id_{\pr}$. Thus, since identities are member of $A^-$ and $C^+$, the
  polarity condition of \Cref{prop:pcc} is satisfied so that
  \[
    \id_\sigma \sthc \id_\rho
    \qeq
    (\id_{\stsupp{\sigma\sthcomp\rho}},\id_{\stdisp\rho_A \circ \pl},\id_{\stdisp\sigma_C\circ\pr})
    \qeq
    \id_{\sigma\sthcomp\rho}
    \zbox.
  \]
  
  Now, given four positive morphisms organized as
\[
\begin{array}{l}
    \rho \xto{F} \rho' \xto{F'} \rho'' \in \Thin(\cA,\cB)\,,\\
    \sigma \xto{G} \sigma' \xto{G'} \sigma'' \in \Thin(\cB,\cC)\,,
\end{array}
\]
  the procedure to compute $F \sthcomp G$ gives us $\stsupp H,H^l$
  and~$H^r$ \emph{s.t.}
  \[
    \begin{tikzcd}
      \stsupp{\sigma\sthcomp\rho}
      \ar[d,"\pl"']
      \ar[r,"\stsupp H",myname=F]
      \cphar[rd,"\xTo{H^l}"]
      &
      \stsupp{\sigma'\sthcomp\rho'}
      \ar[d,"\pl"]
      \\
      \rho
      \ar[r,"\stsupp F"{description}]
      \ar[d,"\stdisp\rho_B"']
      \cphar[rd,"\xTo{F^B}"]
      &
      \rho'
      \ar[d,"\stdisp{\rho'}_B"]
      \\
      B
      \ar[r,equals]
      &
      B
    \end{tikzcd}
    \qeq
    \begin{tikzcd}
      \stsupp{\sigma\sthcomp\rho}
      \ar[d,"\pr"']
      \ar[r,"\stsupp H",myname=F]
      \cphar[rd,"\xTo{H^r}"]
      &
      \stsupp{\sigma'\sthcomp\rho'}
      \ar[d,"\pr"]
      \\
      \sigma
      \ar[r,"\stsupp G"{description}]
      \ar[d,"\stdisp\sigma_B"']
      \cphar[rd,"\xTo{G^B}"]
      &
      \sigma'
      \ar[d,"\stdisp{\sigma'}_B"]
      \\
      B
      \ar[r,equals]
      &
      B
    \end{tikzcd}
  \]
  and with $\stdisp{\rho'}_A H^l$ and $\stdisp{\sigma'}_C H^r$ respectively
  negative on $\cA$ and positive on $\cC$; similarly, the procedure to compute
  $F' \sthcomp G'$ gives us $\stsupp{H'},H'^l$ and~$H'^r$ such that
  \[
    \begin{tikzcd}
      \stsupp{\sigma'\sthcomp\rho'}
      \ar[d,"\pl"']
      \ar[r,"\stsupp {H'}",myname=F]
      \cphar[rd,"\xTo{H'^l}"]
      &
      \stsupp{\sigma''\sthcomp\rho''}
      \ar[d,"\pl"]
      \\
      \rho'
      \ar[d,"\stdisp{\rho'}_B"']
      \ar[r,"\stsupp {F'}"{description}]
      \cphar[rd,"\xTo{F'^B}"]
      &
      \rho''
      \ar[d,"\stdisp{\rho''}_B"]
      \\
      B
      \ar[r,equals]
      &
      B
    \end{tikzcd}
    \qeq
    \begin{tikzcd}
      \stsupp{\sigma'\sthcomp\rho'}
      \ar[d,"\pr"']
      \ar[r,"\stsupp {H'}",myname=F]
      \cphar[rd,"\xTo{H'^r}"]
      &
      \stsupp{\sigma''\sthcomp\rho''}
      \ar[d,"\pr"]
      \\
      \sigma'
      \ar[d,"\stdisp{\sigma'}_B"']
      \ar[r,"\stsupp {G'}"{description}]
      \cphar[rd,"\xTo{G'^B}"]
      &
      \sigma''
      \ar[d,"\stdisp{\sigma''}_B"]
      \\
      B
      \ar[r,equals]
      &
      B
    \end{tikzcd}
  \]
  and with $\stdisp{\rho''}_A H'^l$ and $\stdisp{\sigma''}_C H'^r$ respectively
  negative on $\cA$ and positive on $\cC$. On the one hand, we thus have that
  $(G' \sthc F') \circ (G \sthc F)$ is the span morphism $K = (\stsupp
  K,K^l,K^r)$ with $\stsupp K = \stsupp {H'} \stsupp H$ and
  \begin{gather*}
    K^l
    =
    \begin{tikzcd}[ampersand replacement=\&]
      \stsupp{\sigma\sthcomp\rho}
      \ar[d,"\pl"']
      \ar[r,"\stsupp H",myname=F]
      \cphar[rd,"\xTo{H^l}"]
      \&
      \stsupp{\sigma'\sthcomp\rho'}
      \ar[d,"\pl"{description}]
      \ar[r,"\stsupp {H'}",myname=F]
      \cphar[rd,"\xTo{H'^l}"]
      \&
      \stsupp{\sigma''\sthcomp\rho''}
      \ar[d,"\pl"]
      \\
      \rho
      \ar[r,"\stsupp F"{description}]
      \ar[d,"\stdisp\rho_A"']
      \cphar[rd,"\xTo{F^A}"]
      \&
      \rho'
      \ar[d,"\stdisp{\rho'}_A"{description}]
      \ar[r,"\stsupp {F'}"{description}]
      \cphar[rd,"\xTo{F'^A}"]
      \&
      \rho''
      \ar[d,"\stdisp{\rho''}_A"]
      \\
      A
      \ar[r,equals]
      \&
      A
      \ar[r,equals]
      \&
      A
    \end{tikzcd}
    \\
    K^r
    =
    \begin{tikzcd}[ampersand replacement=\&]
      \stsupp{\sigma\sthcomp\rho}
      \ar[d,"\pr"']
      \ar[r,"\stsupp H",myname=F]
      \cphar[rd,"\xTo{H^r}"]
      \&
      \stsupp{\sigma'\sthcomp\rho'}
      \ar[d,"\pr"{description}]
      \ar[r,"\stsupp {H'}",myname=F]
      \cphar[rd,"\xTo{H'^r}"]
      \&
      \stsupp{\sigma''\sthcomp\rho''}
      \ar[d,"\pr"]
      \\
      \sigma
      \ar[r,"\stsupp G"{description}]
      \ar[d,"\stdisp\sigma_C"']
      \cphar[rd,"\xTo{G^C}"]
      \&
      \sigma'
      \ar[d,"\stdisp{\sigma'}_C"{description}]
      \ar[r,"\stsupp {G'}"{description}]
      \cphar[rd,"\xTo{G'^C}"]
      \&
      \sigma''
      \ar[d,"\stdisp{\sigma''}_C"]
      \\
      C
      \ar[r,equals]
      \&
      C
      \ar[r,equals]
      \&
      C
    \end{tikzcd}
    \zbox.
  \end{gather*}
  On the other hand, we have
  \[
    \begin{array}{rc}
      &
        \begin{tikzcd}
        \stsupp{\sigma\sthcomp\rho}
        \ar[d,"\pl"']
        \ar[r,"\stsupp H",myname=F]
        \cphar[rd,"\xTo{H^l}"]
        &
        \stsupp{\sigma'\sthcomp\rho'}
        \ar[d,"\pl"{description}]
        \ar[r,"\stsupp {H'}",myname=F]
        \cphar[rd,"\xTo{H'^l}"]
        &
        \stsupp{\sigma''\sthcomp\rho''}
        \ar[d,"\pl"]
        \\
        \rho
        \ar[r,"\stsupp F"{description}]
        \ar[d,"\stdisp\rho_B"']
        \cphar[rd,"\xTo{F^B}"]
        &
        \rho'
        \ar[d,"\stdisp{\rho'}_B"{description}]
        \ar[r,"\stsupp {F'}"{description}]
        \cphar[rd,"\xTo{F'^B}"]
        &
        \rho''
        \ar[d,"\stdisp{\rho''}_B"]
        \\
        B
        \ar[r,equals]
        &
        B
        \ar[r,equals]
        &
        B
      \end{tikzcd}
          \\
      =
      &
      \begin{tikzcd}
        \stsupp{\sigma\sthcomp\rho}
        \ar[d,"\pr"']
        \ar[r,"\stsupp H",myname=F]
        \cphar[rd,"\xTo{H^r}"]
        &
        \stsupp{\sigma'\sthcomp\rho'}
        \ar[d,"\pr"{description}]
        \ar[r,"\stsupp {H'}",myname=F]
        \cphar[rd,"\xTo{H'^r}"]
        &
        \stsupp{\sigma''\sthcomp\rho''}
        \ar[d,"\pr"]
        \\
        \sigma
        \ar[r,"\stsupp G"{description}]
        \ar[d,"\stdisp\sigma_B"']
        \cphar[rd,"\xTo{G^B}"]
        &
        \sigma'
        \ar[d,"\stdisp{\sigma'}_B"{description}]
        \ar[r,"\stsupp {G'}"{description}]
        \cphar[rd,"\xTo{G'^B}"]
        &
        \sigma''
        \ar[d,"\stdisp{\sigma''}_B"]
        \\
        B
        \ar[r,equals]
        &
        B
        \ar[r,equals]
        &
        B
      \end{tikzcd}
          \zbox.
    \end{array}
  \]

  Moreover, the two natural transformations obtained as the horizontal pasting
  of $H^l$ and $H'^l$ along $\pl$ and the horizontal pasting of $H^r$ and $H'^r$
  along $\pr$ satisfy the polarity condition of the \PCC.

  Hence, by considering again the diagrammatic definition of $K$, the \PCC tells
  us that $K$ is also $(G' \circ G) \sthc (F' \circ F)$, which concludes functoriality.
\end{proof}

We now show the unitality of the horizontal composition. Given a thin
groupoid~$A$, we write $\stid A$ for the identity span $A \xot{\id_A} A
\xto{\id_A}$. We have
\begin{lemma}
  \label{lem:stid-thin}
  Given a thin groupoid $\cA$, we have $\stid\cA \in \T_{\cA \lin \cA}$.
\end{lemma}
\begin{proof}
  This is an easy consequence of \Cref{prop:carac_lin} and
  \Cref{prop:carac_lin_thin}.
\end{proof}
We also write $\stid\cA$ for the corresponding functor $\termcat \to
\Thin(\cA,\cA)$. Given an additional thin groupoid $\cB$, there is a
transformation~$\stru$ between the functors
\begin{gather*}
  \Thin(\cA,\cB)
  \xto{\sim}
  \Thin(\cA,\cB) \times \termcat
  \to \cdots
  \\
  \xto{\unit{} \times \stid\cA}
  \Thin(\cA,\cB) \times \Thin(\cA,\cA)
  \xto{(-)\sthcomp(-)}
  \Thin(\cA,\cB)
\end{gather*}
and
\[
  \Thin(\cA,\cB)
  \xto{\unit{}}
  \Thin(\cA,\cB)
\]
whose component at $\sigma \in \Thin(\cA,\cB)$ is defined as
follows. Recall that the span $\sigma \sthcomp \stid\cA$ is defined by the
pullback
\[
  \begin{tikzcd}[cs={5em,between origins}]
    &&
    \stsupp {\sigma \sthcomp \stid\cA}
    \ar[dl,"\pl"']
    \ar[dr,"\pr"]
    \ar[dd,phantom,very near start,"\dcorner"]
    &&
    \\
    &
    A
    \ar[dr,"\unit A"{description}]
    \ar[dl,"\unit A"{description}]
    &&
    \stsupp \sigma
    \ar[dl,"{\stdisp[A]\sigma}"'{description}]
    \ar[dr,"{\stdisp[B]\sigma}"{description}]
    \\
    A
    &&
    A
    &&
    B
  \end{tikzcd}
  \zbox.
\]
Then, $\stsupp{\stru_\sigma} \defeq \pr$ is an isomorphism (as the pullback of
an isomorphism) which moreover induces a strong isomorphism of strategies
$\stru_\sigma \co \sigma \sthcomp \stid\cA \To \sigma \in \Thin$.

\begin{lemma}
  \label{lem:runit-natural}
  $\stru \defeq (\stru_\sigma)_{\sigma \in \Thin}$ is a natural isomorphism.
\end{lemma}
\begin{proof}
  Let $\cA,\cB$, $\sigma,\sigma'\co \cA \stto \cB$ and $F \co \sigma \to
  \sigma'$ in $\Thin(\cA,\cB)$. We first picture the two compositions $\tau \defeq \stid\cA
  \sthcomp \sigma$ and $\tau' \defeq \stid\cA \sthcomp \sigma'$ on the diagram
  \[
    \begin{tikzcd}[rsbo=large,csbo=large]
      &&
      \stsupp\tau
      \ar[dd,phantom,very near start,"\dcorner"]
      \ar[ld,dashed,"\pl^\tau"']
      \ar[rd,dashed,"\pr^\tau"]
      &&
      \\
      &
      A
      \ar[ld,"\unit A"{description}]
      \ar[rd,"\unit A"{description}]
      &&
      \stsupp\sigma
      \ar[ld,"\stdisp{\sigma}_A"{description}]
      \ar[rd,"\stdisp{\sigma}_B"{description}]
      \\
      A
      &&
      A
      &&
      B
      \\
      &
      A
      \ar[lu,"\unit A"{description}]
      \ar[ru,"\unit A"{description}]
      &&
      \stsupp{\sigma'}
      \ar[lu,"\stdisp{\sigma'}_A"{description}]
      \ar[ru,"\stdisp{\sigma'}_B"{description}]
      \\
      &&
      \stsupp{\tau'}
      \ar[uu,phantom,very near start,"\ucorner"]
      \ar[lu,dashed,"\pl^{\tau'}"]
      \ar[ru,dashed,"\pr^{\tau'}"']
    \end{tikzcd}
    \zbox.
  \]
  We now compute the composition $F \sthcomp \unit{\stid\cA}$.
  By the \PCC, it is the morphism $\tilde F \co \sigma \sthc \stid\cA \To
  \sigma' \sthc \stid\cA$ defined by
  \[
    \stsupp{\tilde F} 
    \qeq
    \stsupp{\tau} 
    \xto{\pr^\tau}
    \stsupp\sigma
    \xto{\stsupp F}
    \stsupp{\sigma'}
    \xto{\finv{(\pr^{\tau'})}}
    \stsupp{\tau'} 
  \]
  and
  \begin{gather*}
    \tilde F^A
    =
    \begin{tikzcd}[ampersand replacement=\&]
      \stsupp{\tau} 
      \ar[d,"\pl^\tau"']
      \ar[r,"\pr^\tau"]
      \cphar[rd,"="]
      \&
      \stsupp\sigma
      \ar[d,"\stdisp\sigma_A"{description}]
      \ar[r,"\stsupp F"]
      \cphar[rd,"\xTo{F^A}"]
      \&
      \stsupp{\sigma'}
      \ar[d,"\stdisp{\sigma'}_A"{description}]
      \ar[r,"\finv{(\pr^{\tau'})}"]
      \cphar[rd,"="]
      \&
      \stsupp{\tau'} 
      \ar[d,"\pl^{\tau'}"]
      \\
      A
      \ar[d,"\unit A"']
      \ar[r,equals]
      \cphar[rd,"="]
      \&
      A
      \ar[d,"\unit A"{description}]
      \ar[r,equals]
      \cphar[rd,"="]
      \&
      A
      \ar[d,"\unit A"{description}]
      \ar[r,equals]
      \cphar[rd,"="]
      \&
      A
      \ar[d,"\unit A"]
      \\
      A
      \ar[r,equals]
      \&
      A
      \ar[r,equals]
      \&
      A
      \ar[r,equals]
      \&
      A
    \end{tikzcd}
    \\
    \tilde F^B
    =
    \begin{tikzcd}[ampersand replacement=\&]
      \stsupp{\tau} 
      \ar[d,"\pr^\tau"']
      \ar[r,"\pr^\tau"]
      \cphar[rrrd,"="]
      \&
      \stsupp\sigma
      \ar[r,"\stsupp F"]
      \&
      \stsupp{\sigma'}
      \ar[r,"\finv{(\pr^{\tau'})}"]
      \&
      \stsupp{\tau'} 
      \ar[d,"\pr^{\tau'}"]
      \\
      \stsupp\sigma
      \ar[d,"\stdisp\sigma_B"']
      \ar[rrr,"\stsupp F"{description}]
      \cphar[rrrd,"\xTo{F^B}"]
      \&
      \&
      \&
      \stsupp{\sigma'}
      \ar[d,"\stdisp{\sigma'}_B"]
      \\
      B
      \ar[rrr,equals]
      \&
      \&
      \&
      B
    \end{tikzcd}
    \zbox.
  \end{gather*}
  The naturality of $\stru$ is then expressed by the equation
  \[
    \stru_{\sigma'} \circ (F \sthc \unit{\stid\cA}) = F \circ \stru_\sigma
  \]
  that we now check. We first have
  \[
    \stsupp{\stru_{\sigma'} \circ (F \sthc \unit{\stid\cA})}
    \qeq
    \stsupp{\tau} 
    \xto{\pr^\tau}
    \stsupp\sigma
    \xto{\stsupp F}
    \stsupp{\sigma'}
    \qeq
    \stsupp{F \circ \stru_\sigma}
    \zbox.
  \]
  Moreover,
  \[
    \begin{array}{ccc}
      (\stru_{\sigma'} \circ (F \sthc \unit{\stid\cA}))^A
      &
        =
      &
        \begin{tikzcd}
          \stsupp{\tau} 
          \ar[d,"\pl^\tau"']
          \ar[r,"\pr^\tau"]
          \cphar[rd,"="]
          &
          \stsupp\sigma
          \ar[d,"\stdisp\sigma_A"{description}]
          \ar[r,"\stsupp F"]
          \cphar[rd,"\xTo{F^A}"]
          &
          \stsupp{\sigma'}
          \ar[d,"\stdisp{\sigma'}_A"]
          \\
          A
          \ar[d,"\unit A"']
          \ar[r,equals]
          \cphar[rd,"="]
          &
          A
          \ar[d,"\unit A"{description}]
          \ar[r,equals]
          \cphar[rd,"="]
          &
          A
          \ar[d,"\unit A"]
          \\
          A
          \ar[r,equals]
          &
          A
          \ar[r,equals]
          &
          A
        \end{tikzcd}
            \\[5em]
      &=&
         (F \circ \stru_\sigma)^A
    \end{array}
  \]
  and
  \[
    \begin{array}{ccc}
      (\stru_{\sigma'} \circ (F \sthc \unit{\stid\cA}))^B
      &=&
          \begin{tikzcd}
            \stsupp{\tau} 
            \ar[d,"\pr^\tau"']
            \ar[r,"\pr^\tau"]
            \cphar[rrd,"="]
            &
            \stsupp\sigma
            \ar[r,"\stsupp F"]
            &
            \stsupp{\sigma'}
            \ar[d,equals]
            \\
            \stsupp\sigma
            \ar[d,"\stdisp\sigma_B"']
            \ar[rr,"\stsupp F"{description}]
            \cphar[rrd,"\xTo{F^B}"]
            &
            &
            \stsupp{\sigma'}
            \ar[d,"\stdisp{\sigma'}_B"]
            \\
            B
            \ar[rr,equals]
            &
            &
            B
          \end{tikzcd}
              \\[5em]
      &=&
          (F \circ \stru_\sigma)^B
          \zbox.
    \end{array}
  \]
  Which concludes the proof that $\stru$ defines a natural iso.
\end{proof}
Similarly, there is a transformation~$\stlu$ between
\begin{gather*}
  \Thin(\cA,\cB)
  \xto{\sim}
  \termcat \times \Thin(\cA,\cB)
  \to\cdots
  \\
  \xto{\stid\cB \times \unit{}}
  \Thin(\cB,\cB) \times \Thin(\cA,\cB)
  \xto{(-)\sthcomp(-)}
  \Thin(\cA,\cB)
\end{gather*}
and
\[
  \Thin(\cA,\cB)
  \xto{\unit{}}
  \Thin(\cA,\cB)
\]
whose component at a strategy $\sigma \in \Thin(\cA,\cB)$ is defined as
follows. Recall that the span $\stid\cB \sthcomp \sigma$ is defined by the
pullback
\[
  \begin{tikzcd}[cs={5em,between origins}]
    &&
    \stsupp {\stid\cB \sthcomp \sigma}
    \ar[dl,"\pl"']
    \ar[dr,"\pr"]
    \ar[dd,phantom,very near start,"\dcorner"]
    &&
    \\
    &
    \stsupp\sigma
    \ar[dr,"\stdisp\sigma_B"{description}]
    \ar[dl,"\stdisp\sigma_A"{description}]
    &&
    B
    \ar[dl,"{\unit B}"'{description}]
    \ar[dr,"{\unit B}"{description}]
    \\
    A
    &&
    B
    &&
    B
  \end{tikzcd}
  \zbox.
\]
Then, $\stsupp{\stlu_\sigma} \defeq \pl$ is an isomorphism (as the pullback of
an isomorphism) which moreover induces a strong isomorphism of thin spans
$\stlu_\sigma \co \stid\cB \sthcomp \sigma \To \sigma \in \Thin$. As before, we have
\begin{lemma}
  \label{lem:lunit-natural}
  $\stlu \defeq (\stlu_\sigma)_{\sigma \in \Thin}$ is a natural isomorphism.
\end{lemma}
\endgroup
\begingroup
\renewcommand\rho{R}%
\renewcommand\sigma{S}%
\renewcommand\tau{T}%
\renewcommand\upsilon{U}%
\noindent Given thin groupoids $\cA,\cB,\cC,\cD$, there is a transformation
\begin{gather*}
  \stass \co ((-)\sthcomp(-))\sthcomp(-) \To (-)\sthcomp((-)\sthcomp(-))
  \\
  \co \Thin(\cC,\cD) \times \Thin(\cB,\cC) \times \Thin(\cA,\cB)
  \\
  \to \Thin(\cA,\cD)
\end{gather*}
whose component at $\sigma \in \Thin(\cA,\cB)$, $\tau\in \Thin(\cB,\cC)$ and
$\upsilon\in \Thin(\cC,\cD)$ is given by a strong morphism
\[
  \stass_{\sigma , \tau , \ups} \co (\ups \sthcomp \tau) \sthcomp \sigma \to \upsilon \sthcomp (\tau
  \sthcomp \sigma)
\]
defined as expected between the two compositions using the different pullbacks
involved, as in
\[
  \begin{tikzcd}[csbo=3em,rsbo=2em]
   &&&
   \stsupp{(\ups\sthcomp\tau)\sthcomp\sigma}
   \phar[dd,very near start,"\dcorner"]
   \ar[rd,"\pr^{(\ups\sthcomp\tau)\sthcomp\sigma}"]
   \ar[lldd,"\pl^{(\ups\sthcomp\tau)\sthcomp\sigma}"']
   \ar[dddddd,"\stsupp{\stass_{\sigma,\tau,\ups}}"'{description},bend left=80,looseness=2.5]
   &&&
   \\
   &&&&
   \stsupp{\ups\sthcomp\tau}
   \phar[dd,very near start,"\dcorner"]
   \ar[ld,"\pl^{\ups\sthcomp\tau}"']
   \ar[rd,"\pr^{\ups\sthcomp\tau}"]
   \\
   &
   \stsupp\sigma
   \ar[ld]
   \ar[rd]
   &&
   \stsupp\tau
   \ar[ld]
   \ar[rd]
   &&
   \stsupp\ups
   \ar[ld]
   \ar[rd]
   \\
   A
   &&
   B
   &&
   C
   &&
   D
   \\
   &
   \stsupp\sigma
   \ar[lu]
   \ar[ru]
   &&
   \stsupp\tau
   \ar[lu]
   \ar[ru]
   &&
   \stsupp\ups
   \ar[lu]
   \ar[ru]
   \\
   &&
   \stsupp{\tau\sthcomp\sigma}
   \phar[uu,very near start,"\ucorner"]
   \ar[lu,"\pl^{\tau\sthcomp\sigma}"]
   \ar[ru,"\pr^{\tau\sthcomp\sigma}"']
   \\
   &&&
   \stsupp{\ups\sthcomp(\tau\sthcomp\sigma)}
   \phar[uu,very near start,"\ucorner"]
   \ar[lu,"\pl^{\ups\sthcomp(\tau\sthcomp\sigma)}"]
   \ar[rruu,"\pr^{\ups\sthcomp(\tau\sthcomp\sigma)}"']
   &&&
  \end{tikzcd}
  \zbox.
\]
An inverse for $\stass_{\sigma,\tau,\ups}$ is defined symmetrically, so that $\stass$ is
an isomorphic transformation.

\begin{lemma}
  \label{lem:assoc-natural}
  The transformation $\stass$ is a natural isomorphism.
\end{lemma}
\begin{proof}
  Let $F \co \sigma \To \sigma' \co \cA \stto \cB$, $G \co \tau \To \tau' \co
  \cB \stto \cC$ and $H \co \ups \To \ups' \co \cC \stto \cD$ be weak morphisms
  in $\Thin$. We compute $(\ups \sthcomp \tau) \sthcomp \sigma$ as usual but
  moreover factor the projection $\pl^{(\ups \sthcomp \tau) \sthcomp \sigma}$
  canonically through $\stsupp{\tau \sthcomp \sigma}$ by a unique morphism
  $\tilde\pl^{(\ups \sthcomp \tau) \sthcomp \sigma}$ so that we get a diagram
  \[
    \begin{tikzcd}[csbo=3.5em,rsbo=2em]
      &&&
      \stsupp{(\ups\sthcomp\tau)\sthcomp\sigma}
      \phar[dd,very near start,"\dcorner"]
      \ar[ld,"\tilde\pl^{(\ups \sthcomp \tau) \sthcomp \sigma}"']
      \ar[rd]
      &&&
      \\
      &&
      \stsupp{\tau\sthcomp\sigma}
      \phar[dd,very near start,"\dcorner"]
      \ar[ld]
      \ar[rd]
      &&
      \stsupp{\ups\sthcomp\tau}
      \phar[dd,very near start,"\dcorner"]
      \ar[ld]
      \ar[rd]
      \\
      &
      \stsupp\sigma
      \ar[ld]
      \ar[rd]
      &&
      \stsupp\tau
      \ar[ld]
      \ar[rd]
      &&
      \stsupp\ups
      \ar[ld]
      \ar[rd]
      \\
      A
      &&
      B
      &&
      C
      &&
      D
    \end{tikzcd}
  \]
  and we get a similar diagram for $\stsupp{(\ups\sthcomp\tau)\sthcomp\sigma}$.
  Symmetrically, the projection $\pr^{\ups \sthcomp (\tau \sthcomp \sigma)}$ can
  be factored canonically through $\stsupp{\ups\sthc\tau}$ by a morphism $\tilde
  \pr^{\ups \sthcomp (\tau \sthcomp \sigma)}$, and the projection $\pr^{\ups'
    \sthcomp (\tau' \sthcomp \sigma')}$ through $\stsupp{\ups'\sthc\tau'}$ by a
  morphism $\tilde \pr^{\ups' \sthcomp (\tau' \sthcomp \sigma')}$. Note that
  $\pl^{\ups\sthc(\tau\sthc\sigma)} \circ \stsupp{\stass_{\sigma,\tau,\ups}} =
  \tilde\pl^{(\ups\sthc\tau)\sthc\sigma}$ and other similar equalities hold. By
  computing $K \defeq G \sthcomp F$, we get natural transformations $\tilde K^A$
  and $\tilde K^C$ such that
  \begin{gather*}
    K^A
    =
    \begin{tikzcd}[column sep=small,ampersand replacement=\&]
      \stsupp{\tau \sthcomp \sigma}
      \ar[r,"\stsupp K",myname=Garr]
      \ar[d,"\pl^{\tau\sthcomp\sigma}"']
      \cphar[rd,"\xTo{\tilde K^A}"]
      \&
      \stsupp{\tau' \sthcomp \sigma'}
      \ar[d,"\pl^{\tau'\sthcomp\sigma'}"]
      \\
      \stsupp\sigma
      \ar[d,"\stdisp\sigma_A"{description}]
      \ar[r,"\stsupp F"{description},myname=mididA]
      \cphar[rd,"\xTo{F^A}"]
      \&
      \stsupp{\sigma'}
      \ar[d,"\stdisp{\sigma'}_A"{description}]
      \\
      |[alias=botA]|A
      \ar[r,equals]
      \&
      A
    \end{tikzcd}
    \shortintertext{and}
    K^C
    =
    \begin{tikzcd}[column sep=small,ampersand replacement=\&]
      \stsupp{\tau \sthcomp \sigma}
      \ar[r,"\stsupp K",myname=Garr]
      \ar[d,"\pr^{\tau\sthcomp\sigma}"']
      \cphar[rd,"\xTo{\tilde K^C}"]
      \&
      \stsupp{\tau' \sthcomp \sigma'}
      \ar[d,"\pr^{\tau'\sthcomp\sigma'}"]
      \\
      \stsupp\tau
      \ar[d,"\stdisp\tau_C"{description}]
      \ar[r,"\stsupp G"{description},myname=mididA]
      \cphar[rd,"\xTo{G^C}"]
      \&
      \stsupp{\tau'}
      \ar[d,"\stdisp{\tau'}_C"{description}]
      \\
      |[alias=botA]|C
      \ar[r,equals]
      \&
      C
    \end{tikzcd}
  \end{gather*}
  \[
  \]
  which satisfy moreover that
  \[
    \begin{tikzcd}[column sep=small]
      \stsupp{\tau \sthcomp \sigma}
      \ar[r,"\stsupp K",myname=Garr]
      \ar[d,"\pl^{\tau\sthcomp\sigma}"']
      \cphar[rd,"\xTo{\tilde K^A}"]
      &
      \stsupp{\tau' \sthcomp \sigma'}
      \ar[d,"\pl^{\tau'\sthcomp\sigma'}"]
      \\
      \stsupp\sigma
      \ar[d,"\stdisp\sigma_B"{description}]
      \ar[r,"\stsupp F"{description},myname=mididA]
      \cphar[rd,"\xTo{F^B}"]
      &
      \stsupp{\sigma'}
      \ar[d,"\stdisp{\sigma'}_B"{description}]
      \\
      |[alias=botA]|B
      \ar[r,equals]
      &
      B
    \end{tikzcd}
    =
    \begin{tikzcd}[column sep=small]
      \stsupp{\tau \sthcomp \sigma}
      \ar[r,"\stsupp K",myname=Garr]
      \ar[d,"\pr^{\tau\sthcomp\sigma}"']
      \cphar[rd,"\xTo{\tilde K^C}"]
      &
      \stsupp{\tau' \sthcomp \sigma'}
      \ar[d,"\pr^{\tau'\sthcomp\sigma'}"]
      \\
      \stsupp\tau
      \ar[d,"\stdisp\tau_B"{description}]
      \ar[r,"\stsupp G"{description},myname=mididA]
      \cphar[rd,"\xTo{G^B}"]
      &
      \stsupp{\tau'}
      \ar[d,"\stdisp{\tau'}_B"{description}]
      \\
      |[alias=botA]|B
      \ar[r,equals]
      &
      B
    \end{tikzcd}
    \zbox.
  \]
  Similarly, by computing $L \defeq H \sthcomp G$, we get natural
  transformations $\tilde L^B$ and $\tilde L^D$ such that
  \begin{gather*}
    L^B
    =
    \begin{tikzcd}[column sep=small,ampersand replacement=\&]
      \stsupp{\ups \sthcomp \tau}
      \ar[r,"\stsupp L",myname=Garr]
      \ar[d,"\pl^{\ups\sthcomp\tau}"']
      \cphar[rd,"\xTo{\tilde L^B}"]
      \&
      \stsupp{\ups' \sthcomp \tau'}
      \ar[d,"\pl^{\ups'\sthcomp\tau'}"]
      \\
      \stsupp\tau
      \ar[d,"\stdisp\tau_B"{description}]
      \ar[r,"\stsupp G"{description},myname=mididB]
      \cphar[rd,"\xTo{G^B}"]
      \&
      \stsupp{\tau'}
      \ar[d,"\stdisp{\tau'}_B"{description}]
      \\
      |[alias=botB]|B
      \ar[r,equals]
      \&
      B
    \end{tikzcd}
    \shortintertext{and}
    L^D
    =
    \begin{tikzcd}[column sep=small,ampersand replacement=\&]
      \stsupp{\ups \sthcomp \tau}
      \ar[r,"\stsupp L",myname=Garr]
      \ar[d,"\pr^{\ups\sthcomp\tau}"']
      \cphar[rd,"\xTo{\tilde L^D}"]
      \&
      \stsupp{\ups' \sthcomp \tau'}
      \ar[d,"\pr^{\ups'\sthcomp\tau'}"]
      \\
      \stsupp\ups
      \ar[d,"\stdisp\ups_D"{description}]
      \ar[r,"H"{description},myname=mididB]
      \cphar[rd,"\xTo{H^D}"]
      \&
      \stsupp{\ups'}
      \ar[d,"\stdisp{\ups'}_D"{description}]
      \\
      |[alias=botB]|D
      \ar[r,equals]
      \&
      D
    \end{tikzcd}
  \end{gather*}
  which satisfy moreover that
  \[
    \begin{tikzcd}[column sep=small]
      \stsupp{\ups \sthcomp \tau}
      \ar[r,"\stsupp L",myname=Garr]
      \ar[d,"\pl^{\ups\sthcomp\tau}"']
      \cphar[rd,"\xTo{\tilde L^B}"]
      &
      \stsupp{\ups' \sthcomp \tau'}
      \ar[d,"\pl^{\ups'\sthcomp\tau'}"]
      \\
      \stsupp\tau
      \ar[d,"\stdisp\tau_C"{description}]
      \ar[r,"\stsupp G"{description},myname=mididB]
      \cphar[rd,"\xTo{G^C}"]
      &
      \stsupp{\tau'}
      \ar[d,"\stdisp{\tau'}_C"{description}]
      \\
      |[alias=botB]|C
      \ar[r,equals]
      &
      C
    \end{tikzcd}
    =
    \begin{tikzcd}[column sep=small]
      \stsupp{\ups \sthcomp \tau}
      \ar[r,"\stsupp L",myname=Garr]
      \ar[d,"\pr^{\ups\sthcomp\tau}"']
      \cphar[rd,"\xTo{\tilde L^D}"]
      &
      \stsupp{\ups' \sthcomp \tau'}
      \ar[d,"\pr^{\ups'\sthcomp\tau'}"]
      \\
      \stsupp\ups
      \ar[d,"\stdisp\ups_C"{description}]
      \ar[r,"\stsupp H"{description},myname=mididB]
      \cphar[rd,"\xTo{H^C}"]
      &
      \stsupp{\ups'}
      \ar[d,"\stdisp{\ups'}_C"{description}]
      \\
      |[alias=botB]|C
      \ar[r,equals]
      &
      C
    \end{tikzcd}
    \zbox.
  \]
  Since
  \[
    \begin{tikzcd}[column sep={4em,between origins},row sep={3em,between origins}]
      &
      \stsupp{(\ups'\sthcomp\tau')\sthcomp\sigma'}
      \phar[dd,very near start,"\dcorner"]
      \ar[rd,"\pr^{(\ups' \sthcomp \tau') \sthcomp \sigma'}"]
      \ar[ld,"\tilde\pl^{(\ups' \sthcomp \tau') \sthcomp \sigma'}"']
      &
      \\
      \stsupp{\tau'\sthcomp\sigma'}
      \ar[rd,"\pr^{\tau'\sthcomp\sigma'}"']
      &&
      \stsupp{\ups'\sthc\tau'}
      \ar[ld,"\pl^{\ups'\sthc\tau'}"]
      \\
      &
      \stsupp{\tau'}
    \end{tikzcd}
  \]
  is a bipullback by \Cref{lem:rect-left-square-bipullback}, and that the
  natural transformations $\tilde K^C$ and $\tilde L^B$ define a pseudocone of
  vertex $\stsupp{(\ups\sthcomp\tau)\sthcomp\sigma}$ on the associated cospan,
  we get $\stsupp M$, $\tilde M^A$ and $\tilde M^D$ such that
  \[
    \begin{array}{cc}
      &
      \begin{tikzcd}[column sep=small]
        \stsupp{(\ups\sthcomp\tau)\sthcomp\sigma}
        \ar[r,"\stsupp M"]
        \ar[d,"\tilde\pl^{(\ups \sthcomp \tau) \sthcomp \sigma}"']
        \cphar[rd,"\xTo{\tilde M^A}"]
        &
        \stsupp{(\ups'\sthcomp\tau')\sthcomp\sigma'}
        \ar[d,"\tilde\pl^{(\ups' \sthcomp \tau') \sthcomp \sigma'}"]
        \\
        \stsupp{\tau \sthcomp \sigma}
        \ar[r,"\stsupp K"{description},myname=Garr]
        \ar[d,"\pr^{\tau\sthcomp\sigma}"']
        \cphar[rd,"\xTo{\tilde K^C}"]
        &
        \stsupp{\tau' \sthcomp \sigma'}
        \ar[d,"\pr^{\tau'\sthcomp\sigma'}"]
        \\
        \stsupp\tau
        \ar[r,"\stsupp G"{description},myname=mididA]
        &
        \stsupp{\tau'}
      \end{tikzcd}
          \\
      =
      &
        \begin{tikzcd}[column sep=small]
          \stsupp{(\ups\sthcomp\tau)\sthcomp\sigma}
          \ar[r,"\stsupp M"]
          \ar[d,"\pr^{(\ups \sthcomp \tau) \sthcomp \sigma}"']
          \cphar[rd,"\xTo{\tilde M^D}"]
          &
          \stsupp{(\ups'\sthcomp\tau')\sthcomp\sigma'}
          \ar[d,"\pr^{(\ups' \sthcomp \tau') \sthcomp \sigma'}"]
          \\
          \stsupp{\ups\sthcomp\tau}
          \ar[r,"\stsupp L"{description},myname=Garr]
          \ar[d,"\pr^{\tau\sthcomp\sigma}"']
          \cphar[rd,"\xTo{\tilde L^B}"]
          &
          \stsupp{\ups'\sthcomp\tau'}
          \ar[d,"\pr^{\tau'\sthcomp\sigma'}"]
          \\
          \stsupp\tau
          \ar[r,"\stsupp G"{description},myname=mididA]
          &
          \stsupp{\tau'}
        \end{tikzcd}
            \zbox.
    \end{array}
  \]
  We thus get a weak morphism $M = (\stsupp M,M^A,M^D)$ between
  $(\ups\sthcomp\tau)\sthcomp\sigma$ and $(\ups'\sthcomp\tau')\sthcomp\sigma'$
  in $\Span$, where
  \begin{gather*}
    M^A
    =
    \begin{tikzcd}[column sep=small,ampersand replacement=\&]
      \stsupp{(\ups\sthcomp\tau)\sthcomp\sigma}
      \ar[r,"\stsupp M"]
      \ar[d,"\tilde\pl^{(\ups \sthcomp \tau) \sthcomp \sigma}"']
      \cphar[rd,"\xTo{\tilde M^A}"]
      \&
      \stsupp{(\ups'\sthcomp\tau')\sthcomp\sigma'}
      \ar[d,"\tilde\pl^{(\ups' \sthcomp \tau') \sthcomp \sigma'}"]
      \\
      \stsupp{\tau \sthcomp \sigma}
      \ar[r,"\stsupp K"{description},myname=Garr]
      \ar[d,"\pl^{\tau\sthcomp\sigma}"']
      \cphar[rd,"\xTo{\tilde K^A}"]
      \&
      \stsupp{\tau' \sthcomp \sigma'}
      \ar[d,"\pl^{\tau'\sthcomp\sigma'}"]
      \\
      \stsupp\sigma
      \ar[d,"\stdisp\sigma_A"{description}]
      \ar[r,"\stsupp F"{description},myname=mididA]
      \cphar[rd,"\xTo{F^A}"]
      \&
      \stsupp{\sigma'}
      \ar[d,"\stdisp{\sigma'}_A"{description}]
      \\
      |[alias=botA]|A
      \ar[r,equals]
      \&
      A
    \end{tikzcd}
    \shortintertext{and}
    M^D
    =
    \begin{tikzcd}[column sep=small,ampersand replacement=\&]
      \stsupp{(\ups\sthcomp\tau)\sthcomp\sigma}
      \ar[r,"\stsupp M"]
      \ar[d,"\pr^{(\ups \sthcomp \tau) \sthcomp \sigma}"']
      \cphar[rd,"\xTo{\tilde M^D}"]
      \&
      \stsupp{(\ups'\sthcomp\tau')\sthcomp\sigma'}
      \ar[d,"\pr^{(\ups' \sthcomp \tau') \sthcomp \sigma'}"]
      \\
      \stsupp{\ups \sthcomp \tau}
      \ar[r,"\stsupp L"{description},myname=Garr]
      \ar[d,"\pr^{\ups\sthcomp\tau}"']
      \cphar[rd,"\xTo{\tilde L^D}"]
      \&
      \stsupp{\ups' \sthcomp \tau'}
      \ar[d,"\pr^{\ups'\sthcomp\tau'}"]
      \\
      \stsupp\ups
      \ar[d,"\stdisp\ups_D"{description}]
      \ar[r,"H"{description},myname=mididB]
      \cphar[rd,"\xTo{H^D}"]
      \&
      \stsupp{\ups'}
      \ar[d,"\stdisp{\ups'}_D"{description}]
      \\
      |[alias=botB]|D
      \ar[r,equals]
      \&
      D
    \end{tikzcd}
    \zbox.
  \end{gather*}
  Using the biequivalence of \Cref{prop:positivization}, we can suppose that
  $\tilde M^A$ and $\tilde M^B$ where chosen so that $M^A$ and $M^D$ are
  respectively negative on $\cA$ and positive on $\cD$. By the \PCC, we can then
  verify directly that $M = (H\sthcomp G)\sthcomp F$.

  Now consider the positive weak morphism $\bar M = \stass_{\sigma',\tau',\ups'}
  \circ M \circ \finv{\stass_{\sigma,\tau,\ups}}$: we have that $\bar M^A$
  and $\bar M^D$ are as on \Cref{fig:nat-trans-barMA} and
  \Cref{fig:nat-trans-barMD}.
  \begin{figure*}
    \centering
    \[
      \bar M^A
      =
      \begin{tikzcd}[column sep=small]
        \stsupp{\ups\sthcomp(\tau\sthcomp\sigma)}
        \ar[d,"\pl^{\ups \sthcomp (\tau \sthcomp \sigma)}"']
        \ar[r,"\finv {\stsupp{\stass_{\sigma,\tau,\ups}}}"]
        \cphar[rd,"="]
        &
        \stsupp{(\ups\sthcomp\tau)\sthcomp\sigma}
        \ar[r,"\stsupp M"]
        \ar[d,"\tilde\pl^{(\ups \sthcomp \tau) \sthcomp \sigma}"{description}]
        \cphar[rd,"\xTo{\tilde M^A}"]
        &
        \stsupp{(\ups'\sthcomp\tau')\sthcomp\sigma'}
        \ar[d,"\tilde\pl^{(\ups' \sthcomp \tau') \sthcomp \sigma'}"{description}]
        \ar[r,"\stsupp{\stass_{\sigma',\tau',\ups'}}"]
        \cphar[rd,"="]
        &
        \stsupp{\ups'\sthcomp(\tau'\sthcomp\sigma')}
        \ar[d,"\pl^{\ups' \sthcomp (\tau' \sthcomp \sigma')}"]
        \\
        \stsupp{\tau \sthcomp \sigma}
        \ar[d,"\pl^{\tau\sthcomp\sigma}"']
        \ar[r,equals]
        \cphar[rd,"="]
        &
        \stsupp{\tau \sthcomp \sigma}
        \ar[r,"\stsupp K"{description},myname=Garr]
        \ar[d,"\pl^{\tau\sthcomp\sigma}"{description}]
        \cphar[rd,"\xTo{\tilde K^A}"]
        &
        \stsupp{\tau' \sthcomp \sigma'}
        \ar[d,"\pl^{\tau'\sthcomp\sigma'}"{description}]
        \ar[r,equals]
        \cphar[rd,"="]
        &
        \stsupp{\tau' \sthcomp \sigma'}
        \ar[d,"\pl^{\tau'\sthcomp\sigma'}"]
        \\
        \stsupp\sigma
        \ar[d,"\stdisp\sigma_A"']
        \ar[r,equals]
        \cphar[rd,"="]
        &
        \stsupp\sigma
        \ar[d,"\stdisp\sigma_A"{description}]
        \ar[r,"\stsupp F"{description},myname=mididA]
        \cphar[rd,"\xTo{F^A}"]
        &
        \stsupp{\sigma'}
        \ar[d,"\stdisp{\sigma'}_A"{description}]
        \ar[r,equals]
        \cphar[rd,"="]
        &
        \stsupp{\sigma'}
        \ar[d,"\stdisp{\sigma'}_A"]
        \\
        A
        \ar[r,equals]
        &
        |[alias=botA]|A
        \ar[r,equals]
        &
        A
        \ar[r,equals]
        &
        A
      \end{tikzcd}
    \]
    \caption{The natural transformation $\bar M^A$}
    \label{fig:nat-trans-barMA}
    \[
      \bar M^D
      =
      \begin{tikzcd}[column sep=small]
        \stsupp{\ups\sthcomp(\tau\sthcomp\sigma)}
        \ar[d,"\tilde \pr^{\ups \sthcomp (\tau \sthcomp \sigma)}"']
        \ar[r,"\finv {\stsupp{\stass_{\sigma,\tau,\ups}}}"]
        \cphar[rd,"="]
        &
        \stsupp{(\ups\sthcomp\tau)\sthcomp\sigma}
        \ar[r,"\stsupp M"]
        \ar[d,"\pr^{(\ups \sthcomp \tau) \sthcomp \sigma}"{description}]
        \cphar[rd,"\xTo{\tilde M^D}"]
        &
        \stsupp{(\ups'\sthcomp\tau')\sthcomp\sigma'}
        \ar[d,"\pr^{(\ups' \sthcomp \tau') \sthcomp \sigma'}"{description}]
        \ar[r,"\stsupp{\stass_{\sigma',\tau',\ups'}}"]
        \cphar[rd,"="]
        &
        \stsupp{\ups'\sthcomp(\tau'\sthcomp\sigma')}
        \ar[d,"\tilde\pr^{\ups' \sthcomp (\tau' \sthcomp \sigma')}"]
        \\
        \stsupp{\ups \sthcomp \tau}
        \ar[r,equals]
        \ar[d,"\pr^{\ups\sthcomp\tau}"']
        \cphar[rd,"="]
        &
        \stsupp{\ups \sthcomp \tau}
        \ar[r,"\stsupp L"{description},myname=Garr]
        \ar[d,"\pr^{\ups\sthcomp\tau}"{description}]
        \cphar[rd,"\xTo{\tilde L^D}"]
        &
        \stsupp{\ups' \sthcomp \tau'}
        \ar[d,"\pr^{\ups'\sthcomp\tau'}"{description}]
        \ar[r,equals]
        \cphar[rd,"="]
        &
        \stsupp{\ups' \sthcomp \tau'}
        \ar[d,"\pr^{\ups'\sthcomp\tau'}"]
        \\
        \stsupp\ups
        \ar[d,"\stdisp\ups_D"']
        \ar[r,equals]
        \cphar[rd,"="]
        &
        \stsupp\ups
        \ar[d,"\stdisp\ups_D"{description}]
        \ar[r,"H"{description},myname=mididB]
        \cphar[rd,"\xTo{H^D}"]
        &
        \stsupp{\ups'}
        \ar[d,"\stdisp{\ups'}_D"{description}]
        \ar[r,equals]
        \cphar[rd,"="]
        &
        \stsupp{\ups'}
        \ar[d,"\stdisp{\ups'}_D"]
        \\
        D
        \ar[r,equals]
        &
        D
        \ar[r,equals]
        &
        D
        \ar[r,equals]
        &
        D
      \end{tikzcd}
    \]
    \caption{The natural transformation $\bar M^D$}
    \label{fig:nat-trans-barMD}
  \end{figure*}
  By using the \PCC to characterize the composition of $G\sthc F$ with $H$, we
  have that $H\sthc (G\sthc F) = \bar M$, so that
  \[
    \stass_{\sigma',\tau',\ups'}\circ ((H\sthc G)\sthc F)
    =
    (H\sthc (G\sthc F))\circ \stass_{\sigma,\tau,\ups}
  \]
  which was the wanted naturality.
\end{proof}
\endgroup
We can now prove \Cref{thm:thin-bicat}:
\begin{proof}[Proof of \Cref{thm:thin-bicat}]
  By \Cref{lem:lunit-natural,lem:lunit-natural,lem:assoc-natural}, the
  $0$\composition is naturally left unital, right unital and associative.
  Moreover, the coherence conditions on the natural isomorphisms, required by
  the definition of bicategories, directly follow from their pullback
  definitions.
\end{proof}

\subsection{Renamings}
\label{sec:renamings}

\begin{proposition}
  Given $F : A \to B \in \Ren$, $\check F \in \Thin$.
\end{proposition}
\begin{proof}
  We first prove that $\check F \in \U_{B \lin A}$ and we use the dual version
  \Cref{prop:carac_lin} for this purpose. We already have that $\id_A \in \U_A$
  since it is an isomorphism. We are left to show that $\dual{\check F}@S \in
  \U_{B^\perp}$ for every $S \in \U_A$. Up to isomorphism of domain,
  $\dual{\check F}@S$ is the composition $F \circ \display^S$ and, by hypothesis
  on $F$, the latter is in $\U_{B^\perp}$. So $\check F\in \U_{B \lin A}$ by
  \Cref{prop:carac_lin}.

  We are left to show that $\check F \in \T_{B \lin A}$. But it follows from
  \Cref{prop:carac_lin_thin} by the same arguments as for uniformity.
\end{proof}

Given $F \co A \to B$ and $G \co B \to C$ in $\Ren$, there is 
\[
  \prtospcoh F G \co \check{(GF)} \To \check F \sthc \check G
\]
a strong morphism
of $\Thin$, defined by the universal property of the pullback as
\[
  \stsupp{\prtospcoh F G} = \pair{F,\id_A} \co A \to \stsupp{\check F \sthc
    \check G} \zbox.
\]

\begin{lemma}
  \label{lem:charact-sthcomp-prtosp-2cells}
  Let $A,B,C$ be thin groupoids, and $\phi \co F \To F' \co A \to B$ and $\psi
  \co G \To G'\co B \to C$ be two $2$\cells of $\Ren$. The composition $\check
  \phi \sthc \check \psi \co \check{F} \sthc \check{G} \To\check{F'} \sthc
  \check{G'}$ is given by $(H,\chi^C,\chi^A)$ where $\chi^C$ and $\chi^A$ are
  respectively
  \[
    \begin{tikzcd}
      \stsupp{\check{F} \sthc \check{G}}
      \ar[r,"H"]
      \ar[d,"\pl"']
      \cphar[rd,"\xTo{\bar \phi}"]
      &
      \stsupp{\check{F'} \sthc \check{G'}}
      \ar[d,"\pl"]
      \\
      B
      \ar[r,equals]
      \ar[d,"G"']
      \cphar[rd,"\xTo\psi"]
      &
      B
      \ar[d,"G'"]
      \\
      C
      \ar[r,equals]
      &
      C
    \end{tikzcd}
    \qand
    \begin{tikzcd}
      \stsupp{\check{F} \sthc \check{G}}
      \ar[r,"H"]
      \ar[d,"\pr"']
      \cphar[rd,"="]
      &
      \stsupp{\check{F'} \sthc \check{G'}}
      \ar[d,"\pr"]
      \\
      A
      \ar[r,equals]
      \ar[d,"\id_A"']
      \cphar[rd,"="]
      &
      A
      \ar[d,"\id_A"]
      \\
      A
      \ar[r,equals]
      &
      A
    \end{tikzcd}
  \]
  with
  \[
    \bar \phi
    =
    \begin{tikzcd}
      \stsupp{\check{F} \sthc \check{G}}
      \ar[r,"\pr"]
      \ar[d,"\pl"']
      \cphar[rd,"="]
      &
      A
      \ar[r,equals]
      \ar[d,"F"{description}]
      \cphar[rd,"\xTo{\phi}"]
      &
      A
      \ar[r,"\finv\pr"]
      \ar[d,"F'"{description}]
      \cphar[rd,"="]
      &
      \stsupp{\check{F'} \sthc \check{G'}}
      \ar[d,"\pl"]
      \\
      B
      \ar[equals,r]
      &
      B
      \ar[equals,r]
      &
      B
      \ar[equals,r]
      &
      B
    \end{tikzcd}
  \]
  and $H$ as on the top of $\bar \phi$.
\end{lemma}
\begin{proof}
  This is a consequence of the \PCC.
\end{proof}

\begin{proposition}
  Given thin groupoids $A,B,C$, the $2$\cells $\prtospcoh F G$ for $F \co A \to
  B$ and $G \co B \to C$ in $\Thin$ define a natural iso $\prtospcohe$
  of type
  \[
    \begin{array}{c}
      \prtosp{((-)_{(2)} \circ (-)_{(1)})}
      \To
      \prtosp-_{(1)} \sthc \prtosp-_{(2)}
      \\
      \co \Ren(A,B)\times \Ren(B,C) \to \Thin(C,A)
    \end{array}
    \zbox.
  \]
\end{proposition}
\begin{proof}
  Let $\phi \co F \To F' \co A \to B$ and $\psi \co G \To G'\co B \to C$
  be two $2$\cells of $\Ren$. 

We must show that
  \[
    \begin{tikzcd}
      \prtosp{(GF)}
      \ar[r,"\prtosp{(\psi\phi)}"]
      \ar[d,"\prtospcoh{F}{G}"']
      \cphar[rd,"="]
      &
      \prtosp{(G'F')}
      \ar[d,"\prtospcoh{F'}{G'}"]
      \\
      \prtosp{F} \sthcomp \prtosp{G}
      \ar[r,"\prtosp{\phi} \sthcomp \prtosp{\psi}"']
      &
      \prtosp{F'} \sthcomp \prtosp{G'}
    \end{tikzcd}
  \]
  in $\Thin(C,A)$. But this equation can easily be deduced from
  \Cref{lem:charact-sthcomp-prtosp-2cells}, whose statement implies that
  \[
    \prtosp{\phi} \sthcomp \prtosp{\psi}
    =
    \prtospcoh{F'}{G'}
    \circ
    \prtosp{(\psi\phi)}
    \circ
    \finv{\prtospcoh{F}{G}}
    \qedhere
  \]
\end{proof}
\noindent We can now finish the proof of \Cref{prop:check-psfunctor}:
\begin{proof}
  For every $A$ and $B$, $\prtosp-$ can easily be seen to define a functor
  $\prtosp-_{A,B} \co \Ren(A,B) \to \Thin(B,A)$. We are just left to show
  that the usual coherence conditions for pseudofunctors are satisfied by
  $\prtospcohe$. But the required coherence conditions follow directly from the
  universal property of the pullback.
\end{proof}

\subsection{The $\oc$ functor and its structure}
\label{sec:app-oc-functor}
\begingroup \renewcommand\rho{R}%
\renewcommand\sigma{S}%
\renewcommand\tau{T}%
We now finish the definition of the pseudofunctor $\oc \co \Thin \to \Thin$.
First, while we described weak morphisms between two spans $S,S' \co A \stto B$
as triples $(F,F^A,F^B)$, often identifying $F$ with the whole triple, we will
in the following often refer to the first element of the triple $F$ by
$\stsuppreally F$, for disambiguation. We now start by proving the naturality of
the coherence $\pshccohe$.
\begin{lemma}
  \label{lem:charact-sthcomp-oc}
  Let $F=(\stsuppreally F,F^A,F^B) \co \sigma \To \sigma' \co \cA \stto \cB$ and
  $G=(\stsuppreally G,G^B,G^C) \co \tau \To \tau' \co \cB \stto \cC$. Let
  $\chi^\sigma$ and $\chi^\tau$ be two $2$\cells given by the definition of
  horizontal composition so that $G\sthcomp F$ is given by the two $2$\cells
  \[
    \begin{tikzcd}
      \stsupp{\tau \sthcomp \sigma}
      \ar[r,"\stsuppreally{G\sthcomp F}"]
      \ar[d,"\pl"']
      \cphar[rd,"\xTo{\chi^\sigma}"]
      &
      \stsupp{\tau' \sthcomp \sigma'}
      \ar[d,"\pl"]
      \\
      \stsupp\sigma
      \ar[r,"\stsuppreally F"{description}]
      \ar[d,"\stdisp\sigma_A"']
      \cphar[rd,"\xTo{F^A}"]
      &
      \stsupp{\sigma'}
      \ar[d,"\stdisp{\sigma'}_A"]
      \\
      A
      \ar[r,"\unit {A}"']
      &
      A
    \end{tikzcd}
    \qand
    \begin{tikzcd}
      \stsupp{\tau \sthcomp \sigma}
      \ar[r,"\stsuppreally{G\sthcomp F}"]
      \ar[d,"\pr"']
      \cphar[rd,"\xTo{\chi^\tau}"]
      &
      \stsupp{\tau' \sthcomp \sigma'}
      \ar[d,"\pr"]
      \\
      \stsupp\tau
      \ar[r,"\stsuppreally G"{description}]
      \ar[d,"\stdisp\tau_C"']
      \cphar[rd,"\xTo{G^C}"]
      &
      \stsupp{\tau'}
      \ar[d,"\stdisp{\tau'}_C"]
      \\
      C
      \ar[r,"\unit {C}"']
      &
      C
    \end{tikzcd}
    \zbox.
  \]
  The composition $\oc G \sthcomp \oc F$ is then given by the two $2$\cells
  \[
    \begin{tikzcd}
      \stsupp{\oc\tau\sthcomp\oc\sigma}
      \ar[r,"\finv{\pshccoh{\sigma}{\tau}}"]
      \ar[d,"\pl"']
      \cphar[rd,"="]
      &[-0.3em]
      \oc(\stsupp{\tau \sthcomp \sigma})
      \ar[r,"\oc(\stsuppreally{G\sthcomp F})"]
      \ar[d,"\oc\pl"{description}]
      \cphar[rd,"\xTo{\oc\chi^\sigma}"]
      &[-0.3em]
      \oc(\stsupp{\tau' \sthcomp \sigma'})
      \ar[r,"\pshccoh{\sigma'}{\tau'}"]
      \ar[d,"\oc\pl"{description}]
      \cphar[rd,"="]
      &[-0.3em]
      \stsupp{\oc\tau'\sthcomp\oc\sigma'}
      \ar[d,"\pl"]
      \\
      \oc\stsupp\sigma
      \ar[r,"\unit{\oc\stsupp\sigma}"{description}]
      \ar[d,"\oc\stdisp\sigma_A"{description}]
      \cphar[rd,"="]
      &
      \oc\stsupp\sigma
      \ar[r,"\oc \stsuppreally F"{description}]
      \ar[d,"\oc\stdisp\sigma_A"{description}]
      \cphar[rd,"\xTo{\oc F^A}"]
      &
      \oc\stsupp{\sigma'}
      \ar[r,"\unit{\oc\stsupp{\sigma'}}"{description}]
      \ar[d,"\oc\stdisp{\sigma'}_A"{description}]
      \cphar[rd,"="]
      &
      \oc\stsupp\sigma'
      \ar[d,"\oc\stdisp{\sigma'}_A"]
      \\
      \oc A
      \ar[r,"\oc\unit {A}"']
      &
      \oc A
      \ar[r,"\oc\unit {A}"']
      &
      \oc A
      \ar[r,"\oc\unit {A}"']
      &
      \oc A
    \end{tikzcd}
  \]
  and
  \[
    \begin{tikzcd}
      \stsupp{\oc\tau\sthcomp\oc\sigma}
      \ar[r,"\finv{\pshccoh{\sigma}{\tau}}"]
      \ar[d,"\pr"']
      \cphar[rd,"="]
      &[-0.3em]
      \oc(\stsupp{\tau \sthcomp \sigma})
      \ar[r,"\oc(\stsuppreally{G\sthcomp F})"]
      \ar[d,"\oc\pr"{description}]
      \cphar[rd,"\xTo{\oc\chi^\tau}"]
      &[-0.3em]
      \oc(\stsupp{\tau' \sthcomp \sigma'})
      \ar[r,"\pshccoh{\sigma'}{\tau'}"]
      \ar[d,"\oc\pr"{description}]
      \cphar[rd,"="]
      &[-0.3em]
      \stsupp{\oc\tau'\sthcomp\oc\sigma'}
      \ar[d,"\pr"]
      \\
      \oc\stsupp\tau
      \ar[r,"\unit{\oc\stsupp\tau}"{description}]
      \ar[d,"\oc\stdisp\tau_C"{description}]
      \cphar[rd,"="]
      &
      \oc\stsupp\tau
      \ar[r,"\oc \stsuppreally G"{description}]
      \ar[d,"\oc\stdisp\tau_C"{description}]
      \cphar[rd,"\xTo{\oc G^C}"]
      &
      \oc\stsupp{\tau'}
      \ar[r,"\unit{\oc\stsupp{\tau'}}"{description}]
      \ar[d,"\oc\stdisp{\tau'}_C"{description}]
      \cphar[rd,"="]
      &
      \oc\stsupp\tau'
      \ar[d,"\oc\stdisp{\tau'}_C"]
      \\
      \oc C
      \ar[r,"\oc\unit {C}"']
      &
      \oc C
      \ar[r,"\oc\unit {C}"']
      &
      \oc C
      \ar[r,"\oc\unit {C}"']
      &
      \oc C
    \end{tikzcd}
    \zbox.
  \]
\end{lemma}
\begin{proof}
  We use the \PCC. The respective positivity and negativity of the proposed
  $2$\cells follow from the positivity of the vertical composition of
  $\chi^\sigma$ and $F^A$ and the negativity of the vertical composition of
  $\chi^\tau$ and $G^C$. We are left to show the top row of the two proposed
  $2$\cells satisfy the equality required by the \PCC, but it follows
  from that satisfied by $\chi^\sigma$ and $\chi^\tau$ by the functoriality of
  $\oc \co \Gpd \to \Gpd$ on $2$\cells.
\end{proof}

We can now conclude naturality:

\begin{lemma}
  \label{lem:pshccohe-nat-isom}
  The morphisms $\pshccohe^{\cA,\cB,\cC}_{\sigma,\tau}$ define a natural iso
  \begin{gather*}
    m^{\cA,\cB,\cC} \co \oc ((-) \sthcomp (-)) \To \oc (-) \sthcomp \oc (-)\\
    \co \Thin(\cA,\cB) \times \Thin(\cB,\cC) \to \Thin(\cA,\cC)
    \zbox.
  \end{gather*}
\end{lemma}
\begin{proof}
  Let $F \co \sigma \To \sigma' \co \cA \stto \cB$ and $G \co \tau
  \To \tau' \co \cB \stto \cC$. 

We must show that
  \[
    m^{\cA,\cB,\cC}_{\sigma',\tau'} \circ \oc(G \sthcomp F) =
    (\oc G \sthcomp \oc F) \circ m^{\cA,\cB,\cC}_{\sigma,\tau}
    \zbox.
  \]

  But it directly follows from \Cref{lem:charact-sthcomp-oc}, whose conclusion
  states in particular that
\[
     m^{\cA,\cB,\cC}_{\sigma',\tau'} \circ \oc(G \sthcomp F)
     \circ \finv{(m^{\cA,\cB,\cC}_{\sigma,\tau})}
    =
       \oc G \sthcomp \oc F 
       \zbox.
       \tag*{\qedhere}
\]
\end{proof}

We can thus conclude that $\oc$ is a pseudofunctor:
\begin{proposition}
  \label{prop:oc-pseudofunctor}
  The functor $\oc \co \Gpd \to \Gpd$ induces a pseudofunctor $\oc \co \Thin \to
  \Thin$.
\end{proposition}
\begin{proof}
  By \Cref{lem:pshccohe-nat-isom}, we have an adequate natural isomorphism
  expressing the functoriality of $\oc$ on $\Thin$. The coherence laws for
  pseudofunctors can be directly verified by the universal properties of the
  pullbacks involved in the horizontal compositions appearing in these laws.
\end{proof}
\endgroup

\subsection{The $\oc$ pseudocomonad}
\label{sec:app-oc-pseudocomonad}
\newcommand\defpmfletter{{\mathbf{H}}}%
\newcommand\otherpmfletter{{\mathbf{K}}}%

We are going to derive the $\oc$ pseudocomonad on $\Thin$ from the $\oc$
pseudomonad on $\Gpd$ through functoriality. Before using this functoriality
argument, we need to describe what are the (higher) categories we are going to
apply it to. The domain (bi)category\simon{non! ici c'est une 2-cat stricte}
will be the one of endofunctors on $\Gpd$ with properties similar to the ones of
$\oc \co \Gpd \to \Gpd$, while the codomain (bi)category will be the one of
endopseudofunctors on $\Thin$. We shall first discuss how to relate some
functors on $\Gpd$ to pseudofunctors on $\Thin$.

Given a functor $\defpmfletter \co \Gpd \to \Gpd$, there is a canonical uniform groupoid
$\defpmfletter\cA = (\defpmfletter A,\U_{\defpmfletter\cA})$ associated to any uniform groupoid $\cA$, where
$\U_{\defpmfletter\cA} = \set{\defpmfletter S \mid  S \in
\U_\cA}^{\perp\perp}$.\pierre{Bonne idée d'appeler le foncteur $S$? Et
les stratégies $\sigma$?}
\begin{proposition}
  If $\defpmfletter$ preserves pullbacks, and pullbacks which are bipullbacks, then given
  uniform groupoids $\cA$ and $\cB$, and $ S \in \U_{\cA \linto \cB}$, we
  have $\defpmfletter S \in \U_{\defpmfletter\cA \linto \defpmfletter\cB}$.
\end{proposition}
\begin{proof}
  This is a direct consequence of \Cref{prop:carac_lin}.
\end{proof}
We call \emph{bicartesian functors} the functors $\defpmfletter$ which satisfy the
hypothesis of the above property.

A \emph{\pmfunctor} is a tuple $(\defpmfletter,\defpmfletter^+,\iota)$ with
$\defpmfletter,\defpmfletter^+$ being functors $\Gpd \to \Gpd$ where
$\defpmfletter$ and $\defpmfletter^+$ are bicartesian and preserve functors
(between groupoids) that are bijective on objects (of the groupoids)\pierre{Ça
  veut dire quoi?}\simon{détails ok?}, and such that $\defpmfletter^+$ preserves
discrete groupoids, and $\iota \co \defpmfletter^+ \To \defpmfletter$ being a
natural transformation which is pointwise (that is, such that each $\iota_X$ is)
monomorphic and surjective on objects (of the groupoids)\pierre{Not sure what it
  means}\simon{détails ok?}, satisfying moreover that it is \emph{bicartesian},
meaning that its naturality squares are both pullbacks and bipullbacks.
Intuitively, the definition of \pmfunctor is an abstraction of the case of $\oc
\co \Gpd \to \Gpd$, from which we derive a functor $\Thin \to \Thin$. In the
case of $\oc$, given a thin groupoid $A$, a positive sub-groupoid $(\oc A)_+$ is
defined from a construction which is not derivable from the definition of $\oc$
and the data of $A$ and $A_+$, so that we have to take it into account in our
definition of \pmfunctor, in the form of a functor $\defpmfletter^+$ and a
natural transformation $\iota \co \defpmfletter^+ \To \defpmfletter$. We should
\emph{a priori} also require similar data for the negative side, but it so
happens that, in the case of $\oc$, $(\oc A)_- = \oc (A_-)$, so that it is in
fact not necessary. In order to show that $\oc$ induces a pseudocomonad on
$\Thin$, the \pmfunctors we will consider will only be iterated compositions of
$\oc$.

Given a \pmfunctor $(\defpmfletter,\defpmfletter^+,\iota)$ and a thin groupoid
$\cA$, there is a canonical thin groupoid $\defpmfletter \cA$ whose underlying
uniform groupoid is defined as earlier, whose class of thin prestrategies is
$\T_{\defpmfletter\cA} = \set{\defpmfletter S \mid S \in
  \T_\cA}^{\pperp\pperp}$, and whose negative and positive sub-groupoids are
$(\defpmfletter \cA)_- = \defpmfletter A_-$ and $(\defpmfletter\cA)_+ =
\defpmfletter^+A_+$ with embeddings given by the compositions
\[
  \defpmfletter A_-
  \xto{\defpmfletter(\id^-_A)}
  \defpmfletter A
  \qand
  \defpmfletter^+A_+
  \xto{\defpmfletter^+(\id^+_A)}
  \defpmfletter^+A
  \xto{\iota_A}
  \defpmfletter A
  \zbox.
\]

By the conditions of \pmfunctors\pierre{What is a thin functor?}\simon{pardon,
  je voulais dire \pmfunctor}, they are elements of $\T_{\defpmfletter\cA}$ and
$\T_{\defpmfletter\cA}^{\pperp}$ as required (exercise to the reader).
\begin{proposition}
  \label{prop:pmfunctor-span-img}
  Given a \pmfunctor $(\defpmfletter,\defpmfletter^+,\iota)$ and thin groupoids $\cA$ and $\cB$, and
  $ S \in \T_{\cA \linto \cB}$, $\defpmfletter S \in \T_{\defpmfletter\cA \linto \defpmfletter\cB}$.
\end{proposition}
\begin{proof}
  This is an easy consequence of the hypotheses on a \pmfunctor and
  \Cref{prop:carac_lin,prop:carac_lin_thin}.
\end{proof}
\begin{proposition}
  \label{prop:pmfunctor-pres-polarity}
  Given a \pmfunctor $(\defpmfletter,\defpmfletter^+,\iota)$ and a thin groupoid $\cA$, $\defpmfletter$ preserves
  negative (\resp positive) $2$\cells.
\end{proposition}
\begin{proof}
  Given a negative $2$\cell $\phi \co F \To F' \co X \to A$, writing $X_{(0)}$
  for the discrete groupoid with the same object as $X$, we have a commutative
  diagram
  \[
    \begin{tikzcd}
      X_{(0)}
      \ar[r,hookrightarrow,"e"]
      \phar[d,"\xTo{\phi^-}"]
      \ar[d,"\bar F"'{pos=0.48},bend right]
      \ar[d,"\bar F'"{pos=0.48},bend left]
      &
      X
      \phar[d,"\xTo{\phi}"]
      \ar[d,"F"'{pos=0.51},bend right]
      \ar[d,"F'"{pos=0.51},bend left]
      \\
      A_-
      \ar[r,"\id^-_A"']
      &
      A
    \end{tikzcd}
  \]
  for some $2$\cell $\phi^- \co \bar F \To \bar F'$, where the top arrow $e$
  is the canonical embedding. By hypothesis on $\defpmfletter$, the image of $e$ by $\defpmfletter$
  is bijective on objects of $X$. Moreover, by functoriality, we have
  $\defpmfletter(\id^-_A) \circ \defpmfletter(\phi^-) = \defpmfletter(\phi) \circ \defpmfletter(e)$. Thus, all the
  components of the natural transformation $\defpmfletter(\phi)$ are in the image of
  $\defpmfletter(\id^-_A)$. Thus, it is negative.

  Now, given a positive $2$\cell $\phi \co F \To F' \co X \to A$, we have a
  similar commutative diagram
  \[
    \begin{tikzcd}
      X_{(0)}
      \ar[r,hookrightarrow,"e"]
      \phar[d,"\xTo{\phi^+}"]
      \ar[d,"\bar F"'{pos=0.48},bend right]
      \ar[d,"\bar F'"{pos=0.48},bend left]
      &
      X
      \phar[d,"\xTo{\phi}"]
      \ar[d,"F"'{pos=0.51},bend right]
      \ar[d,"F'"{pos=0.51},bend left]
      \\
      A_+
      \ar[r,"\id^+_A"']
      &
      A
    \end{tikzcd}
  \]
  for some $2$\cell $\phi^+ \co \bar F \To \bar F'$. We then have the
  commutative diagram
  \[
    \begin{tikzcd}[csbo=8em]
      \defpmfletter^+X_{(0)}
      \ar[r,"\defpmfletter^+(e)"]
      \phar[d,"\xTo{\defpmfletter^+(\phi^+)}"]
      \ar[d,"\defpmfletter^+(\bar F)"'{pos=0.48},bend right=60]
      \ar[d,"\defpmfletter^+(\bar F')"{pos=0.48},bend left=60]
      &
      \defpmfletter^+X
      \ar[r,"\iota_X"]
      \phar[d,"\xTo{\defpmfletter^+(\phi)}"]
      \ar[d,"\defpmfletter^+(F)"'{pos=0.48},bend right=60]
      \ar[d,"\defpmfletter^+(F')"{pos=0.48},bend left=60]
      &
      \defpmfletter X
      \phar[d,"\xTo{\defpmfletter(\phi)}"]
      \ar[d,"\defpmfletter(F)"'{pos=0.52},bend right=60]
      \ar[d,"\defpmfletter(F')"{pos=0.52},bend left=60]
      \\
      \defpmfletter^+A_+
      \ar[r,"\defpmfletter^+(\id^+_A)"']
      &
      \defpmfletter^+A
      \ar[r,"\iota_A"']
      &
      \defpmfletter A
    \end{tikzcd}
  \]
  where the top arrow is bijective on objects by assumptions. Thus, all the
  components of $\defpmfletter(\phi)$ are in the image of $\iota_A \circ \defpmfletter^+(\id^+_A)$, so
  that $\defpmfletter(\phi)$ is positive.
\end{proof}
Given two \pmfunctors $\defpmfletter = (\defpmfletter,\defpmfletter^+,\iota)$ and $\otherpmfletter = (\otherpmfletter,\otherpmfletter^+,\kappa)$, a
\emph{\pmtransformation} is a pair $(\alpha,\alpha^+)$ of bicartesian natural
transformations where $\alpha \co \defpmfletter \To \otherpmfletter$ and $\alpha^+ \co \defpmfletter^+ \To \otherpmfletter^+$ are
such that
\[
  \kappa \circ \alpha^+ = \alpha \circ \iota \zbox.
\]
Now, a \emph{\pmmodification} between two such \pmtransformations
$(\alpha,\alpha^+)$ and $(\beta,\beta^+)$ is the data of a modification $m \co
\alpha \TO \beta$ in the $3$\category of $2$\categories.
\begin{proposition}
  Given a thin modification $m \co \alpha \TO \beta \co \defpmfletter \To \otherpmfletter$ and a thin
  groupoid $\cA$, the $2$\cell
  \[
    m_A \co \alpha_A \To \beta_A \co \defpmfletter A \to \otherpmfletter  A
  \]
  is negative as a $2$\cell on $\otherpmfletter \cA$.
\end{proposition}
\begin{proof}
  Since $m$ is a modification, we have that $m_A \circ \defpmfletter(\id^-_A) = \otherpmfletter (\id^-_A)
  \circ m_{A_-}$. Since $\defpmfletter(\id^-_A)$ is bijective on objects by hypotheses on
  $\defpmfletter$ and $\id^-_A$, we have that all the components of $m_A$ are in the image
  of $\otherpmfletter (\id^-_A)$, so that they are negative.
\end{proof}
\pmfunctors, \pmtransformations and \pmmodifications can be equipped with the
evident operations in order to form a strict $3$\category $\pmFunct$ with one
object (which, morally, is $\Gpd$, the domain and codomain of each \pmfunctor).
The important point to note here is that the data of $\oc$, $\eta$, $\mu$,
$\oclu$, $\ocru$ and $\ocass$ induces the expected way a pseudomonad in
$\pmFunct$.

We shall now describe an operation $\check{(-)}$ relating $\pmFunct$ and
$\Thin$. First, a \pmfunctor $(\defpmfletter,\defpmfletter^+,\iota)$ induces an endofunctor $\check \defpmfletter
\co \Thin \to \Thin$ mapping a thin groupoid $\cA$ to the thin groupoid $\defpmfletter\cA$
defined as earlier, and mapping spans and their weak morphisms to their images
by $\defpmfletter$, which is well-defined by
\Cref{prop:pmfunctor-span-img,prop:pmfunctor-pres-polarity}.

Now, given a \pmtransformation $(\alpha,\alpha^+) \co
(\defpmfletter,\defpmfletter^+,\iota) \To (\otherpmfletter ,\otherpmfletter
^+,\kappa)$, we define a pseudonatural transformation $\check\alpha \co \check
\otherpmfletter \To \check \defpmfletter$ whose component at a thin groupoid
$\cA$ is $\check{(\alpha_A)}$, that is, the image of $\alpha_A \co \defpmfletter
A \to \otherpmfletter A$ by the pseudofunctor $\check{(-)} : \Ren^{\op} \to
\Thin$.

We can then extend the pseudofunctor $\check{-} : \Ren^{\op} \to \Thin$ to some
sort of $3$\dimensional functor
\[
  \check{(-)} : \pmFunct^\cocat \to \Bicat
\]
sending the unique $0$\cell to $\Thin$ and the higher cells in the
hom-bicategory $\Bicat(\Thin,\Thin)$ (here, $\pmFunct^\cocat$ denotes $\pmFunct$
with $2$\cells reversed).

While this is probably a trifunctor, it would be very tiresome to prove.
Instead, we will only rely on the simpler proposition stating that
\begin{proposition}
  \label{prop:pmfunct-bifunctor}
  Considering $\pmFunct$ as a strict $2$\category by forgetting the dimension
  $0$, $\check{(-)}$ induces a pseudofunctor
  \[
    \check{(-)} : \pmFunct^\cocat \to \Bicat(\Thin,\Thin)
  \]
  between bicategories.
\end{proposition}
\begin{proof}
  By checking the axioms of pseudofunctors.
\end{proof}
\noindent We can now briefly describe a proof of \Cref{thm:oc-pseudomonad}:
\begin{proof}[Proof of \Cref{thm:oc-pseudomonad}]
  The most satisfying proof of this statement would rely on the fact that
  \[
    \check{(-)} : \pmFunct^\cocat \to \Bicat(\Thin,\Thin)
  \]
  is a trifunctor and that a trifunctor sends any pseudocomonad to a
  pseudocomonad, but we do not know a proof for the latter fact (though it is
  probably true) and deem a full proof of the former tedious.

  Instead, we can rely on the weaker \Cref{prop:pmfunct-bifunctor} to prove the
  required coherences. Following~\cite{cheng2003pseudo}, we are required to
  prove the equations of modifications of \Cref{fig:oc-pseudocomonad-eqs} are
  verified.
  \begin{figure*}
    \centering
    \begin{gather*}
      \ocd\psm \sthcomp (\oc \psm \sthc \psm)
      \TO
      (\ocd\psm \sthcomp \oc \psm) \sthc \psm
      \TO
      \oc(\oc\psm \sthcomp \psm) \sthc \psm
      \TO
      \oc(\psm\oc\sthcomp \psm) \sthc \psm
      \TO
      (\oc\psm\oc\sthcomp \oc\psm) \sthc \psm
      \TO
      \oc\psm\oc\sthcomp (\oc\psm \sthc \psm)
      \\
      \TO
      \oc\psm\oc\sthcomp (\psm \oc \sthc \psm)
      \TO
      (\oc\psm\oc\sthcomp \psm \oc) \sthc \psm
      =
      (\oc\psm\sthcomp \psm )\oc \sthc \psm
      \TO
      (\psm\oc\sthcomp \psm )\oc \sthc \psm
      =
      (\psm\ocd\sthcomp \psm\oc ) \sthc \psm
      \\
      =
      \\
      \ocd\psm \sthcomp (\oc \psm \sthc \psm)
      \TO
      \ocd\psm \sthcomp (\psm \oc \sthc \psm)
      \TO
      (\ocd\psm \sthcomp \psm \oc) \sthc \psm
      \xTO{\text{\rm ex}}
      (\psm \ocd \sthcomp \oc\psm) \sthc \psm
      \TO
      \psm \ocd \sthcomp (\oc\psm \sthc \psm)
      \TO
      \psm \ocd \sthcomp (\psm \oc\sthc \psm)
      \TO
      (\psm \ocd \sthcomp \psm \oc)\sthc \psm
      \\
      \text{and}
      \\
      \oc \psu \oc \sthc (\oc \psm \sthc \psm)
      \TO
      \oc \psu \oc \sthc (\psm\oc  \sthc \psm)
      \TO
      (\oc \psu \oc \sthc \psm\oc)  \sthc \psm
      =
      (\oc \psu  \sthc \psm)\oc  \sthc \psm
      \TO
      \stid\oc\oc \sthc \psm
      =
      \stid\ocd \sthc \psm
      \TO
      \psm
      \\
      =
      \\
      \oc \psu \oc \sthc (\oc \psm \sthc \psm)
      \TO
      (\oc \psu \oc \sthc \oc \psm) \sthc \psm
      \TO
      \oc(\psu \oc \sthc \psm) \sthc \psm
      \TO
      \oc\stid\oc \sthc \psm
      =
      \stid\ocd \sthc \psm
      \TO
      \psm
    \end{gather*}
    \caption{The two required equations for $\oc$ to be a pseudocomonad}
    \label{fig:oc-pseudocomonad-eqs}
  \end{figure*}
  The idea is to relate each of these equations to the equations satisfied by
  the pseudomonad $\oc \co \Gpd \to \Gpd$, and this is done through paving. For
  example, we use the following pavings for the two first modifications of the
  left hand-side of the first equation of \Cref{fig:oc-pseudocomonad-eqs}:
  \[
    \begin{tikzcd}
      \ocd\psm \sthcomp (\oc \psm \sthc \psm)
      \ar[r,Rrightarrow]
      \cphar[rd,"="]
      &
      (\ocd\psm \sthcomp \oc \psm) \sthc \psm
      \\
      \prtospw{(\ocd\mu)} \sthcomp (\prtospw{(\oc \mu)} \sthc \psm)
      \ar[u,equals]
      \tar[r]
      \cphar[rdd,"="]
      &
      (\prtospw{(\ocd\mu)} \sthcomp \prtospw{(\oc \mu)}) \sthc \psm
      \ar[u,equals]
      \\
      \prtospw{(\ocd\mu)} \sthcomp \prtospw{(\mu \circ \oc \mu)}
      \tar[u]
      &
      \prtospw{(\oc \mu \circ \ocd\mu)} \sthc \psm
      \tar[u]
      \\
      \prtospw{((\mu \circ \oc \mu) \circ (\ocd\mu))}
      \tar[r,equals]
      \tar[u]
      &
      \prtospw{(\mu \circ (\oc \mu \circ (\ocd\mu)))}
      \tar[u]
    \end{tikzcd}
  \]
  and
  \[
    \begin{tikzcd}[csbo=8em]
      (\ocd\psm \sthcomp \oc \psm) \sthc \psm
      \tar[rr]
      \cphar[rrd,"="]
      &&
      \oc(\oc\psm \sthcomp \psm) \sthc \psm
      \\
      (\oc\prtospw{(\oc\mu)} \sthcomp \oc \psm) \sthc \psm
      \ar[u,equals]
      \tar[rr]
      \cphar[rrdd,"="]
      &&
      \oc(\prtospw{(\oc\mu)} \sthcomp \psm) \sthc \psm
      \ar[u,equals]
      \\
      (\prtospw{(\ocd\mu)} \sthcomp \prtospw{(\oc \mu)}) \sthc \psm
      \ar[u,equals]
      \\
      \prtospw{(\oc \mu \circ \ocd\mu)} \sthc \psm
      \tar[u]
      \ar[r,equals]
      \cphar[rrddd,"="]
      &
      \prtospw{(\oc(\mu \circ \oc\mu))} \sthc \psm
      \ar[r,equals]
      &
      \oc(\prtospw{\mu \circ \oc\mu}) \sthc \psm
      \tar[uu]
      \\
      &&
      \prtospw{(\oc(\mu \circ \oc\mu))} \sthc \psm
      \ar[u,equals]
      \\
      &&
      \prtospw{(\oc\mu \circ \ocd\mu)} \sthc \psm
      \ar[u,equals]
      \\
      \prtospw{(\mu \circ (\oc \mu \circ (\ocd\mu)))}
      \tar[uuu]
      \ar[rr,equals]
      &&
      \prtospw{(\mu \circ (\oc\mu \circ \ocd\mu))}
      \tar[u]
    \end{tikzcd}
    \zbox.
  \]
  The other elementary modifications of \Cref{fig:oc-pseudocomonad-eqs} are
  paved similarly, so that the first equation of \Cref{fig:oc-pseudocomonad-eqs}
  is reduced to the equation
  \begin{gather*}
    \prtospw{(\mu \circ \oc\mu \circ \ocd\mu)}
    \TO
    \prtospw{(\mu \circ \oc\mu \circ \oc\mu\oc)}
    \\
    \TO
    \prtospw{(\mu \circ \mu\oc \circ \oc\mu\oc)}
    \TO
    \prtospw{(\mu\circ\mu\oc\circ \mu\ocd)}
    \\
    =
    \\
    \prtospw{(\mu \circ \oc \mu \circ \ocd\mu)}
    \TO
    \prtospw{(\mu \circ \mu\oc  \circ \ocd\mu)}
    \\
    =
    \prtospw{(\mu\circ \oc\mu \circ\mu\ocd)}
    \TO
    \prtospw {(\mu\circ\mu\oc\circ\mu\ocd)}
  \end{gather*}
  which is the image by $\check{(-)}$ of an equation satisfied by the
  pseudomonad $(\oc,\eta,\mu)$ on $\Gpd$, and the second equation of
  \Cref{fig:oc-pseudocomonad-eqs} to the equation
  \begin{gather*}
    \prtospw{((\mu\circ\oc\mu)\circ \oc\eta\oc)}
    \TO
    \prtospw{((\mu\circ\mu\oc)\circ \oc\eta\oc)}
    \\
    =
    \prtospw{(\mu\circ(\mu\oc\circ \oc\eta\oc))}
    \TO
    \prtospw{(\mu \circ \unit\ocd)}
    =
    \psm
    \\
    =
    \\
    \prtospw{((\mu\circ\oc\mu)\circ \oc\eta\oc)}
    =
    \prtospw{(\mu\circ(\oc\mu\circ \oc\eta\oc))}
    \TO
    \prtospw{(\mu\circ\unit\ocd)}
    =
    \psm
  \end{gather*}
  also an image of an equation satisfied by the pseudomonad $(\oc,\eta,\mu)$ on
  $\Gpd$. So that $(\oc,\psu,\psm)$ is indeed a pseudocomonad on $\Thin$.
\end{proof}

\subsection{The cartesian product}
\label{sec:cartesian-product}
\begingroup
\newcommand\pscpl{\prtosp{\cplprot}}%
\newcommand\pscpr{\prtosp{\cprprot}}%
\newcommand\iotabone{\cpl}%
\newcommand\iotabtwo{\cpr}%
\renewcommand\rho{R}%
\renewcommand\sigma{S}%
\renewcommand\tau{T}%
\newcommand\cGamma{\Gamma}%
Given thin groupoids $\cA,\cB$, we write $\cpl^{A,B}$ and $\cpr^{A,B}$, simply
denoted $\cpl$ and $\cpr$ as earlier when $A,B$ can be deduced from the context,
for the coprojections
\begin{gather*}
  \cpl \co A \hookrightarrow A + B
  \qqand
  \cpr \co B \hookrightarrow A + B
  \zbox.
\end{gather*}
By applying the functor $\check{(-)}$, we get thin spans $\pscpl \co \cA \with
\cB \stto \cA$ and $\pscpr \co \cA \with \cB \stto \cB$. Given thin groupoids
$\cGamma,\cA,\cB$, we define a functor
\[
  \prodfact{-}{-}_{\cGamma,\cA,\cB}
  \co
  \Thin(\cGamma,\cA)
  \times
  \Thin(\cGamma,\cB)
  \to
  \Thin(\cGamma,\cA \with \cB)
  \zbox,
\]
often abbreviated $\prodfact{-}{-}$, as follows. Given
$\sigma\in\Thin(\cGamma,\cA)$ and $\tau\in\Thin(\cGamma,\cB)$, we define
$\prodfact{\sigma}{\tau}$ as the span
\[
  \prodfact{\sigma}{\tau}
  \qqeq
  \begin{tikzcd}[column sep={between origins,5em}]
    &
    \stsupp\sigma + \stsupp\tau
    \ar[ld,"\coprodfact{\stdisp\sigma_\Gamma}{\stdisp\tau_\Gamma}"']
    \ar[rd,"\stdisp\sigma_A + \stdisp\tau_B"]
    &
    \\
    \Gamma
    &&
    A +  B
  \end{tikzcd}
  \zbox.
\]
\begin{proposition}
  Given $\sigma\in\Thin(\cGamma,\cA)$ and $\tau\in\Thin(\cGamma,\cB)$, we have
  $\prodfact{\sigma}{\tau} \in \Thin(\cGamma,\cA \with \cB)$.
\end{proposition}
\begin{proof}
  By an adequate use of \Cref{prop:carac_lin,prop:carac_lin_thin}.
\end{proof}

Given morphisms $F\co \sigma \to \sigma'\in\Thin(\cGamma,\cA)$ and
$G\co\tau\to\tau'\in\Thin(\cGamma,\cB)$, $\prodfact{F}{G}$ is defined as the
morphism $H$ with $\stsuppreally H =\stsuppreally F + \stsuppreally G$ and
\begin{gather*}
  H^\Gamma
  =
  \begin{tikzcd}[ampersand replacement=\&]
    \stsupp\sigma+\stsupp\tau
    \ar[r,"\stsuppreally F+\stsuppreally G"]
    \ar[d,"\coprodfact{\stdisp\sigma_\Gamma}{\stdisp\tau_\Gamma}"']
    \cphar[rd,"\xTo{\coprodfact{F^\Gamma}{G^\Gamma}}"]
    \&
    \stsupp{\sigma'}+\stsupp{\tau'}
    \ar[d,"\coprodfact{\stdisp{\sigma'}_\Gamma}{\stdisp{\tau'}_\Gamma}"]
    \\
    \Gamma
    \ar[r,equals]
    \&
    \Gamma
  \end{tikzcd}
  \shortintertext{and}
  H^{A+ B}
  =
  \begin{tikzcd}[ampersand replacement=\&]
    \stsupp\sigma+\stsupp\tau
    \ar[r,"\stsuppreally F+\stsuppreally G"]
    \ar[d,"\stdisp\sigma_A+\stdisp\tau_B"']
    \cphar[rd,"\xTo{F^A+G^B}"]
    \&
    \stsupp{\sigma'}+\stsupp{\tau'}
    \ar[d,"\stdisp{\sigma'}_A+\stdisp{\tau'}_B"]
    \\
    A+ B
    \ar[r,equals]
    \&
    A+ B
  \end{tikzcd}
  \zbox.
\end{gather*}
One immediately verifies that these two $2$\cells have the adequate polarities,
so that $\prodfact{F}{G} \in \Thin(\cGamma,\cA\with\cB)$. Moreover,
the functoriality of $\prodfact{-}{-}_{\cGamma,\cA,\cB}$ is immediately verified.

Given $\sigma\in\Thin(\cGamma,\cA \with \cB)$, we write $\sigma_A$ for the span
\[
  \sigma_A
  =
  \begin{tikzcd}
    &&
    \stsupp{\sigma}_A
    \ar[dl,hook']
    \ar[ddrr,"\stdisp{\sigma_A}_A"]
    &&
    \\
    &
    \stsupp{\sigma}
    \ar[dl,"\stdisp\sigma_\Gamma"']
    &&
    &
    \\
    \Gamma
    &&
    &&
    A
  \end{tikzcd}
\]
where $\stsupp{\sigma}_A$ is the submonoid of $\stsupp{\sigma}$ whose image by
$\stdisp\sigma_{A+ B}$ is in $A$, and where $\stdisp{\sigma_A}_A$ is the
induced map $\stsupp\sigma_A\to A$ from $\stdisp\sigma_{A+ B}$. We define a span
$\sigma_B$ similarly.
\begin{proposition}
  \label{prop:sigmaa-sigmab-thin}
  Given $\sigma\in\Thin(\cGamma,\cA \with \cB)$, we have
  $\sigma_A\in\Thin(\cGamma,\cA)$ and $\sigma_B\in\Thin(\cGamma,\cB)$.
\end{proposition}
\begin{proof}
  As the result of the composition of two thin spans, we know that $\pscpl \sthc
  \sigma$ is in $\Thin(\cGamma,\cA)$. It is the span
  \[
    \begin{tikzcd}[column sep={5em,between origins}]
      &&
      \stsupp{\pscpl \sthc \sigma}
      \cphar[dd,near start,"\dcorner"]
      \ar[ld,"\pl"',dashed]
      \ar[rd,"\pr",dashed]
      &&
      \\
      &
      \stsupp\sigma
      \ar[ld,"\stdisp\sigma_\Gamma"']
      \ar[rd,"\stdisp\sigma_{A+ B}"{description}]
      &&
      A
      \ar[ld,"\cpl"{description}]
      \ar[rd,"\unit A"]
      &
      \\
      \Gamma
      &&
      A+ B
      &&
      A
    \end{tikzcd}
  \]
  which, by an isomorphism of pullbacks, is isomorphic to
  \[
    \begin{tikzcd}[column sep={5em,between origins}]
      &&
      \stsupp{\sigma}_A
      \cphar[dd,near start,"\dcorner"]
      \ar[ld,hook']
      \ar[rd,"\stdisp{\sigma_A}_A"]
      &&
      \\
      &
      \stsupp\sigma
      \ar[ld,"\stdisp\sigma_\Gamma"']
      \ar[rd,"\stdisp\sigma_{A+ B}"{description}]
      &&
      A
      \ar[ld,"\cpl"{description}]
      \ar[rd,"\unit A"]
      &
      \\
      \Gamma
      &&
      A+ B
      &&
      A
    \end{tikzcd}
  \]
  which is exactly $\sigma_A$. Thus, the latter is in $\Thin(\cGamma,\cA)$. A
  similar argument holds for $\sigma_B$.
\end{proof}
The mapping $\sigma \mapsto \sigma_A$ extends to a functor 
\[
(-)_A\co
\Thin(\cGamma,\cA\with\cB)\to \Thin(\cGamma,\cA)
\]
the expected way. Similarly, we
obtain a functor $(-)_B\co \Thin(\cGamma,\cA\with\cB)\to \Thin(\cGamma,\cB)$.
\begin{proposition}
  \label{prop:pscpl-pscpr-other-form}
  The functors $(-)_A$ and $(-)_B$ are isomorphic to the functors $\pscpl \sthc
  (-)$ and $\pscpr \sthc (-)$ respectively.
\end{proposition}
\begin{proof}
  Using the \PCC, one can show the naturality of the family of
  isomorphisms of thin spans $\pscpl\sthc\sigma \cong \sigma_A$ described in the
  proof of \Cref{prop:sigmaa-sigmab-thin}, obtaining an isomorphism of
  functors. A similar argument holds for $(-)_B$ and $\pscpl\sthc(-)$.
\end{proof}

\begin{proposition}
  \label{prop:stcart-first-adjeq}
  Given thin groupoids $\cGamma,\cA,\cB$, there is an adjoint equivalence
  \[
    \begin{tikzcd}[csbo=7em]
      \Thin(\cGamma,\cA \with \cB)
      \ar[rr,bend left,"{((-)_A,(-)_B)}"]
      &
      \bot
      &
      \Thin(\cGamma,\cA)
      \times
      \Thin(\cGamma,\cB)
      \ar[ll,bend left,"\prodfact{-}{-}"]
    \end{tikzcd}
    \zbox.
  \]
\end{proposition}
\begin{proof}
  Given $\sigma \in\Thin(\cGamma,\cA \with \cB)$, write $\iota^\sigma_A \co
  \stsupp\sigma_A\hookrightarrow \stsupp\sigma$ and $\iota^\sigma_B \co
  \stsupp\sigma_B\hookrightarrow \stsupp\sigma$ for the canonical inclusions of
  $\Gpd$. There is a canonical $\stsupp\gamma_\sigma \co \stsupp\sigma
  \to \stsupp{\prodfact{\sigma_A}{\sigma_B}}$ defined as the inverse of
  \[
    \stsupp\sigma_A + \stsupp\sigma_B
    \xto{\coprodfact{\iota^\sigma_A}{\iota^\sigma_B}} 
    \stsupp\sigma
  \]
  and it induces a strong morphism of thin span $\gamma_\sigma \co \sigma \To
  \prodfact{\sigma_A}{\sigma_B} \in \Thin(\cGamma,\cA \with \cB)$. The naturality of
  $\gamma$ can be shown using the \PCC, so that we obtain a natural
  isomorphism
  \[
    \gamma \co
    \unit{\Thin(\cGamma,\cA \with \cB)}
    \To
    \prodfact{(-)_A}{(-)_B}
    \zbox.
  \]
  Given $(\rho,\tau) \in \Thin(\cGamma,\cA) \times \Thin(\cGamma,\cB) $, we have a
  canonical isomorphism
  $\stsupp{\delta}^A_{\rho,\tau}\co\stsupp{\prodfact{\rho}{\tau}}_A \to
  \stsupp\rho$ defined as the pullback isomorphism between
  \[
    \begin{tikzcd}
      \stsupp{\prodfact{\rho}{\tau}}_A
      \ar[r,"\iota^{\prodfact{\rho}{\tau}}_A"]
      \ar[d,dashed]
      \cphar[rd,near start,"\drcorner"]
      &
      \stsupp{\prodfact{\rho}{\tau}}
      \ar[d,"\stdisp{\prodfact{\rho}{\tau}}_{A+ B}"]
      \\
      A
      \ar[r,"\cpl"']
      &
      A+ B
    \end{tikzcd}
    \qand
    \begin{tikzcd}
      \stsupp{\rho}
      \ar[r,"\cpl"]
      \ar[d,"\stdisp\rho_A"']
      \cphar[rd,near start,"\drcorner"]
      &
      \stsupp{\prodfact{\rho}{\tau}}
      \ar[d,"\stdisp{\prodfact{\rho}{\tau}}_{A+ B}"]
      \\
      A
      \ar[r,"\cpl"']
      &
      A+ B
    \end{tikzcd}
    \zbox.
  \]
  It extends to a morphism $\delta^A_{\rho,\tau}\co\prodfact{\rho}{\tau}_A \to
  \rho\in \Thin(\cGamma,\cA)$. We define similarly a morphism
  $\delta^B_{\rho,\tau}\co\prodfact{\rho}{\tau}_B \to \tau\in \Thin(\cGamma,\cB)$.
  Writing $\delta_{\rho,\tau} =
  \prodfact{\delta^A_{\rho,\tau}}{\delta^B_{\rho,\tau}}$, the naturality of
  $\delta_{\rho,\tau}$ with respect to $\rho$ and $\tau$ can be checked using
  the \PCC, so that we obtain a natural isomorphism
  \begin{gather*}
    \delta
    \co 
    \catprodfact{\prodfact{(-)_{(1)}}{(-)_{(2)}}_A}{\prodfact{(-)_{(1)}}{(-)_{(2)}}_B}
    \\
    \To
    \\
    \unit{\Thin(\cGamma,\cA) \times \Thin(\cGamma,\cB) }
    \zbox.
  \end{gather*}

  We thus have an equivalence, and we verify that it is adjoint. We check the
  first zigzag equation, namely
  \begin{equation}
    \label{eq:stprod-adjeq:zigzag-one}
    (\delta\catprodfact{(-)_A}{(-)_B}) \circ (\catprodfact{(-)_A}{(-)_B}\gamma) = \unit{\catprodfact{(-)_A}{(-)_B}}
    \zbox.
  \end{equation}

  In order to verify the above equality, by symmetry, we just need to check its
  projection on $\Thin(\cGamma,\cA)$. So let $\sigma\in \Thin(\cGamma,\cA \with \cB)$.
  The component of the left-hand side of \eqref{eq:stprod-adjeq:zigzag-one} at
  $\sigma$ is then
  \[
    \sigma_A
    \xto{(\gamma_\sigma)_A}
    (\prodfact{\sigma_A}{\sigma_B})_A
    \xto{\delta^A_{\prodfact{\sigma_A}{\sigma_B}}}
    \sigma_A
    \zbox.
  \]
  By unfolding the definition of $\gamma$ and $\delta$, we compute that
  \[
    \stsupp\sigma_A
    \xto{\finv{(\stsupp\delta^A_{\prodfact{\sigma_A}{\sigma_B}})}}
    (\stsupp\sigma_A+ \stsupp\sigma_B)_A
    \xto{\finv{((\stsupp\gamma_\sigma)_A)}}
    \stsupp\sigma_A
    \xto{\iota^\sigma_A}
    \stsupp\sigma
  \]
  is precisely $\iota^\sigma_A$, which happens to be a monomorphism, so that
  $\finv{((\stsupp\gamma_\sigma)_A)} \circ
  \finv{(\stsupp\delta^A_{\prodfact{\sigma_A}{\sigma_B}})} =
  \unit{\stsupp\sigma_A}$, which is, up to inverses, what we wanted to show.
  Thus, the first zigzag equation holds.

  We now verify the second zigzag equation, namely
  \begin{equation}
    \label{eq:stprod-adjeq:zigzag-two}
    (\prodfact{-}{-}\delta) \circ (\gamma\prodfact{-}{-}) = \unit{\prodfact{-}{-}}
    \zbox.
  \end{equation}
  So let $(\rho,\tau) \in \Thin(\cGamma,\cA) \times \Thin(\cGamma,\cB) $.
  The component of the left-hand side of \eqref{eq:stprod-adjeq:zigzag-two} at
  $(\rho,\tau)$ is
  \[
    \prodfact{\rho}{\tau}
    \xto{\gamma_{\prodfact{\rho}{\tau}}}
    \prodfact {\prodfact{\rho}{\tau}_A} {\prodfact{\rho}{\tau}_B}
    \xto{\prodfact{\delta^A_{\rho,\tau}}{\delta^B_{\rho,\tau}}}
    \prodfact{\rho}{\tau}
  \]
  By unfolding the definition of $\gamma$ and $\delta$, we compute that
  \[
    \stsupp{\rho}
    \xto{\cpl}
    \stsupp{\prodfact{\rho}{\tau}}
    \xto{\finv{\stsupp{\prodfact{\delta^A_{\rho,\tau}}{\delta^B_{\rho,\tau}}}}}
    \stsupp{\prodfact {\prodfact{\rho}{\tau}_A} {\prodfact{\rho}{\tau}_B}}
    \xto{\finv{\stsupp{\gamma}_{\prodfact{\rho}{\tau}}}}
    \stsupp{\prodfact{\rho}{\tau}}
  \]
  reduces to
    $\stsupp{\rho}
    \xto{\cpl}
    \stsupp{\prodfact{\rho}{\tau}}$
  and similarly, 
  \[
    \stsupp{\tau}
    \xto{\cpr}
    \stsupp{\prodfact{\rho}{\tau}}
    \xto{\finv{\stsupp{\prodfact{\delta^A_{\rho,\tau}}{\delta^B_{\rho,\tau}}}}}
    \stsupp{\prodfact {\prodfact{\rho}{\tau}_A} {\prodfact{\rho}{\tau}_B}}
    \xto{\finv{\stsupp{\gamma}_{\prodfact{\rho}{\tau}}}}
    \stsupp{\prodfact{\rho}{\tau}}
  \]
  reduces to
    $\stsupp{\tau}
    \xto{\cpr}
    \stsupp{\prodfact{\rho}{\tau}}$
  so that, since $\cpl$ and $\cpr$ are jointly surjective,
  $\finv{\stsupp{\gamma}_{\prodfact{\rho}{\tau}}} \circ
  \finv{\stsupp{\prodfact{\delta^A_{\rho,\tau}}{\delta^B_{\rho,\tau}}}} =
  \unit{\stsupp{\prodfact{\rho}{\tau}}}$, which is, up to inverses, what we
  wanted. So the second zigzag holds.
\end{proof}

\begin{proposition}
  \label{prop:stcart-second-adjeq}
  Given thin groupoids $\cGamma,\cA,\cB$, there is an adjoint equivalence
  \[
    \begin{tikzcd}[csbo=7em]
      \Thin(\cGamma,\cA \with \cB)
      \ar[rr,bend left,"{(\pscpl\sthc (-),\pscpr\sthc (-))}"]
      &
      \bot
      &
      \Thin(\cGamma,\cA)
      \times
      \Thin(\cGamma,\cB)
      \ar[ll,bend left,"\prodfact{-}{-}"]
    \end{tikzcd}
    \zbox.
  \]
\end{proposition}
\begin{proof}
  This is a consequence of \Cref{prop:stcart-first-adjeq,prop:pscpl-pscpr-other-form}.
\end{proof}
Note that, given $(\rho,\tau) \in \Thin(\cGamma,\cA) \times \Thin(\cGamma,\cB)$, the
component at $(\rho,\tau)$ of the counit associated to the adjoint equivalence
of \Cref{prop:stcart-second-adjeq} is the composite
\[
  (\pscpl\sthc\prodfact\rho\tau,\pscpr\sthc\prodfact\rho\tau)
  \xto{\cong}
  (\prodfact\rho\tau_A,\prodfact\rho\tau_B)
  \xto{\delta}
  (\rho,\tau)
  \zbox.
\]
Now, we can conclude the proof of \Cref{prop:thinb-fp}:
\begin{proof}[Proof of \Cref{prop:thinb-fp}] We have the equalities
  $\Thinb(\Gamma,A \with B) = \Thin(\oc \Gamma,A\with B)$ and $ \Thinb(\Gamma,A)
  \times \Thinb(\Gamma,B) = \Thin(\oc\Gamma,A) \times \Thin(\oc\Gamma,B) $.
  Moreover, it is quite immediate that, up to these identifications, there is an
  isomorphism of functors $(\stbpl\sthcb (-),\stbpr\sthcb (-)) \cong
  (\pscpl\sthc (-),\pscpr\sthc (-))$. Thus, the unit/counit pair of the adjoint
  equivalence of \Cref{prop:stcart-second-adjeq} can be adjusted to get a
  unit/counit pair witnessing that we have an adjoint equivalence as in the
  statement.
\end{proof}
\endgroup

\subsection{The evaluation adjunction}
\label{sec:app-ev-adj}
\newcommand\cGamma{\Gamma}%
We give here some additional details for the proof of \Cref{prop:adj-closure},
stating the existence of an adjoint equivalence between the currying operation
and the evaluation one. This adjoint equivalence will be derived from the Seely
equivalence already introduced.

\subsubsection{Properties of the Seely equivalence}
\begin{proposition}
  \label{prop:seely-2nat}
  The family of functors $\seely_{A,B}$ for groupoids $A,B$ form a $2$\natural
  transformation
  \[
    \seely \co \oc (-) \times \oc (-) \To \oc((-) +  (-)) \co \Gpd \times \Gpd
    \to \Gpd
    \zbox.
  \]
\end{proposition}
\begin{proof}
  The naturality w.r.t $2$\cells is checked by direct point-wise computation.
\end{proof}

\begin{proposition}
  \label{prop:seely-bicart}
  The natural transformation $\seely$ is bicartesian.
\end{proposition}
\begin{proof}
  By a direct use of the point-wise characterization of pullbacks and
  \Cref{lem:carac_pb_bipb} on the naturality squares of $\seely$.
\end{proof}
\noindent Similarly, we have the same kind of properties for $\seelyinv$:
\begin{proposition}
  The family of functors $\seelyinv_{A,B}$ for groupoids $A,B$ form a $2$\natural
  transformation
  \[
    \seelyinv \co \oc((-) +  (-)) \To \oc (-) \times \oc (-) \co \Gpd \times \Gpd
    \to \Gpd
    \zbox.
  \]
\end{proposition}
\begin{proposition}
  \label{prop:seelyinv-bicart}
  The natural transformation $\seelyinv$ is bicartesian.
\end{proposition}
\subsubsection{The Seely coherence $2$-cell}
While the Seely isomorphisms of $1$\categorical models of linear logic are
required to satisfy an equality, in our $2$\categorical setting we only have the
following $2$\cell
\[
  \begin{tikzcd}[column sep={between origins,8em}]
    \ocd A\times \ocd B
    \ar[r,equals]
    \ar[d,"\seely_{\oc A,\oc B}"{description}]
    \cphar[rddd,"\xTo{\seelycoh_{A,B}}"]
    &
    \ocd A\times \ocd B
    \ar[d,"\mu_A\times\mu_B"{description}]
    \\
    \oc (\oc A  +  \oc B)
    \ar[d,"\oc\coprodfact{\oc(\cpl)}{\oc(\cpr)}"{description}]
    &
    \oc A\times \oc B
    \ar[dd,"\seely_{A,B}"{description}]
    \\
    \ocd (A  +  B)
    \ar[d,"\mu_{A  +  B}"{description}]
    &
    \\
    \oc(A +  B)
    \ar[r,equals]
    &
    \oc(A +  B)
  \end{tikzcd}
\]

We first compute the action of the two vertical morphisms on objects of $\ocd A
\times \ocd B$. So let $a = ((a_{i,j})_{j\in J^A_i})_{i \in I^A} \in \ocd A$ and
$b = ((b_{i,j})_{j\in J^B_i})_{i \in I^B} \in \ocd B$. The mappings associated
with the left vertical morphism are
\[
  \begin{split}
    &(a,b)
    \\
    \mapsto
    &
    (\cpl((a_{i,j})_{j \in J^A_i}))_{\ibij_l(i) \in \ibij_l(I^A)}
    \cup
    (\cpr((b_{i,j})_{j \in J^B_i}))_{\ibij_r(i) \in \ibij_r(I^B)}
    \\
    \mapsto
    &
    ((\cpl(a_{i,j}))_{j \in J^A_i})_{\ibij_l(i) \in \ibij_l(I^A)}
    \cup
    ((\cpr(b_{i,j}))_{j \in J^B_i})_{\ibij_r(i) \in \ibij_r(I^B)}
    \\
    \mapsto
    &
    (\cpl(a_{i,j}))_{\prodfact{\ibij_l(i)}{j} \in \sum_{i\in\ibij_l(I^A)}{J^A_{i}}}
    \cup
    (\cpr(b_{i,j}))_{\prodfact{\ibij_r(i)}{j} \in \sum_{i\in\ibij_r(I^B)}{J^B_{i}}}
  \end{split}
\]
and the mappings associated with the right are
\[
  \begin{split}
    (a,b)
    &
    \mapsto
    ((a_{i,j})_{\prodfact{i}{j} \in \sum_{i \in I^A}{J^A_{i}}},
    (b_{i,j})_{\prodfact{i}{j} \in \sum_{i\in I^B}{J^B_{i}}})
    \\
    &
    \mapsto
    \begin{multlined}[t]
      (\cpl(a_{i,j}))_{\ibij_l(\prodfact{i}{j}) \in \ibij_l(\sum_{i\in I^A}{J^A_{i}})}
      \\
      \cup
      (\cpl(b_{i,j}))_{\ibij_r(\prodfact{i}{j}) \in \ibij_r(\sum_{i\in I^B}{J^B_{i}})}
    \end{multlined}
  \end{split}
\]
where, given $D \in \Gpd$ and $(d_i)_{i \in I}, (d'_j)_{j\in J} \in \oc D$ with
$I \cap J = \emptyset$, we write $(d_i)_{i \in I}\cup (d'_j)_{j\in J}$ for the
evident family indexed by $I \cup J$.
We thus define $(\seelycoh_{A,B})_{a,b}$ as the map $(\pi,(\id)_{k \in K})$
where $\pi$ is the bijection $K \to K'$ with
\begin{gather*}
  K = \isum_{i \in \ibij_l(I^A)}{J^A_{i}} \cup \isum_{i\in\ibij_r(I^B)}{J^B_{i}}
  \shortintertext{and}
  K' = \ibij_l(\isum_{i\in I^A}{J^A_{i}}) \cup \ibij_r(\isum_{i\in I^B}{J^B_{i}})
\end{gather*}
such that $\pi$ maps $\iprod{\ibij_l(i)}{j} \in K$ to
$\ibij_l(\iprod{i}{j}) \in K'$, and $\iprod{\ibij_r(i)}{j} \in K$ to
$\ibij_r(\iprod{i}{j}) \in K'$. The naturality of $\seelycoh_{A,B}$ with
respect to morphisms $(a,b) \to (a',b')$ of $\ocd A\times \ocd B$ can be readily
checked. So $\seelycoh_{A,B}$ is indeed a $2$\cell of $\Gpd$. 

We moreover verify
that
\begin{proposition}
  The family of $2$\cells $\seelycoh_{A,B}$ for $A,B \in \Gpd$ is natural with
  respect to functors $F\co A \to A'$ and $G\co B \to B'$ of $\Gpd$. In other
  words, $\seelycoh = (\seelycoh_{A,B})_{A,B \in \Gpd}$ is a modification.
\end{proposition}
\begin{proof}
  A direct point-wise computation of naturality.
\end{proof}

\subsubsection{Properties of the evaluation span}
\noindent Given thin groupoids $\cA,\cB$, recall that we introduced a
span
\[
  \evm_{\cA,\cB} \co (\cA \cptto \cB) \with \cA \stto \cB
\]
defined by
\[
  \evm_{\cA,\cB}
  \qquad
  =
  \begin{tikzcd}[column sep={2em,between origins}]
    &&&\oc A \times B
    \ar[ld,"\catprodfactp{\pl,\pr,\pl}"']
    \ar[rrrddd,"\pr"]
    &&&
    \\
    &&
    \oc A \times B \times \oc A
    \ar[ld,"\eta_{\oc A \times B} \times \oc A"']
    &
    &
    \\
    &
    \oc (\oc A \times B) \times \oc A
    \ar[ld,"\seely_{\oc A \times B,A}"']
    \\
    \oc ((\oc A \times B) \with A)
    &&&&&&
    B
  \end{tikzcd}
  \zbox.
\]

We have that
\begin{proposition}
  \label{prop:ev-thin}
  We have $\evm_{\cA,\cB} \in \stsetb{\oc((\cA \cptto \cB) \with \cA) \linto \cB}$.
\end{proposition}
\begin{proof}
  By an adequate use of \Cref{prop:carac_lin,prop:carac_lin_thin}, using
  \Cref{prop:seely-bicart}, the bicartesianness of $\eta$ and
  \Cref{lem:rect-left-square-bipullback} to get the required bipullbacks.
\end{proof}
\subsubsection{The currying operation}
Recall that, given thin groupoids $\cGamma,\cA,\cB$ and $S \in \Thin_{\oc}(\Gamma
\with A, B)$\todo{si temps, remplacer A->Gamma, B->A, C->B}, we defined
$\Lambda(S)$ as the span
\[
  \Lambda(S)
  \qeq
  \begin{tikzcd}[column sep={5em,between origins}]
    &
    \stsupp{S}
    \ar[ld,"\pl\circ\seelyinv_{\Gamma,A}\circ\stdisp S_{\oc(\Gamma +  A)}"']
    \ar[rd,"\catprodfact{\pr\circ\seelyinv_{\Gamma,A}\circ\stdisp S_{\oc(\Gamma + 
        A)}}{\stdisp S_{B}}"]
    &
    \\
    \oc \Gamma
    &&
    \oc A \times B
  \end{tikzcd}
  \zbox.
\]
\begin{proposition}
  \label{prop:stcurry-thin}
  Given thin groupoids $\cGamma,\cA,\cB$ and $S \in\Thin(\cGamma\with\cA,\cB)$, we
  have $\stcurry S \in\stsetb{\oc\cGamma\linto(\cA\cptto\cB)}$.
\end{proposition}
\begin{proof}
  While more involved than the other instances, this still relies on
  \Cref{prop:carac_lin,prop:carac_lin_thin}, using
  \Cref{prop:seelyinv-bicart,lem:rect-left-square-bipullback} to get the
  required intermediate bipullbacks.
\end{proof}
The operation $\stcurry-$ can be extended to weak morphisms the expected way,
and it is compatible with the polarities, and we moreover have
\begin{proposition}
  Given thin groupoids $\cGamma,\cA,\cB$, the operation
  \[
    \stcurry- \co \Thinb(\cGamma\with \cA,\cB) \to \Thinb(\cGamma,\cA\cptto\cB)
  \]
  is functorial.
\end{proposition}

\subsubsection{The uncurrying operation}

We can define conversely an uncurrying operation. Given $S
\in\Thinb(\cGamma,\cA\cptto\cB)$, we define $\stuncurry S$ as the span
\[
  \stuncurry S
  \qeq
  \begin{tikzcd}[column sep={5em,between origins}]
    &
    \stsupp S
    \ar[ld,"{\seely_{\Gamma,A} \circ \catprodfact{\stdisp[\oc \Gamma] S}{\stdisp[\oc A] S}}"']
    \ar[rd,"{\stdisp[B] S}"]
    &
    \\
    \oc (\Gamma\with A)
    &&
    B
  \end{tikzcd}
  \zbox.
\]
As before, using similar methods, we can verify that
\begin{proposition}
  \label{prop:stuncurry-thin}
  Given thin groupoids $\cGamma,\cA,\cB$ and $\sigma \in\Thinb(\cGamma,\cA\cptto\cB)$, we
  have $\stuncurry\sigma \in\stsetb{\oc(\cGamma\with\cA)\linto\cB}$.
\end{proposition}
Moreover, we can extend this uncurrying operation to the weak morphisms the
expected way, and this operation is compatible with the polarities, and we
moreover have
\begin{proposition}
  Given thin groupoids $\cGamma,\cA,\cB$, the operation
  \[
    \stuncurry- \co \Thinb(\cGamma,\cA\cptto\cB) \to \Thinb(\cGamma\with\cA,\cB)
  \]
  is functorial.
\end{proposition}
\subsubsection{The adjoint equivalences}
We have a first adjoint equivalence between the currying and uncurrying
operations:
\begin{proposition}
  \label{prop:adj-curry-uncurry}
  There is an adjoint equivalence
  \[
    \begin{tikzcd}[sep={6em,between origins}]
      \Thinb(\cGamma,\cA \cptto \cB)
      \ar[rr,bend left,"\stuncurry-"]
      &
      \bot
      &
      \Thinb(\cGamma \with \cA,\cB)
      \ar[ll,bend left,"\stcurry-"]
    \end{tikzcd}
    \zbox.
  \]
\end{proposition}
\begin{proof}
  \newcommand\stuncurcur[1]{#1}%
  \newcommand\seelyu{\Sigma}%
  \newcommand\seelycu{\bar\Sigma}%
  The application of $\stuncurry-$ and $\stcurry-$ to spans essentially amounts
  to adequately postcompose the display maps of these spans by $\seely_{\Gamma,A}$
  and $\seelyinv_{\Gamma,A}$. The unit/counit pair witnessing the adjoint equivalence
  are then easily derived from a unit/counit pair $(\seelyu_{\Gamma,A},\seelycu_{\Gamma,A})$ witnessing the adjoint
  equivalence $\seely_{\Gamma,A} \dashv \seelyinv_{\Gamma,A}$. For example, given $S \in
  \Thinb(\cGamma,\cA \cptto \cB)$, the component of the unit of
  $\stuncurry-\dashv\stcurry-$ at $S$ is the weak morphism
  $(\id_S,\phi)$\simon{mauvaise notation ici} with
  \[
    \phi^{\oc \Gamma}
    =
    \begin{tikzcd}
      \stsupp{S}
      \ar[r,equals]
      \ar[dddd,"\stdisp{S}_{\oc \Gamma}"']
      \cphar[rdddd,"="]
      &
      \stsupp{S}
      \ar[d,"\catprodfact{\stdisp{S}_{\oc \Gamma}}{\stdisp{S}_{\oc A}}"{description}]
      \ar[r,equals]
      \cphar[rd,"="]
      &
      \stsupp{\stuncurcur{S}}
      \ar[d,"\catprodfact{\stdisp{S}_{\oc \Gamma}}{\stdisp{S}_{\oc A}}"{description}]
      \\
      &
      \oc \Gamma \times \oc A
      \ar[dd,equals]
      \ar[r,equals]
      \cphar[rdd,"\xTo{\seelyu_{\Gamma,A}}"]
      &
      \oc \Gamma \times \oc A
      \ar[d,"\seely_{\Gamma,A}"]
      \\
      &&
      \oc(\Gamma +  A)
      \ar[d,"\seelyinv_{\Gamma,A}"]
      \\
      &
      \oc \Gamma \times \oc A
      \ar[d,"\pl"]
      \ar[r,equals]
      \cphar[rd,"="]
      &
      \oc \Gamma \times \oc A
      \ar[d,"\pl"]
      \\
      \oc \Gamma
      \ar[r,equals]
      &
      \oc \Gamma
      \ar[r,equals]
      &
      \oc \Gamma
    \end{tikzcd}
  \]
  and
  \[
    \phi^{\oc A \times B}
    =
    \begin{tikzcd}
      \stsupp{S}
      \ar[r,equals]
      \ar[d,"\stdisp{S}_{\oc A\times B}"']
      \cphar[rd,"\xTo{\catprodfact{\phi^{\oc A}}{\phi^B}}"]
      &
      \stsupp{\stuncurcur{S}}
      \ar[d,"\stdisp{\stuncurcur{S}}_{\oc A\times B}"]
      \\
      \oc A \times B
      \ar[r,equals]
      &
      \oc A \times B
    \end{tikzcd}
  \]
  where
  \[
    \phi^{\oc A}
    =
    \begin{tikzcd}
      \stsupp{S}
      \ar[r,equals]
      \ar[dddd,"\stdisp{S}_{\oc A}"']
      \cphar[rdddd,"="]
      &
      \stsupp{S}
      \ar[d,"\catprodfact{\stdisp{S}_{\oc \Gamma}}{\stdisp{S}_{\oc A}}"{description}]
      \ar[r,equals]
      \cphar[rd,"="]
      &
      \stsupp{\stuncurcur{S}}
      \ar[d,"\catprodfact{\stdisp{S}_{\oc \Gamma}}{\stdisp{S}_{\oc A}}"{description}]
      \\
      &
      \oc \Gamma \times \oc A
      \ar[dd,equals]
      \ar[r,equals]
      \cphar[rdd,"\xTo{\seelyu_{\Gamma,A}}"]
      &
      \oc \Gamma \times \oc A
      \ar[d,"\seely_{\Gamma,A}"]
      \\
      &&
      \oc(\Gamma +  A)
      \ar[d,"\seelyinv_{\Gamma,A}"]
      \\
      &
      \oc \Gamma \times \oc A
      \ar[d,"\pr"]
      \ar[r,equals]
      \cphar[rd,"="]
      &
      \oc \Gamma \times \oc A
      \ar[d,"\pr"]
      \\
      \oc A
      \ar[r,equals]
      &
      \oc A
      \ar[r,equals]
      &
      \oc A
    \end{tikzcd}
  \]
  and
  \[
    \phi^{B}
    =
    \begin{tikzcd}
      \stsupp{S}
      \ar[r,equals]
      \ar[d,"\stdisp{S}_{B}"']
      \cphar[rd,"="]
      &
      \stsupp{\stuncurcur{S}}
      \ar[d,"\stdisp{\stuncurcur{S}}_{B}"]
      \\
      B
      \ar[r,equals]
      &
      B
    \end{tikzcd}
  \]
  where we write $\stdisp{S}_A$ and $\stdisp{S}_B$ for $\pl\circ \stdisp{S}_{\oc
    A\times B}$ and $\pr\circ \stdisp{S}_{\oc A\times B}$ respectively (and
  similarly for $\stdisp{\stuncurcur{S}}_A$ and $\stdisp{\stuncurcur{S}}_B$).
  The counit is defined similarly. The fact that this unit/counit pair satisfies
  the zigzag equations is a consequence of the fact that
  $(\seelyu_{\Gamma,A},\seelycu_{\Gamma,A})$ satisfies the same equations, by vertical
  pasting of $2$\cells.
\end{proof}
\simon{bon endroit pour introduire une sous-sous-section}%
\noindent In order to get another adjoint equivalence, we prove that the
uncurrying functor is isomorphic to the uncurrying-through-evaluation operation:
\begin{proposition}
  \label{prop:evaluncurry-isom-stuncurry}
  The functor $\evm \sthcb (- \with A) \co \Thinb(\cGamma,\cA \cptto \cB) \to\Thinb(\cGamma
  \with \cA,\cB)$ is isomorphic to the uncurrying functor $\stuncurry-$.
\end{proposition}
\begin{proof}
  Let ${S} \in \Thinb(\cGamma,\cA \cptto \cB)$. We compute $\evm \sthcb ({S}
  \with A)$. It is the composition of the spans
  \[
    \begin{tikzcd}[column sep={between origins,2em}]
      &&&
      \oc (\stsupp{S} +  A)
      \ar[ld,"\oc(\stdisp{S}_{\oc \Gamma}  +  \eta_A)"']
      \ar[rrrddd,"\oc (\stdisp{S}_{\oc A\times B} +  A)"]
      &&&
      \\
      &&
      \oc(\oc \Gamma  +  \oc A)
      \ar[ld,"\oc\coprodfact{\oc(\cpl)}{\oc(\cpr)}"']
      &&&&
      &
      \\
      &
      \ocd(\Gamma +  A)
      \ar[ld,"\mu_{\Gamma +  A}"']
      &&&&
      \\
      \oc(\Gamma +  A)
      &&&&&&
      \oc((\oc A \times B)  +  A)
    \end{tikzcd}
  \]
  and
  \[
    \begin{tikzcd}[column sep={2em,between origins}]
      &&&\oc A \times B
      \ar[ld,"\catprodfactp{\pl,\pr,\pl}"']
      \ar[rrrddd,"\pr"]
      &&&
      \\
      &&
      \oc A \times B \times \oc A
      \ar[ld,"\eta_{\oc A \times B} \times \oc A"']
      &
      &
      \\
      &
      \oc (\oc A \times B) \times \oc A
      \ar[ld,"\seely_{\oc A \times B,A}"']
      \\
      \oc ((\oc A \times B)  +  A)
      &&&&&&
      B
    \end{tikzcd}
    \zbox.
  \]
  We compute the inner pullback of this composition as the rectangle of
  pullbacks shown in \Cref{fig:inner-pbs}.
  \begin{figure*}[h!]
    \centering
    \[
      \begin{tikzcd}[column sep={between origins,11em}]
        \oc (\stsupp{S} +  A)
        \ar[d,"\oc (\stdisp{S}_{\oc A\times B} +  A)"']
        &
        \oc \stsupp{S}\times \oc A
        \ar[d,"\oc(\stdisp{S}_{\oc A\times B}) \times \oc A"{description}]
        \ar[l,"\seely_{\stsupp{S},A}"']
        \cphar[ld,"\dlcorner",near start]
        &
        \stsupp{S}\times \oc A
        \ar[d,"\stdisp{S}_{\oc A\times B} \times \oc A"{description}]
        \ar[l,"\eta_{\stsupp{S}} \times \oc A"']
        \cphar[ld,"\dlcorner",near start]
        &
        \stsupp{S}
        \ar[d,"\stdisp{S}_{\oc A\times B}"]
        \ar[l,"\catprodfact{\stsupp{S}}{\stdisp{S}_{\oc A}}"']
        \cphar[ld,"\dlcorner",near start]
        \\
        \oc((\oc A \times B)  +  A)
        &
        \oc (\oc A \times B) \times \oc A
        \ar[l,"\seely_{\oc A \times B,A}"]
        &
        \oc A \times B \times \oc A
        \ar[l,"\eta_{\oc A \times B} \times \oc A"]
        &
        \oc A \times B
        \ar[l,"\catprodfactp{\pl,\pr,\pl}"]
      \end{tikzcd}
      \zbox.
    \]
    \caption{The inner rectangle of the composition of the two spans}
    \label{fig:inner-pbs}
  \end{figure*}
  Thus, up to a canonical isomorphism of pullbacks, $\evm \sthcb ({S} \with A)$
  is ${\bar{S}} \in \Thinb(\cGamma \with \cA,\cB)$ with $S$ as the support of
  ${\bar{S}}$ and
  \begin{gather*}
    \stdisp{\bar{S}}_{\oc(\Gamma +  A)}
    =
    \\
    \stsupp{S}
    \xto{\prodfact{\stsupp{S}}{\stdisp{S}_{\oc A}}}
    \stsupp{S}\times \oc A
    \xto{\eta_{\stsupp{S}} \times \oc A}
    \oc \stsupp{S} \times \oc A
    \xto{\seely_{\stsupp{S},A}}
    \oc (\stsupp{S}  +  A)
    \\
    \xto{\oc(\stdisp{S}_{\oc \Gamma}  +  \eta_A)}
    \oc(\oc \Gamma  +  \oc A)
    \xto{\oc\coprodfact{\oc(\cpl)}{\oc(\cpr)}}
    \ocd(\Gamma  +  A)
    \xto{\mu_{\Gamma +  A}}
    \oc(\Gamma  +  A)
  \end{gather*}
  and
  \[
    \stdisp{\bar{S}}_{B}
    \qqeq
    \stsupp{S}
    \xto{\stdisp{S}_{\oc A\times B}}
    \oc A \times B
    \xto{\pr}
    B
    \zbox.
  \]
  The operation $S \mapsto \bar S$ can be shown to extend naturally to weak
  morphisms, so that we obtain a functor $\bar{(-)} \co \Thinb(\cGamma,\cA \cptto
  \cB) \to \Thinb(\cGamma \with \cA,\cB)$ which is naturally isomorphic to $\evm
  \sthcb (- \with A)$. So we are left to show that $\bar {(-)}\cong
  \stuncurry-$. For this purpose, for ${S} \in \Thinb(\cGamma,\cA \cptto \cB)$, we
  define an isomorphism $\theta_{S} = (\stsuppreally
  \theta_{S},\theta^{\oc(\Gamma+A)}_{S},\theta_{S}^B) \co \bar{S} \To \stuncurry{S}$.
  We take $\stsuppreally \theta_{S} = \unit{\stsupp{S}}$ and $\theta_{S}^B =
  \unit{\stdisp{S}_B}$,
  and define $\theta_{S}^{\oc(\Gamma + A)}$ as the (vertically expressed) $2$\cell of
  \Cref{fig:closure-phi-def}.
  \begin{figure*}
    \centering
    \[
      \begin{tikzcd}[column sep={between origins,8em}]
        \stsupp{S}
        \ar[r,"\gpdprodfact{\stdisp{S}_{\oc \Gamma}}{\stdisp{S}_{\oc A}}"]
        \ar[d,equals]
        \cphar[rrd,"="]
        &
        \oc \Gamma \times \oc A
        \ar[r,"\eta_{\oc \Gamma} \times \oc \eta_{A}"]
        &
        \ocd \Gamma \times \ocd A 
        \ar[r,"\seely_{\oc \Gamma,\oc A}"]
        \ar[d,equals]
        \cphar[rrrd,"\Downarrow \seelycoh_{\Gamma,A}"]
        &
        \oc (\oc \Gamma  +  \oc A)
        \ar[r,"\oc\coprodfact{\oc(\cpl)}{\oc(\cpr)}"]
        &
        \ocd (\Gamma  +  A)
        \ar[r,"\mu_{\Gamma +  A}"]
        &
        \oc(\Gamma  +  A)
        \ar[d,equals]
        \\
        \stsupp{S}
        \ar[r,"\gpdprodfact{\stdisp{S}_{\oc \Gamma}}{\stdisp{S}_{\oc A}}"{description}]
        \ar[d,equals]
        \cphar[rd,"="]
        &
        \oc \Gamma \times \oc A
        \ar[r,"\eta_{\oc \Gamma} \times \oc \eta_{A}"{description}]
        \ar[d,equals]
        \cphar[rrd,"\Downarrow\oclu_\Gamma \times \ocru_A"]
        &
        \ocd \Gamma \times \ocd A 
        \ar[r,"\mu_\Gamma \times \mu_A"{description}]
        &
        \oc \Gamma\times \oc A
        \ar[rr,"\seely_{\Gamma,A}"{description}]
        \ar[d,equals]
        \cphar[rrd,"="]
        &
        &
        \oc(\Gamma  +  A)
        \ar[d,equals]
        \\
        \stsupp{S}
        \ar[r,"\gpdprodfact{\stdisp{S}_{\oc \Gamma}}{\stdisp{S}_{\oc A}}"']
        &
        \oc \Gamma \times \oc A
        \ar[rr,equals]
        &
        &
        \oc \Gamma\times \oc A
        \ar[rr,"\seely_{\Gamma,A}"']
        &
        &
        \oc(\Gamma  +  A)
      \end{tikzcd}
    \]
    \caption{The $\theta_{S}^{\oc(\Gamma +  A)}$ $2$-cell}
    \label{fig:closure-phi-def}
    \scriptsize Recall the definition of $\oclu$ and $\ocru$ from \Cref{fig:unitlaw}.
  \end{figure*}
  We directly observe that $\theta_{S}^{\oc (\Gamma + A)}$ and $\theta_{S}^B$ have
  the adequate polarities, so that $\theta_{S} \in \Thinb(\cGamma \with \cA,\cB)$
  and it is an isomorphism. The naturality of $\theta_{S}$ \wrt $S$ can be
  checked diagrammatically, by pasting. Thus, $\theta$ defines an isomorphism
  $\bar {(-)}\cong \stuncurry-$, so that we have
  \[
    \evm \sthcb ((-) \with A) 
    \xTo{\cong}
    \bar {(-)}
    \xTo{\theta}
    \stuncurry-
    \zbox.
    \tag*{\qedhere}
  \]
\end{proof}
\noindent We can now conclude the proof of \Cref{prop:adj-closure}:
\begin{proof}[Proof of \Cref{prop:adj-closure}] %
  By \Cref{prop:adj-curry-uncurry}, we have an adjoint equivalence between the
  currying and uncurrying operations. By \Cref{prop:evaluncurry-isom-stuncurry},
  we can replace the uncurrying operation by $\evm \sthcb ((-) \with A)$: by
  adjusting the unit and counit the expected way, we keep the adjoint
  equivalence.
\end{proof}


%% file: main.bbl
\begin{thebibliography}{10}
\providecommand{\url}[1]{#1}
\csname url@samestyle\endcsname
\providecommand{\newblock}{\relax}
\providecommand{\bibinfo}[2]{#2}
\providecommand{\BIBentrySTDinterwordspacing}{\spaceskip=0pt\relax}
\providecommand{\BIBentryALTinterwordstretchfactor}{4}
\providecommand{\BIBentryALTinterwordspacing}{\spaceskip=\fontdimen2\font plus
\BIBentryALTinterwordstretchfactor\fontdimen3\font minus
  \fontdimen4\font\relax}
\providecommand{\BIBforeignlanguage}[2]{{%
\expandafter\ifx\csname l@#1\endcsname\relax
\typeout{** WARNING: IEEEtran.bst: No hyphenation pattern has been}%
\typeout{** loaded for the language `#1'. Using the pattern for}%
\typeout{** the default language instead.}%
\else
\language=\csname l@#1\endcsname
\fi
#2}}
\providecommand{\BIBdecl}{\relax}
\BIBdecl

\bibitem{DBLP:journals/tcs/Girard87}
\BIBentryALTinterwordspacing
J.~Girard, ``Linear logic,'' \emph{Theor. Comput. Sci.}, vol.~50, pp. 1--102,
  1987. [Online]. Available: \url{https://doi.org/10.1016/0304-3975(87)90045-4}
\BIBentrySTDinterwordspacing

\bibitem{DBLP:journals/mscs/Ehrhard05}
\BIBentryALTinterwordspacing
T.~Ehrhard, ``Finiteness spaces,'' \emph{Math. Struct. Comput. Sci.}, vol.~15,
  no.~4, pp. 615--646, 2005. [Online]. Available:
  \url{https://doi.org/10.1017/S0960129504004645}
\BIBentrySTDinterwordspacing

\bibitem{DBLP:journals/tcs/CarvalhoPF11}
\BIBentryALTinterwordspacing
D.~de~Carvalho, M.~Pagani, and L.~T. de~Falco, ``A semantic measure of the
  execution time in linear logic,'' \emph{Theor. Comput. Sci.}, vol. 412,
  no.~20, pp. 1884--1902, 2011. [Online]. Available:
  \url{https://doi.org/10.1016/j.tcs.2010.12.017}
\BIBentrySTDinterwordspacing

\bibitem{DBLP:journals/iandc/DanosE11}
\BIBentryALTinterwordspacing
V.~Danos and T.~Ehrhard, ``Probabilistic coherence spaces as a model of
  higher-order probabilistic computation,'' \emph{Inf. Comput.}, vol. 209,
  no.~6, pp. 966--991, 2011. [Online]. Available:
  \url{https://doi.org/10.1016/j.ic.2011.02.001}
\BIBentrySTDinterwordspacing

\bibitem{DBLP:conf/lics/LairdMMP13}
\BIBentryALTinterwordspacing
J.~Laird, G.~Manzonetto, G.~McCusker, and M.~Pagani, ``Weighted relational
  models of typed lambda-calculi,'' in \emph{28th Annual {ACM/IEEE} Symposium
  on Logic in Computer Science, {LICS} 2013, New Orleans, LA, USA, June 25-28,
  2013}.\hskip 1em plus 0.5em minus 0.4em\relax {IEEE} Computer Society, 2013,
  pp. 301--310. [Online]. Available: \url{https://doi.org/10.1109/LICS.2013.36}
\BIBentrySTDinterwordspacing

\bibitem{DBLP:conf/popl/PaganiSV14}
\BIBentryALTinterwordspacing
M.~Pagani, P.~Selinger, and B.~Valiron, ``Applying quantitative semantics to
  higher-order quantum computing,'' in \emph{The 41st Annual {ACM}
  {SIGPLAN-SIGACT} Symposium on Principles of Programming Languages, {POPL}
  '14, San Diego, CA, USA, January 20-21, 2014}, S.~Jagannathan and P.~Sewell,
  Eds.\hskip 1em plus 0.5em minus 0.4em\relax {ACM}, 2014, pp. 647--658.
  [Online]. Available: \url{https://doi.org/10.1145/2535838.2535879}
\BIBentrySTDinterwordspacing

\bibitem{DBLP:journals/jacm/EhrhardPT18}
\BIBentryALTinterwordspacing
T.~Ehrhard, M.~Pagani, and C.~Tasson, ``Full abstraction for probabilistic
  {PCF},'' \emph{J. {ACM}}, vol.~65, no.~4, pp. 23:1--23:44, 2018. [Online].
  Available: \url{https://doi.org/10.1145/3164540}
\BIBentrySTDinterwordspacing

\bibitem{DBLP:journals/pacmpl/ClairambaultV20}
\BIBentryALTinterwordspacing
P.~Clairambault and M.~de~Visme, ``Full abstraction for the quantum
  lambda-calculus,'' \emph{Proc. {ACM} Program. Lang.}, vol.~4, no. {POPL}, pp.
  63:1--63:28, 2020. [Online]. Available: \url{https://doi.org/10.1145/3371131}
\BIBentrySTDinterwordspacing

\bibitem{DBLP:journals/corr/abs-0905-4251}
\BIBentryALTinterwordspacing
D.~de~Carvalho, ``Execution time of lambda-terms via denotational semantics and
  intersection types,'' \emph{CoRR}, vol. abs/0905.4251, 2009. [Online].
  Available: \url{http://arxiv.org/abs/0905.4251}
\BIBentrySTDinterwordspacing

\bibitem{DBLP:journals/jfp/AccattoliGK20}
\BIBentryALTinterwordspacing
B.~Accattoli, S.~Graham{-}Lengrand, and D.~Kesner, ``Tight typings and split
  bounds, fully developed,'' \emph{J. Funct. Program.}, vol.~30, p. e14, 2020.
  [Online]. Available: \url{https://doi.org/10.1017/S095679682000012X}
\BIBentrySTDinterwordspacing

\bibitem{DBLP:journals/lmcs/BucciarelliKR18}
\BIBentryALTinterwordspacing
A.~Bucciarelli, D.~Kesner, and S.~R.~D. Rocca, ``Inhabitation for
  non-idempotent intersection types,'' \emph{Log. Methods Comput. Sci.},
  vol.~14, no.~3, 2018. [Online]. Available:
  \url{https://doi.org/10.23638/LMCS-14(3:7)2018}
\BIBentrySTDinterwordspacing

\bibitem{DBLP:journals/iandc/AbramskyJM00}
S.~Abramsky, R.~Jagadeesan, and P.~Malacaria, ``Full abstraction for {PCF},''
  \emph{Inf. Comput.}, vol. 163, no.~2, pp. 409--470, 2000.

\bibitem{DBLP:journals/iandc/HylandO00}
\BIBentryALTinterwordspacing
J.~M.~E. Hyland and C.~L. Ong, ``On full abstraction for {PCF:} i, ii, and
  {III},'' \emph{Inf. Comput.}, vol. 163, no.~2, pp. 285--408, 2000. [Online].
  Available: \url{https://doi.org/10.1006/inco.2000.2917}
\BIBentrySTDinterwordspacing

\bibitem{DBLP:conf/csl/BaillotDER97}
\BIBentryALTinterwordspacing
P.~Baillot, V.~Danos, T.~Ehrhard, and L.~Regnier, ``Timeless games,'' in
  \emph{Computer Science Logic, 11th International Workshop, {CSL} '97, Annual
  Conference of the EACSL, Aarhus, Denmark, August 23-29, 1997, Selected
  Papers}, ser. Lecture Notes in Computer Science, M.~Nielsen and W.~Thomas,
  Eds., vol. 1414.\hskip 1em plus 0.5em minus 0.4em\relax Springer, 1997, pp.
  56--77. [Online]. Available: \url{https://doi.org/10.1007/BFb0028007}
\BIBentrySTDinterwordspacing

\bibitem{DBLP:journals/tcs/Mellies06}
P.~Melli{\`{e}}s, ``Asynchronous games 2: The true concurrency of innocence,''
  \emph{Theor. Comput. Sci.}, vol. 358, no. 2-3, pp. 200--228, 2006.

\bibitem{DBLP:conf/tlca/Boudes09}
\BIBentryALTinterwordspacing
P.~Boudes, ``Thick subtrees, games and experiments,'' in \emph{Typed Lambda
  Calculi and Applications, 9th International Conference, {TLCA} 2009,
  Brasilia, Brazil, July 1-3, 2009. Proceedings}, ser. Lecture Notes in
  Computer Science, P.~Curien, Ed., vol. 5608.\hskip 1em plus 0.5em minus
  0.4em\relax Springer, 2009, pp. 65--79. [Online]. Available:
  \url{https://doi.org/10.1007/978-3-642-02273-9\_7}
\BIBentrySTDinterwordspacing

\bibitem{leinster2004higher}
T.~Leinster, \emph{Higher operads, higher categories}.\hskip 1em plus 0.5em
  minus 0.4em\relax Cambridge University Press, 2004, no. 298.

\bibitem{DBLP:journals/mscs/FioreS21}
\BIBentryALTinterwordspacing
M.~Fiore and P.~Saville, ``Coherence for bicategorical cartesian closed
  structure,'' \emph{Math. Struct. Comput. Sci.}, vol.~31, no.~7, pp. 822--849,
  2021. [Online]. Available: \url{https://doi.org/10.1017/S0960129521000281}
\BIBentrySTDinterwordspacing

\bibitem{DBLP:conf/lics/FioreS19}
\BIBentryALTinterwordspacing
------, ``A type theory for cartesian closed bicategories (extended
  abstract),'' in \emph{34th Annual {ACM/IEEE} Symposium on Logic in Computer
  Science, {LICS} 2019, Vancouver, BC, Canada, June 24-27, 2019}.\hskip 1em
  plus 0.5em minus 0.4em\relax {IEEE}, 2019, pp. 1--13. [Online]. Available:
  \url{https://doi.org/10.1109/LICS.2019.8785708}
\BIBentrySTDinterwordspacing

\bibitem{fiore2008cartesian}
M.~Fiore, N.~Gambino, M.~Hyland, and G.~Winskel, ``The cartesian closed
  bicategory of generalised species of structures,'' \emph{Journal of the
  London Mathematical Society}, vol.~77, no.~1, pp. 203--220, 2008.

\bibitem{DBLP:journals/lmcs/CastellanCW19}
S.~Castellan, P.~Clairambault, and G.~Winskel, ``Thin games with symmetry and
  concurrent hyland-ong games,'' \emph{Log. Methods Comput. Sci.}, vol.~15,
  no.~1, 2019.

\bibitem{paquet2020probabilistic}
H.~Paquet, ``Probabilistic concurrent game semantics,'' Ph.D. dissertation,
  2020.

\bibitem{DBLP:conf/lics/Mellies19}
P.~Melli{\`{e}}s, ``Template games and differential linear logic,'' in
  \emph{{LICS}}.\hskip 1em plus 0.5em minus 0.4em\relax {IEEE}, 2019, pp.
  1--13.

\bibitem{DBLP:conf/lics/TsukadaAO18}
T.~Tsukada, K.~Asada, and C.~L. Ong, ``Species, profunctors and taylor
  expansion weighted by {SMCC:} {A} unified framework for modelling
  nondeterministic, probabilistic and quantum programs,'' in
  \emph{{LICS}}.\hskip 1em plus 0.5em minus 0.4em\relax {ACM}, 2018, pp.
  889--898.

\bibitem{DBLP:conf/fscd/Galal20}
\BIBentryALTinterwordspacing
Z.~Galal, ``A profunctorial scott semantics,'' in \emph{5th International
  Conference on Formal Structures for Computation and Deduction, {FSCD} 2020,
  June 29-July 6, 2020, Paris, France (Virtual Conference)}, ser. LIPIcs, Z.~M.
  Ariola, Ed., vol. 167.\hskip 1em plus 0.5em minus 0.4em\relax Schloss
  Dagstuhl - Leibniz-Zentrum f{\"{u}}r Informatik, 2020, pp. 16:1--16:18.
  [Online]. Available: \url{https://doi.org/10.4230/LIPIcs.FSCD.2020.16}
\BIBentrySTDinterwordspacing

\bibitem{DBLP:conf/fscd/Galal21}
\BIBentryALTinterwordspacing
------, ``A bicategorical model for finite nondeterminism,'' in \emph{6th
  International Conference on Formal Structures for Computation and Deduction,
  {FSCD} 2021, July 17-24, 2021, Buenos Aires, Argentina (Virtual Conference)},
  ser. LIPIcs, N.~Kobayashi, Ed., vol. 195.\hskip 1em plus 0.5em minus
  0.4em\relax Schloss Dagstuhl - Leibniz-Zentrum f{\"{u}}r Informatik, 2021,
  pp. 10:1--10:17. [Online]. Available:
  \url{https://doi.org/10.4230/LIPIcs.FSCD.2021.10}
\BIBentrySTDinterwordspacing

\bibitem{DBLP:conf/lics/TsukadaAO17}
\BIBentryALTinterwordspacing
T.~Tsukada, K.~Asada, and C.~L. Ong, ``Generalised species of rigid resource
  terms,'' in \emph{32nd Annual {ACM/IEEE} Symposium on Logic in Computer
  Science, {LICS} 2017, Reykjavik, Iceland, June 20-23, 2017}.\hskip 1em plus
  0.5em minus 0.4em\relax {IEEE} Computer Society, 2017, pp. 1--12. [Online].
  Available: \url{https://doi.org/10.1109/LICS.2017.8005093}
\BIBentrySTDinterwordspacing

\bibitem{DBLP:conf/lics/Olimpieri21}
F.~Olimpieri, ``Intersection type distributors,'' in \emph{{LICS}}.\hskip 1em
  plus 0.5em minus 0.4em\relax {IEEE}, 2021, pp. 1--15.

\bibitem{relevant23}
A.~Kerinec, G.~Manzonetto, and F.~Olimpieri, ``Why are proofs relevant in
  proof-relevant models?'' 2023, accepted in POPL 2023.

\bibitem{kelly1980coherence}
G.~M. Kelly and M.~L. Laplaza, ``Coherence for compact closed categories,''
  \emph{Journal of pure and applied algebra}, vol.~19, pp. 193--213, 1980.

\bibitem{DBLP:journals/tcs/Lamarche92}
\BIBentryALTinterwordspacing
F.~Lamarche, ``Quantitative domains and infinitary algebras,'' \emph{Theor.
  Comput. Sci.}, vol.~94, no.~1, pp. 37--62, 1992. [Online]. Available:
  \url{https://doi.org/10.1016/0304-3975(92)90323-8}
\BIBentrySTDinterwordspacing

\bibitem{stay2013compact}
M.~Stay, ``Compact closed bicategories,'' \emph{arXiv preprint
  arXiv:1301.1053}, 2013.

\bibitem{hoffnung2011spans}
A.~E. Hoffnung, ``Spans in 2-categories: A monoidal tricategory,'' \emph{arXiv
  preprint arXiv:1112.0560}, 2011.

\bibitem{DBLP:journals/pacmpl/Mellies19}
P.~Melli{\`{e}}s, ``Categorical combinatorics of scheduling and synchronization
  in game semantics,'' \emph{Proc. {ACM} Program. Lang.}, vol.~3, no. {POPL},
  pp. 23:1--23:30, 2019.

\bibitem{DBLP:conf/lics/BaillotDE97}
P.~Baillot, V.~Danos, T.~Ehrhard, and L.~Regnier, ``Believe it or not, ajm's
  games model is a model of classical linear logic,'' in \emph{{LICS}}.\hskip
  1em plus 0.5em minus 0.4em\relax {IEEE} Computer Society, 1997, pp. 68--75.

\bibitem{DBLP:conf/csl/CastellanCW14}
S.~Castellan, P.~Clairambault, and G.~Winskel, ``Symmetry in concurrent
  games,'' in \emph{{CSL-LICS}}.\hskip 1em plus 0.5em minus 0.4em\relax {ACM},
  2014, pp. 28:1--28:10.

\bibitem{DBLP:journals/corr/abs-2107-03155}
P.~Clairambault and H.~Paquet, ``The quantitative collapse of concurrent games
  with symmetry,'' \emph{CoRR}, vol. abs/2107.03155, 2021.

\bibitem{DBLP:conf/lics/Vial17}
P.~Vial, ``Infinitary intersection types as sequences: {A} new answer to klop's
  problem,'' in \emph{{LICS}}.\hskip 1em plus 0.5em minus 0.4em\relax {IEEE}
  Computer Society, 2017, pp. 1--12.

\bibitem{DBLP:journals/tcs/Ehrhard12}
T.~Ehrhard, ``The scott model of linear logic is the extensional collapse of
  its relational model,'' \emph{Theor. Comput. Sci.}, vol. 424, pp. 20--45,
  2012.

\bibitem{cheng2003pseudo}
E.~Cheng, M.~Hyland, and J.~Power, ``Pseudo-distributive laws,''
  \emph{Electronic Notes in Theoretical Computer Science}, vol.~83, pp.
  227--245, 2003.

\end{thebibliography}
